\documentclass[11pt]{article}

\input{conf}

\title{An Algorithm-to-Contract Framework without Demand Queries}

\author{Ilan Doron-Arad\thanks{Math Department, MIT, Cambridge, Massachusetts, USA. \texttt{ilanda@mit.edu}} \and
Hadas Shachnai\thanks{Computer Science Department, Technion, Haifa, Israel. \texttt{hadas@cs.technion.ac.il}} \and 
Gilad Shmerler\thanks{Computer Science Department, Technion, Haifa, Israel. \texttt{shmerler@campus.technion.ac.il}} \and
Inbal Talgam-Cohen\thanks{School of Computer Science, Tel Aviv University. Email: \texttt{inbaltalgam@gmail.com}}
}
\date{}

\begin{document}

\newcommand{\algoPhaseOne}{
\begin{algorithm}[hbt]
\caption{LP Rounding-Based Action Selection}
\label{alg:phase1}
\begin{algorithmic}[1]
\Require Actions with $(p_i, c_i, w_i)$ for all $i \in \sa$, a threshold $R \in \drz$, a contract $\al \in (0,1)$, a budget $W \geq 0$ and $\eps > 0$.
\Ensure A set of actions $S$.
\State Set $h \gets \left\lceil \frac{d+1}{\eps} \right\rceil$, $S_{ALG} \gets \emptyset$, and $z \gets -\infty$.
\ForAll{$S_1 \se \sa$ such that $|S_1| \leq h$}
    \ForAll{$S_2 \se \sa$ such that $|S_2| \leq h \text{ and } S_1 \cup S_2 \in \sv$}
        \State $E_1(S_1) \gets \{i \in \sa \setminus S_1 \mid p_i > p_{min}(S_1)\}$
        \State $E_2(S_2) \gets \{i \in \sa \setminus S_2 \mid q_i > q_{min}(S_2)\}$
        \\
        \State Solve LP \eqref{prob:singleLP} to obtain a fractional basic solution $x^*$.
        \State $T \gets \left\{ i \in 
        \sa \mid x_i^* = 1 \right\}$
        \If{$u_a(T, \al) \geq z$}
            \State $z \gets u_a(T, \al)$
            \State $S_{ALG} \gets T$
        \EndIf
    \EndFor
\EndFor
\If{$S_{ALG} = \emptyset$}
    \State \Return "No feasible set of actions exists"
\Else
    \State \Return $S_{ALG}$
\EndIf
\end{algorithmic}
\end{algorithm}
}

\newcommand{\algoPhaseTwo}{
\begin{algorithm}[hbt]
\caption{Ensuring Principal Utility Through Adaptive Thresholding}
\label{alg:phase2}
\begin{algorithmic}[1]
\Require A contract $\al \in (0,1)$, an algorithm $\alg$ and $\eps > 0$.
\Ensure A set of actions approximating the principal and agent utilities.
\State $H \gets \left\{(1-\al) p_{min} \cdot (1+\eps)^i \mid i \in \left\{ 0, 1, \ldots, \log_{1+\eps} \left(\frac{\sum_{i=1}^{n}{p_i}}{p_{min}} \right) \right\} \right\}$
\State $Y \gets $ Run $\alg$ with threshold $R=0$.
\State $z_0 \gets (1-\eps) \cdot u_a(Y, \al)$
\State $X \gets \emptyset$
\For{$R \in H$}
    \State $S \gets$ Run $\alg$ with threshold $R$.
    \If{$u_a(S,\al) \geq z_0 \text{ and } u_p(S,\al) \geq u_p(X,\al)$}
        \State $X \gets S$
    \EndIf
\EndFor
\\
\Return $X$
\end{algorithmic}
\end{algorithm}
}

\newcommand{\algoMatroid}{
\begin{algorithm}[hbt]
\caption{Greedy Weighted Matroid Maximization}
\label{alg:GreedyMatroid}
\begin{algorithmic}[1]
\Require A Matroid $M = (E, \si)$ and a weight function $w: E \rightarrow \mathbb{R}$.
\Ensure A maximum weight independent set $S \in \si$.
\State Sort and relabel the elements of $E$ in non-increasing order of their weights: $w(e_1) \geq w(e_2) \geq \ldots \geq w(e_n)$
\State $S \gets \emptyset$
\For{each element $x$ in the sorted order}
    \If{$S \cup \{x\} \in \mathcal{I}$}
        \State $S \gets S \cup \{x\}$
    \EndIf
\EndFor
\State \Return $S$
\end{algorithmic}
\end{algorithm}
}

\newcommand{\algoSMFindTx}{
\begin{algorithm}[hbt]
\caption{Find Approximation to $T_k$}
\label{alg:SMFindTx}
\begin{algorithmic}[1]
\Require Parameters $b \in \drz$ and $k \in \left\{0, \delta b, 2 \delta b, \ldots, \left \lceil \frac{n}{\delta} \right \rceil \cdot \delta b \right \}$.
\Ensure A set $T_k$ of agents which is $(1-\eps)$-approximation to $T_k^*$.

\State $T_k \gets \emptyset$
\For{$t \in \{1 ,\ldots, n\}$}
    \For{$r \in \{\ell_i \mid i \in \sa\}$}
    \State $S \gets$ Run Algorithm \ref{alg:SolveDP} with parameters $t, b, r$ and $k$.
    \If{$\tilde{\ell}(S) > \tilde{\ell}(T_k)$}
        \State $T_k \gets S$
    \EndIf
    \EndFor
\EndFor
\State \Return $T_k$
\end{algorithmic}
\end{algorithm}
}

\newcommand{\algoMultiSolveDP}{
\begin{algorithm}[hbt]
\caption{Solve DP}
\label{alg:SolveDP}
\begin{algorithmic}[1]
\Require Parameters $t \in \{1, \ldots, n\}$, $b, r \in \drz$ and $k \in \left\{0, \delta b, 2 \delta b, \ldots, \left \lceil \frac{n}{\delta} \right \rceil \cdot \delta b \right \}$.
\Ensure Set of agents $S$.

\State $Q_0 \gets \left\{(i, x, y, z) \mid (i = 0 \text{ or } x = 0) \text{ and } y \leq 0 \text{ and } z \leq 0 \right\}$ \Comment{Base cases}
\State $Q_{\infty} \gets \left\{(i, x, y, z) \mid (i = 0 \text{ or } x = 0) \text{ and } (y > 0 \text{ or } z > 0) \right\}$ \Comment{Illegal cases}

 \State $B(\cdot,\cdot,\cdot,\cdot) \gets \infty$ \Comment{Initialize the table}
 \State $B(0,0,0,0) \gets 0$

 \For{$i = 1$ to $n$}
    \For{$x = 1$ to $t$}
        \For{$y \in \left\{0, \delta b, 2 \delta b, \ldots, \left\lceil \frac{n}{\delta} \right\rceil \cdot \delta b \right\}$}
            \For{$z \in \left\{0, \delta r, 2 \delta r, \ldots, \left\lceil \frac{n}{\delta} \right\rceil \cdot \delta r \right\}$}
                \State\[
B(i, x, y, z) = 
\begin{cases} 
0, & (i, x, y, z) \in Q_0 \\
\infty, & (i, x, y, z) \in Q_{\infty} \\
\min\left\{ \begin{array}{c}
B(i-1, x, y, z), \\
B(i-1, x-1, y-\tilde{p}_i, z-\tilde{\ell}_i) + w_i
\end{array} \right\}, & \text{otherwise}
\end{cases}
\]
            \EndFor
        \EndFor
    \EndFor
\EndFor

\State \Return the corresponding set of: $\max \left\{z \mid B(n,t,k,z) \leq W \right\}$

\end{algorithmic}
\end{algorithm}
}

\newcommand{\algoMMFindTx}{
\begin{algorithm}[hbt]
\caption{Find Approximation to $T_k$}
\label{alg:MultiFindTx}
\begin{algorithmic}[1]
\Require Agents with $(p_i, c_i, \Vec{w}_i)$ for all $i \in \sa$, a parameter $k \in \dr_{\geq 0}$, budgets $W_1, \ldots, W_d \in \dr_{\geq 0}$ and $\eps > 0$.
\Ensure A set $T_k$ of agent which is $(1-\eps)$-approximation to $T_k^*$.

\State Set $h \gets \left \lceil \frac{d+2}{\eps} \right \rceil$, $S_{ALG} \gets \emptyset$, and $z \gets -\infty$.
\For{$t \in \{1, \ldots, n\}$}
    \ForAll{$S_1 \subseteq \sa$ such that $|S_1| \leq h$}
        \ForAll{$S_2 \subseteq \sa$ such that $|S_2| \leq h \text{ and } S_1 \cup S_2 \in \sv$}
            \State $E_1(S_1) \gets \{i \in \sa \setminus S_1 \mid \tilde{p}_i > \tilde{p}_{min}(S_1)\}$
            \State $E_2(S_2) \gets \{i \in \sa \setminus S_2 \mid \ell_i > \ell_{min}(S_2)\}$
            \\
            \State Solve LP \eqref{prob:multiLP} to obtain a fractional basic solution $x^*$.
            \State $T \gets \left\{ i \in 
            \sa \mid x_i^* = 1 \right\}$
            \If{$\sum_{i \in T}{\ell_i} \geq z$}
                \State $z \gets \sum_{i \in T}{\ell_i}$
                \State $S_{ALG} \gets T$
            \EndIf
        \EndFor
    \EndFor
\EndFor

\If{$S_{ALG} = \emptyset$}
    \State \Return "The problem has no solution"
\Else
    \State \Return $S_{ALG}$
\EndIf
\end{algorithmic}
\end{algorithm}
}

\newcommand{\algoMulti}{
\begin{algorithm}[hbt]
\caption{Framework for Approximate Multi-Agent Contract Problems}
\label{alg:algoMulti}
\begin{algorithmic}[1]
    \Require Agents with $(p_i, c_i, w_i)$ for all $i \in \sa$, budget $W \in \dr_{\geq 0}$, $\eps > 0$ and an algorithm $\alg$.
    \Ensure Subset of agents $S$ that should exert effort.
    \State $S_{ALG} \gets \emptyset$
    \State $k_{ALG} \gets 0$
    \State $\delta \gets \frac{\eps}{n}$
    \For{$b \in \{p_i \mid i \in \sa\}$}
        \For{$k \in \left\{0, \delta b, 2 \delta b, \ldots, \left \lceil \frac{n}{\delta} \right \rceil \cdot \delta b \right \}$}
            \State $T_k \gets \text{Run $\alg$ with parameters } k \text{ and } b$. \label{step:multiRunAlgo}
            \If{$\left(1 - \sum_{i \in S_{ALG}}{\frac{c_i}{p_i}}\right) \cdot k_{ALG} \leq \left(1 - \sum_{i \in T_k}{\frac{c_i}{p_i}}\right) \cdot k$}
                \State $S_{ALG} \gets T_k$
                \State $k_{ALG} \gets k$
            \EndIf
        \EndFor
    \EndFor
    \State \Return $S_{ALG}$
\end{algorithmic}
\end{algorithm}
}

\newcommand{\algoSSFindTx}{
\begin{algorithm}[hbt]
\caption{FPTAS to $\sask$ Problem}
\label{alg:SSInit}
\begin{algorithmic}[1]
    \Require Actions with $(p_i, c_i, w_i)$ for all $i \in \sa$, budget $W \in \dr_{\geq 0}$ and $\eps > 0$.
\Ensure A set of actions $S$.

\State $z \gets $ Run traditional FPTAS of knapsack with profit $q_i$ and weight $w_i$ for each action $i \in \sa$.
\State $bound \gets (1-\eps)z$
\State $S \gets \emptyset$
\For{$b \in \{p_i \mid i \in \sa\}$}
    \For{$r \in \{q_i \mid i \in \sa\}$}
        \State $X \gets$ Run Algorithm \ref{alg:SSSolveDP} with parameters $b$, $r$ and $bound$.
        \If{$\tilde{p}(X) > \tilde{p}(S)$}
            \State $S \gets X$
        \EndIf
    \EndFor
\EndFor
\State \Return $S$
\end{algorithmic}
\end{algorithm}
}

\newcommand{\algoSSSolveDP}{
\begin{algorithm}[hbt]
\caption{Solve DP}
\label{alg:SSSolveDP}
\begin{algorithmic}[1]
\Require Parameters $b, r, bound \in \drz$.
\Ensure A set of actions.

 \State $B(\cdot,\cdot,\cdot) \gets \infty$ \Comment{Initialize the table}
 \State $B(0,0,0) \gets 0$

 \For{$i = 1$ to $n$}
    \For{$x \in \left\{0, \delta b, 2 \delta b, \ldots, \left\lceil \frac{n}{\delta} \right\rceil \cdot \delta b \right\}$}
        \For{$y \in \left\{0, \delta r, 2 \delta r, \ldots, \left\lceil \frac{n}{\delta} \right\rceil \cdot \delta r \right\}$}
                \State\[
B(i, x, y) = 
\begin{cases} 
0, & i = 0 \text{ and } x \leq 0 \text{ and } y \leq 0 \\
\infty, & i = 0 \text{ and } (x > 0 \text{ or } y > 0) \\
\min\left\{ B(i-1, x, y), B(i-1, x-\tilde{p}_i, y-\tilde{q}_i) + w_i \right\}, & \text{otherwise}
\end{cases}
\]
        \EndFor
    \EndFor
\EndFor

\State \Return the corresponding set of: $\max \left\{x \mid B(n,x,y) \leq W \text{ and } y \geq bound \right\}$

\end{algorithmic}
\end{algorithm}
}

\newcommand{\algoMCP}{
\begin{algorithm}[hbt]
\caption{Algorithm for $\mcp$ Problem}
\label{alg:MCP}
\begin{algorithmic}[1]
\Require A matroid $M = (\sa, \si)$ and a contract $\al \in (0,1)$.
\Ensure An optimal solution $S \in \si$.
\State $\delta \gets \min\{|q_k - q_\ell| \mid k, \ell \in \sa \text{ and } |q_k - q_\ell| \neq 0 \}$ (or $0$ if the minimum is undefined).
\State $\Delta \gets 
\begin{cases}
    \frac{\max_{k, \ell \in \sa} |p_k - p_\ell|}{\delta}, & \delta \neq 0 \\
    0, & \delta = 0
\end{cases}$
\State Define the weight function:
\Statex \quad $w(i) = (\Delta + 1) \cdot q_i + p_i \quad \forall i \in \sa$ \footnotemark \label{step:MCPWeightFunction}
\State $S \gets$ Run \cref{alg:GreedyMatroid} with $M$ and $w$.
\State \Return $S$
\end{algorithmic}
\end{algorithm}
\footnotetext{In case $w(j) = w(k)$ for some $j, k \in \sa$, we break ties by the action index as in \cref{dfn:mcpOrder}.}
}

\newcommand{\algoInnerLP}{
\begin{algorithm}[hbt]
\caption{Inner Solver - Matroid Constraint}
\label{alg:innerLPMatroid}
\begin{algorithmic}[1]
    \Require A representative set $T \se \sa$, a threshold $R \in \drz$, a contract $\al \in (0,1)$, a $\frac{1}{2}$-approximation $\beta_q$ to the optimal agent's utility and $0 < \eps < \frac{1}{2}$.
    \Ensure A set of actions $S_{ALG}$.
    \State $S_{ALG} \gets \emptyset$

    \ForAll{ subsets $D \subseteq T \textnormal{ such that } |D| \leq \frac{2}{\eps} \text{ and } D \in \sv$} \label{step:forInnerLP}
        \State Solve the linear program $\textnormal{LP}(D, \beta_q, R)$ and let $x$ be a basic optimal solution.
        \State $K \gets \left \{i \in E(\beta_q, R) \setminus D \mid x_i=1 \right \} \cup D$ \label{step:BMIC-defS}        
        \If{$u_a(S_{ALG}, \al) < u_a(K, \al)$}
            \State $S_{ALG} \gets K$
        \EndIf
    \EndFor
    
    \State \Return $S_{ALG}$
\end{algorithmic}
\end{algorithm}
}

\newcommand{\algoInnerLPMatching}{
\begin{algorithm}[hbt]
\caption{Inner Solver - Matching Constraint}
\label{alg:innerLPMatching}
\begin{algorithmic}[1]
    \Require A representative set $T \se \sa$, a threshold $R \in \drz$, a contract $\al \in (0,1)$ and $0 < \eps < \frac{1}{2}$.
    \Ensure A set of edges $S_{ALG}$.
    \State $S_{ALG} \gets \emptyset$
    
    \ForAll{ subsets $D \subseteq T \textnormal{ such that } |D| \leq \frac{2}{\eps} \text{ and } D \in \sv$} \label{step:forInnerLPMatching}
        \State Compute a basic optimal solution $x$ to $\textnormal{LP}(D, R)$
        \State Apply the rounding algorithm of \cite{chekuri2011multi} to $x$ to obtain an integral matching $\bar{x}$
        \State $M \gets \{ e \in \sa \mid \bar{x}_e = 1 \}$
        \State $K \gets M \cup D$
        \If{$u_a(S_{ALG}, \al) < u_a(K, \al)$}
            \State $S_{ALG} \gets K$
        \EndIf
    \EndFor
    
    \State \Return $S_{ALG}$
\end{algorithmic}
\end{algorithm}
}

\newcommand{\algoBMIC}{
\begin{algorithm}[hbt]
\caption{EPTAS for $\bmic$ and $\match$}
\label{alg:BMICmain}
\begin{algorithmic}[1]
\Require A contract $\al \in (0,1)$, a threshold $R \in \drz$, and $0 < \eps < \frac{1}{2}$.
\Ensure A solution $S$.
\State Compute a $\frac{1}{2}$-approximation $Y$ to the optimal agent's utility using the EPTAS of \cite{AradKS23icalp} with $\eps = \frac{1}{2}$. \label{step:halfApprox}

\State $\beta_q \gets u_a(Y, \al)$
\State $S_i \gets \emptyset \quad \forall i \in \left \{ 0, 1, \ldots, \left \lceil \eps^{-1} \right \rceil \right \}$

\For{$i \in \left \{ 0, 1, \ldots, \left \lceil \eps^{-1} \right \rceil \right \}$}
    \For{$\beta_p \in C$}
        \State $T \gets \emptyset$
        \For{$(r, t) \in \sigma$} \label{step:sigmaLoop}
            \State Construct an exchange set $X$ using the black-box algorithm of Huang and Ward \cite{huang2023fpt}.\label{step:exchangeSet}
            \State $T \gets T \cup X$
        \EndFor \\

        \If{the constraint is a matching}
            \State $S \gets$ Run \cref{alg:innerLPMatching} with $T$ as a representative set
        \ElsIf{the constraint is a matroid}
            \State $S \gets$ Run \cref{alg:innerLPMatroid} with $T$ as a representative set
        \EndIf \\
        
        \If{$S \in \sv$ and $u_a(S, \al) \geq 2\beta_q \cdot (1-14\eps)^i$ and $u_p(S, \al) \geq u_p(S_i, \al)$} \label{step:exchangeIf}
            \State $S_i \gets S$
        \EndIf
    \EndFor
\EndFor

\State $j \gets \argmax \left\{u_a(S_i, \al) \mid i \in  \left \{ 0, 1, \ldots, \left \lceil \eps^{-1} \right \rceil \right \} \right\}$ \label{step:bmicArgmax}
\State \Return $S_j$
\end{algorithmic}
\end{algorithm}
}

\begin{titlepage}
\maketitle

\thispagestyle{empty}
\begin{abstract}
    Consider costly and time-consuming tasks that add up to the success of a project, and must be fitted into a given time-frame. This is an instance of the classic budgeted maximization (knapsack) problem, which admits an FPTAS. Now assume an agent is performing these tasks on behalf of a principal, who is the one to reap the rewards if the project succeeds. The principal must design a contract to incentivize the agent. 
    Is there still an approximation scheme?
    
    In this work we lay the foundations for an algorithm-to-contract framework, which transforms algorithms for combinatorial problems to handle contract design problems subject to the same combinatorial constraints. 
    Our approach diverges from previous works in avoiding the assumption of demand oracle access. 
    As an example, for budgeted maximization, we show how to ``lift'' the classic FPTAS to the best-possible (approximately-IC) FPTAS for the contract problem. We establish this through our local-to-global framework, in which the local step is to approximately solve a two-sided strengthened variant of the demand problem. The global step then utilizes the local one to find the approximately optimal contract. 
    
    We apply our framework to a host of combinatorial constraints: multi-dimensional budgets, budgeted matroid, and budgeted matching constraints. In all cases we essentially match the best purely algorithmic approximation.
    Separately, we also develop a method for multi-agent contract settings. 
    Our method yields the first approximation schemes for multi-agent contract settings that go beyond additive reward functions.
\end{abstract}


\newpage
\setcounter{tocdepth}{2}
\tableofcontents

\thispagestyle{empty}
\end{titlepage}
\newpage

\section{Introduction}
\label{sec:introduction}

Consider the following two problems, one purely algorithmic and the other involving incentives:

\begin{problem}[Budgeted agent] 
    \label{ex:knapsack-alg}
    An agent has $n$ possible actions, 
    where action $i\in [n]$ costs the agent $c_i\ge0$ and leads to $p_i\ge 0$ expected reward.
    The agent can choose any \emph{feasible} combination of actions $S\subseteq[n]$, earning total expected reward $p(S)=\sum_{i\in S}{p_i}$ at total cost $c(S)=\sum_{i\in S}{c_i}$.
    \footnote{Throughout we use the notation $v(S):=\sum_{i\in S}{v_i}$ for any vector $v=(v_1,\dots,v_n)$.} 
    Action set $S$ is \emph{feasible} if $w(S)\le W$, where $w_i\ge 0$ is the size of action $i$ (say, the time it takes), and $W$ is the total budget (say, the remaining time before a deadline).
    The goal is to find a feasible action set $S^\star$ that maximizes the total reward minus cost, i.e., $S^\star\in \arg\max_{\{S\subseteq [n] : S\text{ is feasible}\}} \{p(S)-c(S)\}$.
\end{problem}

In Problem~\ref{ex:knapsack-alg}, the agent's constrained optimization problem coincides with the NP-hard knapsack problem, where item $i$'s value is $p_i-c_i$. The agent can thus apply the well-known FPTAS for knapsack. We now add a strategic component: What changes if the agent is not solving the problem for himself, but rather on behalf of a principal? Does the FPTAS extend from the purely algorithmic setting to the contract setting?

\begin{problem}[Budgeted agent under a contract]
    \label{ex:knapsack-contract}
    As in Problem~\ref{ex:knapsack-alg}, the agent takes a combination of actions $S$ and incurs cost $c(S)$. But now a second player, the principal, reaps the expected reward $p(S)$. The principal must compensate the agent but can only observe the reward, not the cost.
    Without loss of generality, the principal designs a linear contract $\alpha\in(0,1)$, committing to paying the agent an $\alpha$-fraction of the reward. The problem is to find a feasible action set $S^\star$, as well as a linear contract $\alpha^\star$, such that the contract \emph{incentivizes} the agent to take action set $S^\star$;  
    in this case we say that $(S^\star,\alpha^\star)$ is \emph{incentive compatible} (\emph{IC} for short).
    Subject to feasibility and incentive compatibility of $(S^\star,\alpha^\star)$, the goal is to maximize the principal's remaining reward $(1-\alpha^\star) p(S^\star)$ after compensating the agent.
\end{problem}

While Problem~\ref{ex:knapsack-alg} is an optimization problem subject to a combinatorial constraint (in this case, a knapsack/budget constraint), in Problem~\ref{ex:knapsack-contract} we face a contract design problem subject to the same combinatorial constraint. In comparison to Problem~\ref{ex:knapsack-alg}, it has additional IC constraints, and in this work we study the relation between the two problems. Our goal is to understand, from an algorithmic perspective, the loss from delegating the knapsack (or other combinatorial) problem rather than solving it ``in house''. 
Our main question can be formulated as:

\begin{question}
\label{que:main}
Can we transform an approximation algorithm for a constrained optimization problem into an algorithm for a constrained contract problem with the same combinatorial constraint, and if so, at what loss to the approximation factor? 
\end{question}

Our work adds to a long line of research in algorithmic game theory initiated by Nisan and Ronen~\cite{nisan1999algorithmic}, which explores how incentive constraints modify the algorithmic landscape, and seeks general methods of reducing incentive design to algorithm design (see, e.g.,~\cite{briest2005approximation,mu2008truthful,CaiDW12,CaiDW13,CaiDW16,DughmiHKN21}); our work is the first to our knowledge to address this for contract design.

\paragraph{Overview of contribution.}

We consider three families of combinatorial constraints: budget (the action set fits within a single or multi-dimensional budget), budgeted matroid (in addition, the action set is an independent set of a matroid), and budgeted matching (in addition, the action set constitutes a matching in an appropriate graph). These families have natural applications as well as good approximation schemes for their purely algorithmic variants. 

For these families we study single-agent (\textsf{SA}) contract design, where the constraints apply to the agent's action set.
We show a strong positive result --- approximately-optimal contract design reduces to the design of purely algorithmic approximation schemes. In particular we show how to compute, in polynomial time, an approximately-optimal approximately-IC pair $(S^\star,\alpha^\star)$ of action set and contract, 
where the approximation factor matches that of the best-possible algorithmic approximation scheme. 
\footnote{Best-possible assuming $\pnotnp$.} 
This is established through a general approach to transforming positive algorithmic results for classic combinatorial problems, into algorithms for the single-agent contract design variants of these problems under the same combinatorial constraints on the actions. We refer to this method as the \emph{local-to-global} framework. 

We then turn to multi-agent (\textsf{MA}) settings, where multiple agents each choose whether or not to exert effort (a binary action space). We impose budget constraints on the agent (rather than action) set, and again establish a reduction from contract to algorithm design. 

\subsection{Challenges and Impossibilities}

Before detailing our positive results, we highlight the main challenges and impossibilities we face when seeking an answer to \Cref{que:main}. 

\paragraph{Combinatorial contract design.} Our setting of contract design with combinatorial constraints falls within the realm of \emph{combinatorial contract design}. This prolific area of research focuses on the complexity of computing optimal contracts, where these involve selecting a combination of outcomes~\cite{dutting2021complexity}, actions~\cite{dutting2022combinatorial}, agents~\cite{babaioff2006combinatorial,dutting2023multi,castiglioni2023multi}, principals~\cite{alon2024incomplete} --- or several of the above~\cite{DuettingEFK24}. 

A leading approach to combinatorial contract design in the literature encodes the setting's combinatorial structure as a monotone set function $f:2^{[n]} \rightarrow [0,1]$, where $n$ is the number of actions~\cite{dutting2022combinatorial} or agents~\cite{dutting2023multi}. To illustrate the challenges we focus now on $n$ actions. The function $f$ is often referred to as the \emph{reward function}, since it maps a given action set~$S\subseteq[n]$ to the principal's expected reward $f(S)$.
\footnote{The expectation is over the stochasticity of the action set's outcome.}
The algorithm for designing the combinatorial contract is then assumed to have oracle access to $f$, with either just \emph{value} oracle access, or also \emph{demand} oracle access (a stronger assumption originating from auction theory). The level of access assumed usually depends on whether $f$ is \emph{complement-free}, and if so where it falls in the complement-free hierarchy of~\cite{LehmannLN06}: additive $\subsetneq$ gross substitutes $\subsetneq$ submodular $\subsetneq$ XOS $\subsetneq$ subadditive. For XOS and its superclass of subadditive reward functions, demand query access is usually required.

Applying the standard approach to our contract setting with combinatorial constraints from Example~\ref{ex:knapsack-contract}, we get a reward function $f$ that encompasses the budget constraint and is in XOS:

\begin{observation}[Budgeted reward function is XOS]
\label{obs:XOS}
The reward function $f$ in Problem~\ref{ex:knapsack-contract} is
$$
f(S):={\max}_{\{T\subseteq S : T\text{ feasible}\}}p(T),
$$
where $T$ is feasible if $w(T)\le W$; this function is monotone XOS, but not submodular.
\end{observation}

See Appendix~\ref{appx:XOS-proof} for a proof.

\paragraph{Challenge: No demand queries.}
While the standard approach to combinatorial contract design assumes demand oracle access to $f$, this is problematic in the context of \Cref{que:main}, i.e., when studying whether contract design is computationally harder than algorithm design.
The reason is that in our context, giving demand oracle access to the contract design algorithm endows the algorithm with the ability to solve an NP-hard problem --- thus giving it an unfair advantage over the purely combinatorial algorithm. By definition, a demand oracle gets a cost vector~$c$, runs in polynomial time, and returns a set $S\subseteq [n]$ maximizing $f(S)-c(S)$. In the setting of \Cref{ex:knapsack-contract}, this amounts to solving the NP-hard knapsack problem.

More generally, the work of \cite{duetting2025multi} develops a \emph{universal} FPTAS, which applies to \emph{any} single-agent combinatorial contract problem and requires only polynomially-many value and demand queries to~$f$. A consequence of this strong result is that the assumption of demand query access effectively decouples contract design from algorithm design: regardless of the complexity of the algorithmic problem, there is an FPTAS for the corresponding contract problem. Thus, to address \Cref{que:main}, we must diverge from the standard approach and avoid the assumption of demand query access. 

We note that the existence of polytime approximation algorithms for combinatorial contracts without use of demand oracles was posed as an open problem in~\cite{dutting2022combinatorial}.

\paragraph{An impossibility result for single-agent (\textsf{SA}) settings.}

Without demand queries, contract design subject to combinatorial constraints becomes a much harder problem. 
We establish the following impossibility result, which holds already for our simplest family of constraints, namely single-dimensional budget constraints. Note that the impossibility result holds for the single-agent (\textsf{SA}) case only; for the multi-agent (\textsf{MA}) case there is no similar impossibility --- due to the low dimensionality of the action space.

We state the impossibility result using the following notation: Let $(S,\alpha)$ be a pair of feasible action set $S$ and contract~$\alpha$.
Let $u_a(S,\alpha)$ denote the agent's expected utility $\alpha p(S)-c(S)$, and let $u_p(S,\alpha)$ denote the principal's expected utility $(1-\alpha)p(S)$. Given a contract $\alpha$, let $S_{\alpha}$ denote the agent's \emph{best-response} action set, i.e., the action set that maximizes the agent's expected utility subject to feasibility. 
\footnote{We assume standard tie-breaking, by which if multiple action sets yield the same utility for the agent, the one with highest utility for the principal is selected.} 
If $u_a(S, \al)=u_a(S_\al,\al)$ we say $(S, \al)$ is incentive compatible (IC).
Let $\OPT$ be the optimal principal's utility $u_p(S,\alpha)$ over all pairs $(S_\alpha,\alpha)$.
If $u_p(S, \al)=\OPT$ we say $(S, \al)$ is optimal.

Using this terminology, the budgeted \textsf{SA} contract design problem (\Cref{ex:knapsack-contract}) is to find an optimal and IC pair $(S,\alpha)$ where $S$ is feasible. By standard arguments this problem is NP-hard; the following impossibility result strengthens this by ruling out any one-sided approximation. In our context, a one-sided approximation means finding either an optimal approximately-IC pair, or an approximately-optimal IC pair:

\begin{theorem*}[Impossibility; see \cref{hard:principalEpsNP} and \cref{hard:agentEpsNP}]
    Consider the budgeted \textsf{SA} contract design problem (\Cref{ex:knapsack-contract}). For every constant $\eps \in (0,1)$ and $\delta\ge 0$, unless $\pnp$ there is no polynomial-time algorithm that computes $(S,\al)$ such that $S$ is feasible and one of the following holds:
    \begin{itemize}[nosep,noitemsep]
        \item $(S, \al)$ is IC, and $u_p(S, \al)$ is a $(1-\eps)$-approximation to $\OPT$ up to additive $\delta$.
        \footnote{The suffix rules out the transformation of~\cite{dutting2021complexity,zuo2024new}, which allows to obtain IC at a small additive loss to $\OPT$ when the contract setting is normalized.}
        \item $(S, \al)$ is optimal, and $u_a(S, \al)$ is a $(1-\eps)$-approximation to $u_a(S_\al, \al)$.
     \end{itemize}  
\end{theorem*}

The intuition behind this impossibility result is as follows: In \Cref{ex:knapsack-contract}, once the contract $\alpha$ is fixed, then from the agent's point of view he can take any combination of $n$ actions that fit into the budget $W$, with a value of $\alpha p_i-c_i$ for every action $i$. The agent thus faces an NP-hard knapsack problem, and without a demand oracle there is no hope (under standard complexity assumptions) to identify in polynomial time whether a suggested action set $S$ is a best response in order to achieve perfect IC. A similar argument holds for the principal, ruling out perfect optimality. On the positive side, verifying \emph{approximate}-optimality is tractable. Moreover, given the well-known FPTAS for knapsack, it is reasonable to assume that neither party will settle for less than a $(1-\epsilon)$-approximation. 
This motivates \cref{dfn:approximationSchemes} below.

\subsection{Our Results} 

We give here an overview of our results; for details of our technical contribution see \cref{sec:overview}.

\paragraph{Approximation scheme for single-agent (\textsf{SA}) settings.}

To go beyond the impossibility result, we formulate a notion of two-sided approximation in which both principal and agent are guaranteed near-optimal utility. 

\begin{definition}
\label{dfn:approximationSchemes}
    An algorithm $\alg$ is an \emph{approximation scheme for \textsf{SA} contract design} if, given a setting with $n$ actions and an error bound 
    $\eps>0$, it returns a feasible solution $(S, \al)$ satisfying: 
    \begin{equation}
        u_a(S, \al) \geq (1-\eps) \cdot
        u_a(S_{\al}, \al) \quad \textrm{and} 
        \quad u_p(S, \al) \geq (1-\eps) \cdot \OPT.
        \label{eq:approx-scheme}
    \end{equation}  
\end{definition}

While a straightforward generalization of Definition~\ref{dfn:approximationSchemes} allows separate error bounds $\eps,\eps' > 0$ for the principal and agent, for notational simplicity we use the same error bound throughout (this is without loss given our results). 

The two-sided nature of the approximation scheme in Definition~\ref{dfn:approximationSchemes} is captured by Equation~\eqref{eq:approx-scheme}, where the first inequality is an $\eps$-IC requirement. 
The notion of $\eps$-IC is well-studied in equilibrium computation (e.g.~\cite{Daskalakis11,GoosR18}), algorithmic mechanism design (e.g.~\cite{bei2011bayesian,cai2013mechanism,Weinberg14,GonczarowskiW21,cai2021efficient}), and algorithmic contract design (e.g.~\cite{dutting2021complexity,zuo2024new,BacchiocchiGC0025}).
\footnote{The focus of~\cite{BacchiocchiGC0025} is on robust contract design, where the principal’s performance is measured according to the agent’s worst-case approximate best response. This is orthogonal to our goal: we relax incentive constraints, in the spirit of~\cite{Milgrom11}, in order to achieve more with only value queries, whereas robust designs \textit{tighten} incentive constraints to ensure robustness, at the price of achieving less.}
As explained by Balseiro et al.~\cite{BalseiroBC24}, in practice agents may be ``limited by their computational capabilities [...] and may not be able to perfectly optimize their response to a mechanism. This has motivated the introduction of approximate incentive compatibility (IC) as an appealing solution concept for practical mechanism design.'' 
Milgrom~\cite{Milgrom11} highlights that understanding notions of approximate IC and their implications for performance is one of the four critical issues in the practice of market design (see also~\cite{DayMilgrom08,Budish10,KojimaPathak09}). 
There is also behavioral motivation for $\eps$-IC, including bias towards truthfulness in mechanism design and towards fulfilling the principal's instructions in contract design, even at a small personal loss. 

Our focus is on \emph{relative} (or \emph{multiplicative}) $\eps$-IC, which has the advantage of being scale invariant.
\footnote{Daskalakis~\cite{Daskalakis11} in his work on relative $\eps$-NE (Nash equilibrium) gives the following example to illustrate this: ``Imagine a play of some game in which a player is gaining an expected payoﬀ of \$1M from her current strategy, but could improve her payoff to \$1.1M via some other strategy. Compare this situation to a play of the same game where the player's payoff is -\$50k and could become \$50k via a different strategy [...] If one subscribes to the theory of diminishing marginal utility of wealth [...] the two situations could be very different, making the relative notion of approximation more appropriate''.} 
A multiplicative guarantee is appropriate in our context of approximation algorithms: In \cref{ex:knapsack-contract}, both parties face knapsack problems. An additive approximation to the knapsack problem is NP-hard to achieve in general~\cite[Theorem 2.5.2]{kellerer2004basic}, and while it becomes achievable for bounded values, it is then strictly weaker compared to the multiplicative guarantee of the known FPTAS (see \cref{appx:add-approx}). \cref{dfn:approximationSchemes} thus requires a multiplicative approximation for both principal and agent.

\paragraph{Results for single-agent (\textsf{SA}) settings.}

An approximation scheme $\alg$ according to \cref{dfn:approximationSchemes} is classified by its running time; $\alg$ is
\begin{itemize}[nosep, noitemsep]
\item a {\em PTAS} if its runtime is $O( n^{h(1/\eps)})$ for some function $h : \dr_+ \to \dr_+$;
\item an {\em EPTAS} if its runtime is $h (1/\eps) \cdot poly(n)$ for some exponential function $h : \dr_+ \to \dr_+$;
\item an {\em FPTAS} if its runtime is polynomial in both $n$ and $1/\eps$.  
\end{itemize}

\noindent Using this classification, our results for \textsf{SA} contract settings appear in Table~\ref{tab:results-SA}. Interestingly, all results match the best-known guarantees for their algorithmic counterparts, and almost all results are tight (up to the open question for budgeted matching). 
All guarantees are obtained through our local-to-global framework, demonstrating its flexibility. The framework enables us to bring cutting-edge techniques from approximation algorithms into the realm of contract design. By utilizing relax and round-type optimization, it takes a different approach than current techniques for combinatorial contract design, enriching our toolbox.

\begin{table}[t]
\centering
\renewcommand{\arraystretch}{1.2}
\begin{NiceTabular}{cccc}[hvlines]
\toprule
\multicolumn{1}{c}{\bf Single-Agent Problem} & \bf FPTAS & \bf EPTAS & \bf PTAS \\
\midrule

\hline

\makecell{Budgeted \\ ($\sask$)}  & \markwithref{\checkmark}{thm:SASK-FPTAS} & \checkmark & \checkmark \\

\makecell{Multi-Budgeted \\ ($\samk$)} & \xmark & \markwithref{\xmark}{hard:SAMK-EPTAS} & \markwithref{\checkmark}{thm:SAMK-PTAS} \\

\makecell{Budgeted Matroid \\($\bmic$)} & \markwithref{\xmark}{hard:BMIC-FPTAS} & \markwithref{\checkmark}{thm:BMICalgo} & \checkmark \\

\makecell{Budgeted Matching \\ ($\match$)} & Open & \markwithref{\checkmark}{thm:BMICalgo} & \checkmark \\

\bottomrule
\end{NiceTabular}
\caption{Summary of our results for single-agent (\textsf{SA}) contracts with combinatorial constraints. A cell with no reference to a theorem indicates the result follows from another result in the same row.}
\label{tab:results-SA}
\end{table}

\begin{table}[t]
\centering
\renewcommand{\arraystretch}{1.2}
\begin{NiceTabular}{cccc}[hvlines]
\toprule
\multicolumn{1}{c}{\bf Multi-Agent Problem} & \bf FPTAS & \bf EPTAS & \bf PTAS \\
\midrule

\hline

\makecell{Budgeted\\ ($\mask$)} & \markwithref{\checkmark}{thm:MASKalgo} & \checkmark & \checkmark \\

\makecell{Multi-Budgeted\\ ($\mamk$)} & \xmark & \markwithref{\xmark}{hard:MAMK-EPTAS} & \markwithref{\checkmark}{thm:MAMKalgo} \\

\bottomrule
\end{NiceTabular}
\caption{Summary of our results for multi-agent (\textsf{MA}) contracts with combinatorial constraints.} 
\label{tab:results-MA}
\end{table}

\paragraph{Results for multi-agent (\textsf{MA}) settings.}
We adopt the multi-agent binary-effort model of~\cite{dutting2023multi} (see \cref{sec:MASK}).
Each agent $i$ decides whether to exert effort.
If a set $S \subseteq \sa$ exerts effort, the principal obtains reward $f(S)$ and each agent
$i \in S$ incurs cost $c_i \ge 0$.
The principal offers a nonnegative contract $\alpha \in \drz^n$, where $\alpha_i$ denotes agent $i$’s share of the realized reward.
We require \emph{stability}, meaning that $S$ is an exact Nash equilibrium under $\alpha$.
For any fixed $S$, there exists an optimal contract inducing it, so the principal’s problem reduces to selecting a feasible set $S$
that maximizes
$g(S) = \left(1 - \sum_{i \in S} \frac{c_i}{f(S) - f(S \setminus \{i\})} \right) \cdot f(S)$.
An approximation algorithm returns a feasible pair $(S,\alpha)$ such that $S$ is stable under $\alpha$
and $g(S)$ approximates the optimum over all feasible equilibrium outcomes.

In our setting, each agent $i$ additionally has a size $w_i \in \drz$, and the principal faces a total budget $W \in \drz$.
We require that the induced set $S$ be both stable and budget-feasible, i.e.,
$\sum_{i \in S} w_i \le W$.
Accordingly, the principal maximizes $g(S)$ over all equilibrium-inducible sets satisfying this constraint.

The impossibility for \textsf{SA} contract settings with multiple actions does not apply to \textsf{MA} contract settings with binary actions. Thus, for \textsf{MA} settings there is hope to establish approximation schemes that return a feasible solution $(S,\al)$, where $\al$ is near-optimal for the principal and $S$ forms an exact Nash equilibrium for the agents.
However, currently an FPTAS is known only for the vanilla case of an additive function $f$~\cite{dutting2023multi} (in this context $f(S)$ is the principal's reward if the set of agents $S$ exerts effort).
\footnote{There is also an additive PTAS for an instructive special case of graph-based supermodular functions~\cite{deo2024supermodular}.} 
For submodular $f$, a PTAS is ruled out even with demand queries~\cite{ezra2024approximability,duetting2025multi}. 

As noted above (see \cref{obs:XOS}), when viewed through the lens of classical combinatorial contracts, our results yield the first approximation schemes for \textsf{MA} settings that extend beyond additive functions.
Previously, the only approximation scheme known for \textsf{MA} settings was an FPTAS for additive reward functions \cite{dutting2023multi}. 
Our schemes handle budgeted and even multi-budgeted reward functions, which reside in XOS; see Table~\ref{tab:results-MA}.

While we illustrate the application of our method to budget constraints, it can be naturally extended to a broader class of combinatorial constraints.

\subsection{Related Work}
\label{subsec:related-work}

\paragraph{Algorithmic Contract Theory.}

Contracts are an essential economic tool for facilitating delegation and collaboration \cite{holmstrom1979moral,grossman1992analysis}. A contract is an incentive scheme, designed to bridge the informational gap between an agent who exerts effort, and a principal who observes only the noisy outcome of this effort. 
In recent years, contracts are increasingly studied from a computational perspective, starting with works on multiple agents~\cite{babaioff2006combinatorial}, repeated settings~\cite{HoSV16}, and the power of simple linear contracts~\cite{dutting2019simple}. The literature on algorithmic contract design is expanding in directions like learning~\cite{CaiDP15,ChenCDH24}, Bayesian analysis~\cite{GuruganeshSW21,AlonDT21}, or combinatorial inspections~\cite{EzraLR24}; see~\cite{dutting2024algorithmic} for a survey.

Most closely related to our results are works on contract design under explicit combinatorial constraints. Gong et al.~\cite{GongGLY23} study multi-agent contracts with a \emph{cardinality} constraint on the number of active agents, while D\"utting et al.~\cite{dutting2024combinatorial} consider \emph{bipartite matching} constraints in single-agent settings. The latter work gives an optimal algorithm for the special case of \emph{one-sided costs}—where costs are incurred only on one side of the matching—leveraging the existence of a polynomial-time demand oracle, and shows that in the general case the structure becomes significantly more complex. In contrast, in Section~\ref{subsec:matching_LP} we study \emph{general matching} constraints under a budget in the multi-agent setting and focus on approximation algorithms due to inherent computational hardness. Finally, the work of~\cite{DFGR24} considers single-agent settings with submodular reward functions and addresses a question orthogonal to ours: whether the computational hardness of contract optimization is merely an artifact of the difficulty of answering demand queries. They answer this in the negative, showing that even with demand-query access, an exponential number of queries is required.

\paragraph{Knapsack (Budgeted Maximization) Problems.} 

The knapsack problem, famously established as NP-hard by Karp \cite{karp2010reducibility}, remains a cornerstone in computational complexity. While a pseudo-polynomial time dynamic programming solution is known \cite{kellerer2004knapsack}, the development of efficient approximation algorithms has remained a key area of focus. Ibarra and Kim \cite{ibarra1975fast} introduce the first FPTAS for the knapsack problem, employing a rounding technique and dynamic programming. Subsequent research has led to the development of several refined and faster FPTAS (see, e.g., \cite{kellerer2004knapsack,deng2023approximating,chen2024nearly}).
For the multi-dimensional knapsack problem, assuming a constant number of dimensions $d > 1$, Frieze and
Clarke \cite{frieze1984approximation} develop a PTAS. Caprara et al.~\cite{caprara2000approximation} subsequently refine their approach to achieve improved runtime. Kulik and Shachnai \cite{kulik2010there} show there is no EPTAS for this problem already when $d=2$, unless $\pnp$.

There is an extensive line of research on problems generalizing the classic knapsack problem with an additional feasibility constraint, such as a matroid or a matching constraint \cite{GrZe10,BBGS11,chekuri2011multi}, or even independence in graphs \cite{pferschy2009knapsack,doron2025tight}. Berger et al.~\cite{BBGS11} obtain a PTAS for budgeted matching and budgeted matroid intersection based on a Lagrangian relaxation of the problem. Recently, these results were improved to an EPTAS; first for one matroid \cite{doron2023eptas} and then for matroid intersection and matching constraints \cite{AradKS23icalp}, using the representative set technique that was later generalized to other more general constraints \cite{doron2023IPEC,DGK24}. For budgeted maximization on matroids an FPTAS is ruled out \cite{doron2024lower}, and the question is open for budgeted matching. 

While for budgeted matroids an FPTAS is ruled out, there are FPTAS for special cases such as knapsack with a cardinality constraint \cite{caprara2000approximation}, {\em multiple-choice} knapsack \cite{kellerer2004multiple}, and more generally budgeted maximization on a laminar matroid \cite{DKS23laminar}. For more than one budget constraint, there are PTAS for multi-budgeted matching and multi-budgeted matroid intersection \cite{GrZe10,chekuri2011multi}. Finally, there are also works exploring similar variants such as a covering rather than a budgeted constraint along with a matroid constraint \cite{chakaravarthy2013knapsack}, or restricting a solution to be a {\em basis}; for the latter, there is an EPTAS for spanning tree constraint \cite{hassin2004efficient} and a multi-criteria FPTAS \cite{hong2004fully}.

\paragraph{Budgeted Economic Design Problems.}

The study of auctions (rather than contracts) with a combinatorial budget constraint, a.k.a.~\emph{knapsack auctions}, is summarized by Roughgarden in~\cite[Chapter 4]{Roughgarden16} from a computational perspective, and by Milgrom in~\cite{Milgrom21} from an economic perspective. Beyond their theoretic appeal, knapsack auctions are of practical importance with applications like fitting ads onto a webpage or into the Super Bowl commercial break.

We remark that our problem of contract design subject to a budget constraint is unrelated to the settings studied in~\cite{budget-contracts,SaigET24,AharoniHT25,feldman2025budget}. In these works, the budget limits the principal's payment to the agent(s), often with respect to the total project's value. In our work, the budget is a constraint on the agent's set of actions, or alternatively on which sets of agents are feasible (independently of the value they produce).

\paragraph{Approximate demand.}
Our local-to-global framework effectively implements approximate demand queries for the classes of combinatorial constraints that we study. Prior work on combinatorial auctions developed a weaker notion of approximate demand queries to a set function $f$: Given a pricing $c$, the query returns a set $S$ such that $f(S)-c(S) \ge a\cdot f(S')-c(S')$ for all $S' \ne S$, where $a$ is the approximation factor \cite{harshaw2019submodular, sviridenko2017optimal}. In particular, submodular functions have such approximate demand oracles with $a=1-1/e$. An important difference is that $S$ need not approximately-maximize the agent's utility, which is what we seek in this work.

\subsection{Organization}
\cref{sec:Preliminaries} formally introduces our model. 
In \cref{sec:overview}
we present our local-to-global framework and our main technical contributions.
Section~\ref{sec:global_phase} gives a detailed outline of the global phase in our framework. In \cref{sec:SAMK} we instantiate our 
framework to get a PTAS for the multi-budgeted single-agent ($\samk$) problem. In \cref{sec:BMIC} we develop our EPTAS for the budgeted matroid and matching single-agent ($\bmic$ and $\match$) problems.
Section~\ref{sec:Discussion} concludes.
Our remaining results as given in Tables~\ref{tab:results-SA} and \ref{tab:results-MA}, are deferred to Appendices~\ref{appendix:Hardness}-\ref{sec:SASK}.

\section{Preliminaries}
\label{sec:Preliminaries}

\paragraph{Single-agent setting.} 

Let $\mathcal{A}$ be a finite set of actions. Each action $i\in\mathcal{A}$ is associated with a \emph{success contribution} (profit) $p_i\in\dr_{\ge 0}$ and a \emph{cost} $c_i\in\drz$. For any set $S\subseteq\mathcal{A}$, let $p(S) = \sum_{i\in S} p_i$ denote the total success contribution.
A linear contract is parameterized by $\alpha\in(0,1)$: for every unit of realized success, the agent receives an $\alpha$-fraction and the principal retains the remaining $(1-\alpha)$-fraction. Accordingly, if the agent selects $S$ under contract $\alpha$, then the agent’s utility is $u_a(S,\alpha) = \sum_{i\in S}(\alpha p_i-c_i)$ and the principal’s utility is $u_p(S,\alpha) = (1-\alpha)\,p(S)$. For convenience, define the agent’s per-action utility contribution under $\alpha$ as
$q_i = \alpha p_i-c_i$.

In our paper, each action $i$ is further associated with $d$ nonnegative prices $w_{i,1},\ldots,w_{i,d}\in\drz$, and the instance specifies budgets $W_1,\ldots,W_d\in\drz$ (for a fixed constant $d$). For each type of constraint, let $\sv$ denote the family of feasible subsets $S$.

As mentioned above, Theorems~\ref{hard:principalEpsNP} and~\ref{hard:agentEpsNP} establish that the single-agent contract problem is NP-hard, even when approximating either the principal's utility or the agent's utility independently. This necessitates allowing for intermediate, non-optimal solutions on the agent’s part. We begin by introducing some definitions that formalize our requirements for a solution with respect to the agent's utility.

\begin{definition}[IC]
\label{dfn:IC}
    A solution $(S, \al)$ is incentive compatible (IC) if the agent has no incentive to deviate from their chosen action set $S$ under the given contract $\al$. Formally, for any $S' \in \sv$, it holds that $u_a(S, \al) \geq u_a(S', \al)$.
\end{definition}

\begin{definition}[$\eps$-IC]
\label{dfn:epsIC}
    A solution $(S, \al)$ is $\eps$-incentive compatible ($\eps$-IC) if for any $S' \in \sv$, it holds that
    $u_a(S, \al) \geq (1-\eps) \cdot u_a(S', \al)$.
\end{definition}

Given $\varepsilon > 0$, the principal determines both the contract parameter $\alpha$ and the set of actions $S$. This variant, in which the principal specifies both the contract and the agent’s intended response, has been studied in several prior works on algorithmic contract design, including~\cite{dutting2021complexity, AlonDT21, dutting2024ambiguous}. The agent’s role is limited to accepting the proposed set $S$ if it satisfies the $\varepsilon$-incentive compatibility ($\varepsilon$-IC) constraint, or rejecting it otherwise.

The tie-breaking rule is implicit in this model: in the event of a tie in the agent’s utility, the principal is assumed to choose a set of actions that favors their own objective. Notably, setting $\varepsilon = 0$ recovers the original model, where the agent is required to choose an optimal response. Under this condition, the principal must select the utility-maximizing set $S_\alpha$, subject to the tie-breaking rule.

Based on the above discussion, we can formulate our problem with respect to a feasibility family $\sv$ in the following way:
\begin{equation}
\label{prob:general}
\begin{aligned}
    \max_{\al \in (0,1), S \se \sa} \quad & u_p(S,\al) \\
    \textrm{subject to} \quad & u_a(S,\al) \geq (1-\eps) \cdot u_a(S', \al) \quad \forall S' \in \sv, \\
    & S \in \sv.
\end{aligned}
\end{equation}
Using this formulation we obtain the following problems:
\begin{itemize}
    \item \emph{Multi-budget} ($\samk$): $\sv=\{S\subseteq\sa: \sum_{i\in S} w_{i,j}\le W_j, \forall j\in\{1,\ldots,d\}\}$ (and $\sask$ for $d=1$).

    \item \emph{Budgeted matroid} ($\bmic$): given a matroid $\mathcal{M}=(\sa,\si)$, $\sv=\{S\in\si:\ \sum_{i\in S} w_i\le W\}$.

    \item \emph{Budgeted matching} ($\match$): given an undirected graph $G=(V,\sa)$ (actions are the edges), we define $\sv=\{S\subseteq\sa:\ S\in\Gamma(G),\ \sum_{e\in S} w_e\le W\}$, where $\Gamma(G)$ is the family of all matchings in $G$.
\end{itemize}

For a comprehensive treatment of budgeted (knapsack) maximization, we refer the reader to~\cite{kellerer2004knapsack}. For an overview of matroid theory, see, e.g., \cite{oxley2006matroid}.

\section{Overview of Techniques}
\label{sec:overview}
Below we overview the techniques used for deriving our results.

\subsection{A Unifying Local-Global Framework}

A fundamental difference between budgeted maximization and contract design for a budgeted agent is that in the latter problem, there are two objective functions rather than one, namely the expected utilities of both principal and agent. 
To address this, we develop a general, two-phase algorithmic framework comprising of \textit{local} and \textit{global} approximation phases. 
Notably, neither phase requires a demand oracle.

\paragraph{Local phase.}
Implementing the local phase without demand queries poses a key challenge. In a standard two-stage approach, for a fixed contract parameter $\alpha$, one first computes the agent’s best response $u_a(S,\alpha)$ and then selects the value of $\alpha$ that maximizes the resulting principal utility $u_p(S,\alpha)$ over a suitable discretization of $\alpha$. When exact best responses are available, this strategy can yield a good approximation to the principal’s optimal utility.

This reasoning, however, breaks down in our oracle-free setting, where only approximate best responses can be computed due to the NP-hardness of the underlying problem. For a given $\alpha$, there may exist multiple sets $S_1,S_2 \in \mathcal{S}_\mathcal{V}$ that are $\varepsilon$-IC, while only $S_1$ achieves a $(1-\varepsilon)$-approximation to the principal’s best attainable utility at $\alpha$. An algorithm that merely approximates $\max_{T \in \mathcal{S}_\mathcal{V}} u_a(T,\alpha)$ may therefore return $S_2$, resulting in a principal utility that is arbitrarily far from optimal, despite satisfying the agent’s incentive constraint.

In the remainder of this section, we discuss how to address this challenge under several natural constraints, including multi-dimensional budgets and budgeted matroid and matching constraints (see Appendix~\ref{sec:SASK} for the single-dimensional budget case).

\paragraph{Global phase.}
The core new insight behind this framework is that, given a \emph{local} approximation algorithm $-$ that is, one that guarantees a $(1 - \eps)$-approximation to the agent’s utility and at least $(1 - \eps) R$ utility to the principal for some threshold $R \in \drz$, under a fixed contract $\alpha$ $-$ this guarantee can be systematically extended to yield a global approximation across all possible contracts.
This is called the \emph{global} phase, and is achieved by evaluating a polynomial-sized set of candidate thresholds and contracts, ensuring that one pair of these candidates will satisfy the $(1-\eps)$-approximation guarantee for the principal among all contracts, while simultaneously satisfying the $\eps$-IC constraint.
Importantly, the global phase is generic and can be applied in the same manner to all variants of the single agent setting studied in this paper.

\paragraph{Interpretation.}
One way to interpret our framework is as follow. The local phase essentially achieves a strengthened version of an approximate demand query. 
Approximate demand queries as previously defined in the literature (see \cref{subsec:related-work}) do \emph{not} provide strong enough guarantees for our purpose.
We require an approximate demand query that simultaneously approximates both the agent's and the principal's utilities. Achieving such oracles in polynomial-time when the combinatorial constraint yields an NP-hard optimization problem is a major challenge. For example, such oracles are infeasible for certain natural classes of rewards functions like submodular functions~\cite{FeigeJ14,CaiTW20}. 
The difficulty stems directly from having two objective functions. 

Once the local phase is complete, i.e., a two-sided approximate demand oracle is designed, the global phase can be applied.
The local phase utilizes approximation schemes for the underlying combinatorial problems. Thus, our framework essentially shows how to translate approximation schemes into contracts, and spotlights two-sided approximate demand as the crucial component. This answers our motivating question of when do incentive constraints \emph{not} have a significant affect on approximation guarantees (when approximate demand oracles are tractable).

\subsection{Main Technical Contribution: Budgeted Matroid and Budgeted Matching} 

Our method for solving $\bmic$ and $\match$ (see \cref{sec:BMIC}) generalizes and improves the {\em representative set} technique of \cite{doron2023eptas,DGK24,doron2023IPEC}. Given a set of items, the technique finds a small {\em representative set} $-$
the high profit items of a nearly optimal solution. This allows to obtain an EPTAS by: $(i)$ enumerating over subsets of the representative set, and then $(ii)$ adding small profit items, e.g., by solving an LP relaxation of a sub-instance including only small-profit items. Rounding the LP solution may cause only a small harm to the approximation guarantee.

As opposed to previous settings in which the technique was used, $\bmic$ and $\match$ require to approximate both the principal's and the agent's utility, i.e., two objectives. This hinders the direct application of the classic representative set technique, which heavily builds on the attribute that for two items of roughly the same profit $-$ one has a smaller or equal weight and can always be considered more cost-effective than the other. In contrast, our setting does not satisfy this monotonicity property as the utilities of the agent and the principal are unrelated.

Our generalization of the representative set technique leads to two key advances that may be of independent interest. First, it supports the simultaneous optimization of multiple objectives, achieving a $(1 - \eps)$-approximation for each. Second, it constructs a representative set of size $O\left(\eps^{-6}\right)$, a significant improvement over the previous $\Omega\left(\eps^{-(1/\varepsilon)}\right)$ bound of~\cite{AradKS23icalp}. As a result, we obtain faster running times for existing EPTASs for budgeted matroid and matching problems, reducing them from $2^{O \left(\eps^{-2} \cdot \log \left( \frac{1}{\eps} \right) \right)} \cdot poly(n)$ to $2^{O \left(\eps^{-1} \cdot \log \left( \frac{1}{\eps} \right) \right)} \cdot poly(n)$ (see the proof of \cref{lemma:SRSisRepset}).

We note that \cref{hard:BMIC-FPTAS} rules out the existence of an FPTAS for $\bmic$, under standard complexity assumptions.
Furthermore, the existence of an FPTAS for classical budgeted matching  remains an open question (see, e.g.,~\cite{doron2023eptas,doron2023IPEC}). This extends to the contract-design setting. Thus, we focus on developing an EPTAS instead.

\paragraph{Improved runtime via representative sets.}
A natural approach to solving $\bmic$ and $\match$ is to generalize the classical technique for constrained budgeted maximization of~\cite{schafer2015budgeted} to our two-utility setting. The idea is to guess the set of \emph{large} actions with respect to both $p_i$ and $q_i$, and then complete the solution by solving and rounding an LP relaxation over the remaining low-value actions.
While conceptually straightforward, this approach yields only a PTAS, since the brute-force enumeration over candidate sets of large actions scales poorly with $|\mathcal{A}|$.

Our main technical contribution is a refinement of this approach that substantially improves the running time. We construct a small \emph{representative set} $T \subseteq \mathcal{A}$, whose size depends only on $\varepsilon$, with the following guarantee: there exists a $(1-\varepsilon)$-approximate $\varepsilon$-IC solution $S$ whose \emph{large} actions are all contained in $T$. As a consequence, it suffices to enumerate subsets of large actions drawn from $T$; the remaining low-impact actions can then be selected by solving and rounding an LP relaxation, yielding an EPTAS.

To construct the representative set, we introduce the notion of \emph{exchange sets}: small collections of actions with the property that any high-impact action used by a candidate solution can be replaced by a \emph{similar} action from the exchange set, incurring only a negligible loss in both objectives while preserving feasibility (matroid independence or matching). We build exchange sets within coarse \emph{profit classes}, which group actions according to the magnitude of their $(p_i,q_i)$ values, and take their union over all classes.

\paragraph{Completing the solution via LP rounding.}
After constructing the representative set, we address the remaining low-impact actions—those with small $p_i$ and $q_i$ values. We design separate algorithms for the matroid and matching constraints, both relying on a linear programming relaxation to incorporate these actions while maintaining feasibility and preserving the overall approximation guarantee.
\begin{itemize}
    \item For the matroid constraint, we show that, despite restricting enumeration to a representative set rather than the full action space, it is still possible to identify a small collection of actions that are critical to the optimal solution. Once this set is correctly guessed, we solve a linear program over the matroid polytope. The resulting solution contains only a small number of fractionally selected actions. Because this phase involves only low-value items, discarding the fractional components incurs only a negligible loss, enabling us to obtain an EPTAS.

    \item For the matching constraint, the algorithm follows a structure similar to the matroid case. We first identify a small set of critical actions within the representative set. We then focus on actions with sufficiently small profit and cost values and formulate a linear program over the induced subgraph. The resulting fractional solution is rounded using the randomized algorithm of~\cite{chekuri2011multi}, yielding a matching with desirable structural properties. By carefully tuning the rounding parameters and the thresholds defining small values, we show that this matching can be combined with the initial guess to obtain a feasible solution that, with high probability, simultaneously approximates the optimal principal and agent utilities.
\end{itemize}
The following lemma captures the core guarantee provided by the local phase.
\begin{lemma}
\label{lemma:approxBMI}
Given a contract $\alpha \in (0,1)$ and a threshold $R \in \mathbb{R}_{\geq 0}$, if there exists an optimal solution $S_\alpha$ such that $u_p(S_\alpha, \alpha) \geq R$, then there exists a (randomized)
algorithm that runs in time $2^{O \left(\eps^{-1} \cdot \log \left( \frac{1}{\eps} \right) \right)} \cdot \mathrm{poly}(n)$, a feasible set of actions $S \in \mathcal{V}$ such that (with high
probability) $u_p(S, \alpha) \geq (1 - \eps) R$ and the $\eps$-IC constraint is satisfied.
\end{lemma}

\subsection{Multi-Agent Setting}
For the multi-agent setting (see Sections \ref{sec:MASK}, \ref{sec:MAMK}), 
we introduce a general framework that reduces constrained multi-agent instances to a polynomial-size family of structured subproblems. While we describe the framework for single- and multi-dimensional budgets, it extends naturally to other combinatorial feasibility constraints. The framework is modular: any approximation algorithm for the induced subproblems can be invoked only polynomially many times and lifted, via a black-box guarantee, to a near-optimal solution for the original instance.

We then instantiate the framework for the $\mask$ objective. \cite{dutting2023multi} show that $\mask$ is hard even without budgets, yet admits an FPTAS in the unconstrained setting. Building on their structural insights, we extend this guarantee to the \emph{budget-feasible} case. Recall that for a feasible set $S$,  the utility of the principal is $g(S)=\Bigl(1-\sum_{i\in S}\frac{c_i}{p_i}\Bigr)\cdot \sum_{i\in S} p_i$.
Our algorithm exploits this multiplicative structure and proceeds in two phases. First, we discretize agents' parameters via rounding, reducing the relevant value scales to be polynomial in $n$ and $1/\eps$. This yields only polynomially many candidate subproblems and allows us to identify near-optimal candidates for the ``profit'' term $\sum_{i\in S}p_i$, while enforcing auxiliary structural constraints that preserve near-optimality. Second, conditioned on this structure, we optimize the ``penalty'' term $1-\sum_{i\in S}\frac{c_i}{p_i}$ by reformulating the resulting subproblem into an equivalent, more tractable form and solving it with a dynamic program that achieves a $(1-\eps)$-approximation. The lifting guarantee then yields an FPTAS for $\mask$ under the budget constraint.

Finally, we demonstrate the generality of the approach by also deriving a PTAS for $\mamk$ under multi-dimensional budgets via the same reduction and an LP-based subroutine for the resulting subproblems.

\section{Global Phase}
\label{sec:global_phase}
In this section we present the global phase and demonstrate how to systematically leverage the local approximation algorithm—originally providing guarantees only for a fixed contract and a fixed threshold—to obtain a globally near-optimal solution. This step elevates the local guarantees to full approximation guarantees for the overall problem.

\subsection{Approximating the Principal’s Optimal Utility}
We first determine an approximate optimal value of $R$ that can be ensured to the principal for a given contract. 

\begin{lemma}
\label{lemma:phase2}
    Given a contract $\alpha \in (0,1)$,
    let $\alg$ be an algorithm that, for any threshold $R \in [0,u_p(S_\al,\al)]$, returns a set $S \se \sa$ satisfying
    \begin{equation}
    \label{eq:phase2algProps}
        u_a(S, \alpha) \geq (1 - \eps) \cdot u_a(S_\alpha, \alpha) \quad \text{and} \quad u_p(S, \alpha) \geq (1 - \eps) R.
    \end{equation}
    Then one can compute, using a $poly(n, 1/\eps)$ number of calls to $\alg$, a set of actions $S$ such that
    \begin{equation*}
        u_a(S, \alpha) \geq (1 - \eps) \cdot u_a(S_\alpha, \alpha) \quad \text{and} \quad u_p(S, \alpha) \geq (1 - \eps) \cdot u_p(S_\alpha, \alpha).
    \end{equation*}
\end{lemma}
\begin{proof}
    Observe that the optimal principal utility, $u_p(S_\al,\al)$, must satisfy:
    \begin{equation*}
        u_p(S_\al,\al) \in \left[(1-\al) p_{min},(1-\al)\sum_{i=1}^{n}{p_i}\right].
    \end{equation*}
    Denote by $H$ the set of guesses for the optimal principal utility:
    \begin{equation*}
        H = \left\{(1-\al) p_{min} \cdot (1+\eps)^i \mid i \in \left\{ 0, 1, \ldots, \log_{1+\eps} \left(\frac{\sum_{i=1}^{n}{p_i}}{p_{min}} \right) \right\} \right\}.
    \end{equation*}
    Thus, there exist some $r \in H$ such that:
    \begin{equation}
    \label{eq:r-range}
        (1-\eps) \cdot u_p(S_\al,\al) \leq r \leq u_p(S_\al,\al).
    \end{equation}
    Note that the size of $H$ is polynomial in both $n$ and $\frac{1}{\eps}$, even when the values of $p_i$ are arbitrarily small, due to the logarithmic scaling. This holds under the assumption that all $p_i$ values are explicitly provided as part of the input. We now describe an algorithm that constructs the set $S$ in the lemma.
    \algoPhaseTwo

    Algorithm \ref{alg:phase2} iteratively explores different values of $R$ to approximate the optimal principal utility while ensuring that the resulting solution satisfies the $\eps$-IC constraint. Initially, $\alg$ is invoked with the given $\al$ and with the threshold $R=0$, yielding a set $Y$. By Eq.~\eqref{eq:phase2algProps}, since $R=0$ (i.e., the principal does not require any utility), we are guaranteed that: $(1-\eps) \cdot u_a(S_\al,\al) \leq u_a(Y,\al) \leq u_a(S_\al,\al)$ which implies the following bounds on the initial utility threshold $z_0$:
    \begin{equation}
        \label{eq:z0-ratio}
        (1-\eps)^2 \cdot u_a(S_\al,\al) \leq z_0 \leq (1-\eps) \cdot u_a(S_\al,\al).
    \end{equation}

    Next, consider the iteration of \cref{alg:phase2} where $R = r$, and let $S$ denote the set returned by $\alg$ in this iteration. By Eq.~\eqref{eq:phase2algProps}, such a set $S$ must exist, as $r \leq u_p(S_\alpha, \alpha)$. The lemma further guarantees that the principal’s utility satisfies $u_p(S, \al) \geq (1-\eps) \cdot r \geq (1-\eps)^2 \cdot u_p(S_\al, \al)$, as derived from \cref{eq:r-range}. Moreover, by Eq.~\eqref{eq:phase2algProps}, the solution $(S, \al)$ satisfies the $\eps$-IC constraint, and hence, by Eq.~\eqref{eq:z0-ratio}, we have $u_a(S, \alpha) \geq z_0$, as required in \cref{alg:phase2}. By choosing a sufficiently small value of $\eps$, we obtain the desired approximation guarantees for both the principal and the agent.
\end{proof}

\subsection{Computing a Near-Optimal Contract}
To complete the proof, it remains to establish the existence of a contract $\al$ that achieves the desired approximation guarantee. We show that such a contract can be computed using a number of queries polynomial in $n$ and $1/\eps$ to a given algorithm. Our approach builds on key ideas from \cite{duetting2025multi}. For completeness, we provide the full proof in Appendix~\ref{appx:missing_proofs}.

\begin{prop}[Adapted from \cite{duetting2025multi}]
    \label{claim:alphaStarRange}
    There exists a contract $\alpha \in [0,1]$ satisfying
      $1 - \al \leq 1 - \al^* \leq \frac{1 - \al}{1 - \eps}$,
    that can be computed in time $\mathrm{poly}(n, 1/\eps)$, where $\alpha^*$ denotes the optimal contract.
\end{prop}

Leveraging \Cref{claim:alphaStarRange}, we now present the second component of the global phase.

\begin{lemma}
\label{lemma:findAlpha}
Suppose that for any $\eps > 0$ and any contract $\al' \in (0,1)$, algorithm $\alg$ returns a set $S'$ satisfying:
    \begin{equation}
    \label{eq:DemandsOfAlg}
        u_a(S', \al') \geq (1-\eps) \cdot u_a(S_{\al'}, \al') \quad \text{and} \quad u_p(S', \al') \geq (1-\eps) \cdot u_p(S_{\al'}, \al').
    \end{equation}
    Then, one can compute, using a number of queries to $\alg$ that is polynomial in $n$ and $1/\eps$, a solution $(S, \al)$ such that:
    \begin{equation*}
        u_a(S, \al) \geq (1-\eps) \cdot u_a(S_{\al}, \al) \quad \text{and} \quad u_p(S, \al) \geq (1-\eps) \cdot u_p(S_{\al'}, \al') \quad \forall \al' \in (0,1).
    \end{equation*}
\end{lemma}

\begin{proof}
    Let $\eps > 0$, and let $\al^*$ denote the optimal contract. Let $\alpha$ be the contract guaranteed by \cref{claim:alphaStarRange}, and let $S_\al$ and $S_{\al^*}$ denote the optimal sets for the agent when using the contracts $\al$ and $\al^*$, respectively. Proposition 3.1 from \cite{dutting2022combinatorial} guarantees that if $0 < \al^* \leq \al$, then $p(S_{\al^*}) \leq p(S_\al)$. Therefore, by \cref{claim:alphaStarRange}, we have: $(1-\eps) \cdot (1-\al^*) \cdot p(S_{\al^*}) \leq (1-\al) \cdot p(S_{\al})$. Equivalently, based on the definition of the principal's utility, we have: $(1 - \eps) \cdot u_p(S_{\al^*}, \al^*) \leq u_p(S_{\al}, \al)$.
    
    Recall that $\alg$, when provided with a contract $\al$, returns a set of actions $S$ such that $u_p(S, \al) \geq (1-\eps) \cdot u_p(S_{\al}, \al)$. By combining this with the above equations, we derive: $(1 - \eps)^2 \cdot u_p(S_{\al^*}, \al^*) \leq u_p(S, \al)$. Furthermore, the solution $(S, \al)$ satisfies the $\eps$-IC constraint, as directly implied by \cref{eq:DemandsOfAlg}. Therefore, by choosing an appropriate value of $\eps$, we can achieve the desired approximation ratio for both the principal's and agent's utilities, as required.
\end{proof}

\section{Multi-Budgeted Single-Agent Settings}
\label{sec:SAMK}

In this section, we introduce our local-to-global framework through the lens of the $\samk$ problem.

\subsection{Local Phase}

\begin{theorem}
\label{thm:SAMK-PTAS}
There exists a PTAS for $\samk$; that is, for any given $\eps > 0$, the algorithm computes a $(1 - \eps)$-approximation of the optimal principal's utility while satisfying the $\eps$-IC constraint.
\end{theorem}

Given a threshold $R \in [0, u_p(S_\alpha,\alpha)]$ and a contract $\alpha \in (0,1)$, our goal is to find a feasible action set $S$ such that $u_p(S,\alpha) \ge (1-\varepsilon)R$ while satisfying the $\varepsilon$-IC constraint, if such an $S$ exists.

We start by guessing a partial solution, specified by two small sets of actions
$S_1, S_2 \subseteq \sa$, each of size at most
$h = \left\lceil \frac{d+1}{\varepsilon} \right\rceil$.
The combined set $S := S_1 \cup S_2$ is then extended to an approximate solution
by solving a linear program over the remaining actions.

To this end, we exclude certain actions from consideration in the LP phase.
Define
$E_1(S_1) := \{\, i \in \sa \setminus S_1 \mid p_i > p_{\min}(S_1) \,\}$ and
$E_2(S_2) := \{\, i \in \sa \setminus S_2 \mid q_i > q_{\min}(S_2) \,\}$,
where $p_{\min}(S) := \min_{j \in S} p_j$ and
$q_{\min}(S) := \min_{j \in S} q_j$.
We then let $E(S) := E_1(S_1) \cup E_2(S_2)$.
Given an initial guess of $S$, we formulate a linear programming relaxation in which 
$x_i$ is a (fractional) indicator for the selection of action $i$.

\begin{align}
\textnormal{LP}(S) \quad \max_{x \in [0,1]^n} \quad & \textstyle\sum_{i=1}^n q_i \cdot x_i \label{prob:singleLP}\\
\textrm{subject to} \quad & \textstyle\sum_{i=1}^n w_{i,j} \cdot x_i \leq W_j \quad \text{for all } j \in \{1, \ldots, d\}, \label{constraint:knapsackLP} \\
\quad & (1-\al) \cdot \textstyle\sum_{i=1}^n p_i x_i \geq R, \label{constraint:principalUtilityLP} \\
\quad & 0 \leq x_i \leq 1 \quad \text{for } i \notin S \cup E(S), \nonumber \\
\quad & x_i = 1 \quad \text{for } i \in S, \nonumber \\
\quad & x_i = 0 \quad \text{for } i \in E(S) \setminus S. \nonumber
\end{align}

Given an initial guess of $S_1, S_2$, compute an optimal basic solution $x^*$ for $\textnormal{LP}(S)$. 
Denote by $F$ the actions that are fractionally allocated in the solution, i.e., $F = \{i \in \sa \mid 0 < x^*_i < 1\}$. Substituting $x_i = 0$ for all $i \in F$ results in an integral solution. Finally, we return as a solution the set of actions that maximizes the objective function across all possible guesses of $S_1, S_2$.
We give the pseudocode of the local phase in Algorithm \ref{alg:phase1}.
\algoPhaseOne

The correctness of the algorithm follows from the next claims; their proofs appear in Appendix~\ref{appx:missing_proofs}.

\begin{claim}
\label{claim1:SAD}
    Algorithm \ref{alg:phase1} yields a feasible solution for $\sask$ and runs in polynomial time.
\end{claim}

\begin{claim}
\label{claim:numOfFractional}
    The number of fractional variables in $x^*$ is at most $d+1$, i.e., $|F| \leq d+1$.
\end{claim}

\begin{lemma}
\label{lemma:ApproximationRatioProof}
    Given $\al \in (0,1)$ and $R \in \drz$, if $R \leq u_p(S_\al,\al)$, then the integral solution $S$ produced by Algorithm \ref{alg:phase1} satisfies $u_p(S,\al) \geq (1-\eps)R$ and the $\eps$-IC constraint.
\end{lemma}

\begin{proof}
    Let $S_\al = \{i_1, \ldots, i_g\}$ denote an optimal integral solution for the problem. If $g \leq 2h$ then we are done, since in some iteration, the scheme will try $S_1, S_2$ such that $S_1 \cup S_2 = S_\al$. Otherwise, we define two subsets of $S_\al$ in the following way. The first subset, denoted by $S_{1,h}^*$, is constructed by selecting the $h$ actions from $S_\al$ that exhibit the maximum profit values $p_i$. The second subset, $S_{2,h}^*$, is defined analogously, comprising the $h$ actions from $S_\al$ possessing the highest net profit values $q_i$.

    Let $\sigma_1 = \sum_{i \in S_{1,h}^*}{p_i}$ and $\sigma_2 = \sum_{i \in S_{2,h}^*}{q_i}$. Therefore, for any action $i \notin \bigcup_{j \in \{1,2\}}{S_{j,h}^* \cup E_j(S_{j,h}^*)}$, it holds that $p_i \leq \sigma_1 / h$ and $q_i \leq \sigma_2 / h$.

    Let $z_{LP}, x^{LP}$ denote the objective value and the fractional allocation vector of LP \eqref{prob:singleLP}, respectively. Consider the iteration where $S_{1,h}^*$ and $S_{2,h}^*$ are guessed. Observe that by our definitions of $S_{1,h}^*$ and $S_{2,h}^*$, these subsets may overlap. As we do not know $S_{\alpha}$, we need to guess each of these subsets independently, so $S_{2,h}^*$ is selected out of the \emph{entire} set of actions. It is important to note that each action $k \in S_\al$ either belongs to $S_{1,h}^* \cup S_{2,h}^*$ or to the complement of $E_1(S_{1,h}^*) \cup E_2(S_{2,h}^*)$. It follows from the definition of these sets - if $p_k > p_{min}(S_{1,h}^*)$ then $k \in S_{1,h}^*$, and if $q_k > q_{min}(S_{2,h}^*)$ then $k \in S_{2,h}^*$. Conversely, if neither condition holds, then $k \notin E_1(S_{1,h}^*) \cup E_2(S_{2,h}^*)$. Hence, the LP formulation permits the selection of all actions in $S_\al$. Therefore, it holds that:
    \begin{equation}
    \label{eq:validSolution}
        z_{LP} \geq \sum_{i \in S_\al}{q_i} \geq \sum_{i \in S_{2,h}^*}{q_i} = \sigma_2,
    \end{equation}
    where the second inequality holds since $S_{2,h}^* \se S_\al$. Now, to establish the approximation ratio for the agent's utility, we have:
    \begin{align*}
        u_a(S,\al) & = z_{LP} - \sum_{i \in F}{q_i \cdot x_i^{LP}} \geq z_{LP} - \sum_{i \in F}{q_i} \\
        & \geq z_{LP} - {(d+1) \cdot \frac{\sigma_2}{h}} \geq z_{LP} - {(d+1) \cdot \frac{z_{LP}}{h}} \geq z_{LP} \cdot (1-\eps),
    \end{align*}
    where the second inequality holds by \cref{claim:numOfFractional} and the the fact that $F \cap (S_{2,h}^* \cup E(S_{2,h}^*)) = \emptyset$ implies $q_i \leq \sigma_2 / h$, for all $i \in F$, and the third inequality is due to \cref{eq:validSolution}.

    To show that the rounded principal's utility, $u_p(S,\al)$, satisfies $u_p(S,\al) \geq (1-\eps) R$, we focus on the loss caused by the rounding. Since $x^{LP}$ is a feasible solution, constraint \eqref{constraint:principalUtilityLP} implies that $(1-\al) \cdot \sum_{i=1}^n p_i x_i^{LP} \geq R$. We consider two cases: If $(1-\al) \sigma_1 > R$, then we achieve the required utility directly, as $u_p(S,\al) \geq (1-\al) \sigma_1 > R \geq (1-\eps)R$. Here, the first inequality holds because the initial selection of $S_1$ and $S_2$ ensures that the principal's utility is at least $(1-\al) \sigma_1$, and including additional actions in the LP cannot decrease this value. The second case arises when $(1-\al) \sigma_1 \leq R$. In this scenario we have:
    \begin{align*}
        u_p(S,\al) & = R - (1-\al)\cdot \sum_{i \in F}{p_i x_i^{LP}} \geq R - (1-\al)\cdot \sum_{i \in F}{p_i} \\
        & \geq R - (d+1) \cdot {\frac{(1-\al) \cdot \sigma_1}{h}} \geq R - (d+1) \cdot {\frac{R}{h}} \geq R \cdot (1-\eps),
    \end{align*}
    where the justification for the second inequality is analogous to the explanation for the agent's utility, and the third inequality follows directly from the assumption $(1-\al) \sigma_1 \leq R$.
\end{proof}

\subsection{Global Phase}
We now apply the
global phase, as described in \cref{sec:global_phase}.

\begin{proof}[Proof of \cref{thm:SAMK-PTAS}]
Observe that \cref{lemma:ApproximationRatioProof} establishes the approximation guarantees required to apply \cref{lemma:phase2}, which in turn satisfies the conditions needed to invoke \cref{lemma:findAlpha}. The desired approximation guarantee follows directly.

Furthermore, the resulting algorithm is a PTAS, since \cref{alg:phase1} runs in polynomial time and both
\cref{lemma:phase2} and \cref{lemma:findAlpha} invoke it only polynomially many times in $n$ and $1/\varepsilon$.
\end{proof}

\section{Budgeted Matroid and Matching in Single-Agent Settings}
\label{sec:BMIC}

In this section, we present an EPTAS for the $\bmic$ and $\match$ problems, where the agent’s set of actions is subject to a budget constraint as well as a matroid or a matching constraint.
To tackle the multi-objective scenario imposed by contracts, we generalize the {\em representative set} technique for budgeted matroid and budgeted matching proposed by \cite{doron2023eptas,doron2023IPEC} (see \cref{sec:introduction}). Our main result is the following.
\begin{theorem}
\label{thm:BMICalgo}
    Both $\bmic$ and $\match$ admit an EPTAS.
\end{theorem}

\subsection{Local Phase: Constructing a Representative Set}

Given a contract $\al \in (0,1)$, let $p_i'(\al) = (1-\al) \cdot p_i$. When $\al$ is clear from the context, we use $p_i'= p_i'(\al)$. Towards proving \cref{lemma:approxBMI} for this setting, we introduce some useful notations and claims. To simplify notation, we use $z_i$ for either $p_i'$ or $q_i$ when applying the same arguments to both, and $z(S)$ to denote either $u_p(S, \al)$ or $u_a(S, \al)$. This unified notation will streamline the subsequent analysis.

Denote by $H_p$ the set of actions with "high" profit, i.e., $H_p = \left \{ i \in \sa \mid p_i' \geq \eps \cdot u_p(S_\al, \al) \right \}$, and similarly for $q_i$, $H_q = \left \{ i \in \sa \mid q_i \geq \eps \cdot u_a(S_\al, \al) \right \}$. Moreover, let $\frac{u_a(S_\al, \al)}{2} \leq \beta_q \leq u_a(S_\al, \al)$. This value can be found in $poly(n)$ time using, e.g., the EPTAS of \cite{AradKS23icalp} with $\eps = \frac{1}{2}$ using $q_i$ as the profit of the $i$-th item. Finally, let $H = H_p \cup H_q$ and $\psi(\eps) = 2 \eps^{-2}$.

\begin{definition}
    Denote by $\sv$ the set of feasible solutions under either the matroid or matching constraint. Given $x \in \drz$, let $\sv_{\leq x}$ be the subset of solutions $S \in \sv$ satisfying $|S| \leq x$.
\end{definition}

\begin{definition}[Representative Set]
\label{dfn:representative}
    Let $0<\eps<\frac{1}{2}$ and $T \se \sa$. We say that $T$ is a representative set, if there is a solution $S \in \sv$ such that the following holds:
    \begin{enumerate}
        \item $S \cap H \se T$,
        \item $u_p(S, \al) \geq (1-6\eps) \cdot u_p(S_\al, \al)$,
        \item $u_a(S, \al) \geq (1-6\eps) \cdot u_a(S_\al, \al)$.
    \end{enumerate}
\end{definition}

\begin{lemma}
\label{lemma:InnerSolveLP}
    Given a representative set $T \se \sa$ such that $|T| \leq 4 \psi(\eps)^3$, there exists an algorithm that runs in time $2^{O \left(\eps^{-1} \cdot \log \left( \frac{1}{\eps} \right) \right)} \cdot poly(n)$ and returns (with high probability) a feasible solution $S \in \sv$ such that $u_a(S, \al) \geq (1-14\eps) \cdot u_a(S_\al, \al)$ and $u_p(S, \al) \geq (1-14\eps) R$ for any 
    $R \in \left[\frac{u_p(S_\al, \al)}{2},u_p(S_\al, \al)\right]$.
    \footnote{By a standard rescaling of $\eps$, this implies a $(1-\eps)$-approximation.}
\end{lemma}

Assuming the correctness of \cref{lemma:InnerSolveLP} (see Sections~\ref{subsec:matching_LP} and \ref{subsec:matroid_LP} for the proof), we now demonstrate how to construct a small representative set and prove \cref{lemma:approxBMI}.

\begin{definition}[Replacement]
\label{dfn:replacement}
    Let $0<\eps<\frac{1}{2}$, $S \in \sv_{\leq \psi (\eps)}$, and $Z_S \subseteq \sa$. Then $Z_S$ is a replacement of $S$ if the following holds: 
    \begin{enumerate}
        \item $(S \setminus H) \cup Z_S \in \sv_{\leq \psi(\eps)}$, \label{prop:rep1}
        \item $w(Z_S) \leq w(S \cap H)$, \label{prop:rep2}
        \item $u_p \left((S \setminus H) \cup Z_S, \al \right) \geq (1-\eps) \cdot u_p(S, \al)$, \label{prop:rep3}
        \item $u_a\left((S \setminus H) \cup Z_S, \al \right) \geq (1-\eps) \cdot u_a(S, \al)$, \label{prop:rep4}
        \item $|Z_S| \leq |S \cap H|$. \label{prop:rep5}
    \end{enumerate}
\end{definition}

\begin{definition}[{\sc SRS}]
\label{dfn:SRS}
    Let $0<\eps<\frac{1}{2}$ and $T \subseteq \sa$, we say that $T$ is a strict representative set (SRS) if, for any $S \in \sv_{\leq \psi(\eps)}$, there is a replacement $Z_S \subseteq T$ of $S$.
\end{definition}

\subsubsection{SRS Implies a Representative Set}

The next lemma will be instrumental in proving \cref{lemma:approxBMI}.

\begin{lemma}
\label{lemma:SRSisRepset}
    Given $\al \in (0,1)$ and $\eps \in (0,\frac{1}{2})$, let $T \se \sa$ be an SRS. Then, $T$ is also a representative set.
\end{lemma}
\begin{proof}
   Let $S_\al$ be the optimal response of the agent. Denote by
  \begin{equation*}
       L_p = \{ i \in S_\al \mid p_i' \geq \eps^2 \cdot u_p(S_\al, \al) \} \qquad \text{and} \qquad
       L_q = \{ i \in S_\al \mid q_i \geq \eps^2 \cdot u_a(S_\al, \al) \}
   \end{equation*}
   the sets of actions with high values of utility for the principal and the agent, respectively.
   Also denote $L = L_p \cup L_q$ and $Q = S_\al \setminus L$.
   
  \begin{claim}
   \label{claim:LinMatroid}
       $L \in \sv_{\leq \psi(\eps)}$.
   \end{claim}
   
   Since $T$ is an SRS, and due to \cref{claim:LinMatroid}, there is a replacement $Z_L \se T$ of $L$. Let $\Delta_L = (L \setminus H) \cup Z_L$. By \cref{dfn:replacement}, it holds that $\Delta_L \in \sv_{\leq \psi(\eps)}$. The following claim is required in order to complete the construction of $S$:

  \begin{claim}
  \label{claim:findXtoS}
       There is a set $X \se Q \setminus \Delta_L$ such that $|X| \geq |Q| - 5 \eps^{-1}$ and $\Delta_L \cup X \in \sv$.
   \end{claim}

   Let $X$ be the set from \cref{claim:findXtoS}, and define $S = \Delta_L \cup X$. The following claims establish key properties of these sets:
   \begin{claim}
   \label{claim:SisFeasible}
       $S \in \sv$ and $S \cap H \subseteq T$.
   \end{claim}

   It remains to analyze the approximation ratio of $S$. Note that:
   \begin{equation}
   \label{eq:smallElementsBound}
       z(Q \setminus X) \leq |Q \setminus X| \cdot z_{max}(Q) \leq 5 \eps^{-1} \cdot z_{\max}(Q) \leq 5 \eps^{-1} \cdot \eps^2 \cdot z(S_\alpha) = 5\eps \cdot z(S_\alpha),
   \end{equation}
   where the second inequality follows from \cref{claim:findXtoS}, and the third holds since $z_i \leq \eps^2 \cdot z(S_\al)$ for all $i \in Q$. We can conclude that the approximation ratio of $S$ is:
   \begin{align*}
      z(S) & = z(\Delta_L) + z(X) \geq (1-\eps) \cdot z(L) + z(Q) - z(Q \setminus X) \\
       & \geq (1-\eps) \cdot z(S_\al) - z(Q \setminus X) \geq (1-\eps) \cdot z(S_\al) - 5\eps \cdot z(S_\al) = (1-6\eps) \cdot z(S_\al),
   \end{align*}
   where the first equality follows from the fact that $\Delta_L \cap X = \emptyset$. The first inequality is justified by \cref{dfn:replacement}. The second inequality holds because $S_\al = L \cup Q$ and it is trivially true that $z(Q) \geq (1-\eps) \cdot z(Q)$. Finally, the last inequality is based on \cref{eq:smallElementsBound}.    
\end{proof}

\subsubsection{Constructing an SRS via Exchange and Substitution Sets}

By \cref{lemma:SRSisRepset}, it suffices to construct an SRS to obtain a representative set. Let
\begin{equation*}
   C = \left\{(1-\al) p_{min} \cdot 2^i \mid i \in \left[ \log_{2} \left(\frac{\sum_{i=1}^{n}{p_i}}{p_{min}} \right) \right] \right\}
\end{equation*}
be the set of guesses for the principal's optimal utility, where $p_{min} = \min_{i \in \sa}p_i$. Since $u_p(S_\al, \al) \in \left[(1-\al) p_{min},(1-\al)\sum_{i=1}^{n}{p_i}\right]$, there exists $\beta_p \in C$ such that:
\begin{equation}
\label{eq:betaP-range}
    \frac{u_p(S_\al, \al)}{2} \leq \beta_p \leq u_p(S_\al, \al).
\end{equation}
Consider the guess of $\beta_p$ that satisfies \cref{eq:betaP-range}. Partition the actions into profit classes as follows:
\begin{gather*}
    \forall r \in \left \{1, \ldots, \left \lceil \log_{1-\eps} \left(\frac{\eps}{2}\right) \right \rceil \right \}: \quad \sk_{r}^{(p)}(\beta_p) = \left \{ i \in \sa \mid \frac{p_i'}{2 \beta_p} \in \left ((1-\eps)^r, (1-\eps)^{r-1} \right ] \right \}, \\
    \forall t \in \left \{1, \ldots, \left \lceil \log_{1-\eps} \left(\frac{\eps}{2}\right) \right \rceil \right \}: \quad \sk_{t}^{(q)}(\beta_q) = \left \{ i \in \sa \mid \frac{q_i}{2 \beta_q} \in \left ((1-\eps)^t, (1-\eps)^{t-1} \right ] \right \}.
\end{gather*}
Furthermore, we define:
\begin{gather*}
    \sk_0^{(p)}(\beta_p) = \sa \setminus \bigcup_{r \in \left \{1, \ldots, \left \lceil \log_{1-\eps} \left(\frac{\eps}{2}\right) \right \rceil \right \}} \sk_{r}^{(p)}(\beta_p), \qquad
    \sk_0^{(q)}(\beta_q) = \sa \setminus \bigcup_{t \in \left \{1, \ldots, \left \lceil \log_{1-\eps} \left(\frac{\eps}{2}\right) \right \rceil \right \}} \sk_{t}^{(q)}(\beta_q).
\end{gather*}
Finally, for $(r, t) \in \left \{0, 1, \ldots, \left \lceil \log_{1-\eps} \left(\frac{\eps}{2}\right) \right \rceil \right \} ^2 \setminus (0,0)$, denote:
\begin{equation*}
    \sk_{r, t}(\beta_p, \beta_q) = \sk_{r}^{(p)}(\beta_p) \cap \sk_{t}^{(q)}(\beta_q).
\end{equation*}
For convenience, let $\sigma = \left \{0, 1, \ldots, \left \lceil \log_{1-\eps} \left(\frac{\eps}{2}\right) \right \rceil \right \} ^2 \setminus (0,0)$. This set represents the indices for which we defined the profit classes $\sk_{r, t}(\beta_p, \beta_q)$. This partition of the actions has the following useful property:

\begin{observation}
\label{obs:exactlyOne}
    For every $i \in \bigcup_{(r, t) \in \sigma} \sk_{r,t}(\beta_p, \beta_q)$ such that $\{i\} \in \sv$, there exists a unique $(r, t) \in \sigma$ for which $i \in \sk_{r,t}(\beta_p, \beta_q)$.
\end{observation}

Based on the above partition to profit classes, we define the concepts of an exchange set and a substitution set.
\begin{definition}[Exchange Set]
\label{dfn:exchangeSet}
    Given $\al \in (0,1)$, $0<\eps<\frac{1}{2}$, $(r, t) \in \sigma$, and $X \subseteq {\sk}_{r,t}(\beta_p, \beta_q)$. We say that $X$ is an exchange set if: 

    For all $\Delta \in \sv_{\leq \psi (\eps)}$ and $a \in (\Delta \cap {\sk}_{r,t}(\beta_p, \beta_q)) \setminus X$ there is $b \in X \setminus \Delta$ satisfying 
    \begin{itemize}
        \item $w(b) \leq w(a)$,
        \item $\Delta \setminus \{a\} \cup \{b\} \in \sv_{\leq \psi(\eps)}$.
    \end{itemize} 
\end{definition}

\begin{definition}[Substitution]
\label{dfn:sub}
    For $G \in \sv_{\leq \psi(\eps)}$ and $Z_G \se \bigcup_{(r, t) \in \sigma } {\sk}_{r,t}(\beta_p, \beta_q)$, $Z_G$ is a substitution of $G$ if the following holds.
    \begin{enumerate}
        \item $(G \setminus H) \cup Z_G \in \sv_{\leq \psi(\eps)}$, \label{prop:sub1}
        
        \item $w(Z_G) \leq w(G \cap H)$, \label{prop:sub2}
    
        \item $(G \setminus H) \cap Z_G = \emptyset$, \label{prop:sub3}
        
        \item For all $(r, t) \in \sigma$ it holds that $|{\sk}_{r,t}(\beta_p, \beta_q) \cap Z_G| = |{\sk}_{r,t}(\beta_p, \beta_q) \cap G \cap H|$. \label{prop:sub4}
    \end{enumerate}
\end{definition}
The following claims give some useful properties of the substitution set.
\begin{claim}
\label{claim:thereIsSubstitution}
    If for all $(r, t) \in \sigma$ it holds that $T \cap {\sk}_{r, t}(\beta_p, \beta_q)$ is an exchange set, then, for any $G \in \sv_{\leq \psi (\eps)}$, there is a substitution $Z_G$ of $G$ such that $Z_G \se T$. 
\end{claim}

\begin{claim}
\label{claim:subtitueIsRepset}
    Let $G \in \sv_{\leq \psi(\eps)}$, and let $Z_G$ be a substitution of $G$ such that $Z_G \se T$. Then $Z_G$ is a replacement of $G$.
\end{claim}

Using the above claims we have the following:
\begin{lemma}
\label{lemma:sufficientRep}
    Let $T \se \sa$. If $T \cap {\sk}_{r, t}(\beta_p, \beta_q)$ is an exchange set
 for all $(r, t) \in \sigma$, then $T$ is an SRS. 
\end{lemma}
\begin{proof}
    Let $G \in \sv_{\leq \psi(\eps)}$. To show that $T$ is an SRS, we need to demonstrate that there exists a replacement $Z_G \se T$ of $G$. By \cref{claim:thereIsSubstitution}, there exists a substitution $Z_G$ of $G$ such that $Z_G \se T$. Moreover, by \cref{claim:subtitueIsRepset}, $Z_G$ is also a replacement of $G$, as required.
\end{proof}

\subsubsection{Constructing an Exchange Set}

For the proof of \cref{lemma:approxBMI}, we need to show how to construct an exchange set. Huang and Ward \cite[Theorem 3.6]{huang2023fpt} presented an algorithm that, given $\beta_p, \beta_q$ and $(r,t) \in \sigma$, constructs an exchange set $X \se T \cap {\sk}_{r, t}(\beta_p, \beta_q)$ such that $|X| \leq \psi(\eps)$ in $poly(n)$ time. By leveraging this black-box algorithm, we can propose an algorithm for constructing a solution for the local phase for our problems (see \Cref{alg:BMICmain}).

\algoBMIC

\begin{proof}[Proof of \cref{lemma:approxBMI}]
    Note that when $\beta_p$ satisfies \cref{eq:betaP-range}, $T$ is a representative set, since for all $(r, t) \in \sigma$, $T \cap {\sk}_{r, t}(\beta_p, \beta_q)$ is an exchange set. Therefore, by Lemmas \ref{lemma:sufficientRep} and \ref{lemma:SRSisRepset}, we may consider the iteration of the algorithm where $T$ is a representative set. Hence, for all $i \in \left \{ 0, 1, \ldots, \left \lceil \eps^{-1} \right \rceil \right \}$, if $S_i \neq \emptyset$, \cref{lemma:InnerSolveLP} implies that:
    \begin{equation}
    \label{eq:bmicPrincipalBound}
        (1-14\eps) \cdot u_p(S_\al, \al) \leq u_p(S_i, \al).
    \end{equation}
    Now, by the definition of $\beta_q$, there exists an $i^* \in \left \{ 0, 1, \ldots, \left \lceil \eps^{-1} \right \rceil \right \}$ such that:
    \begin{equation*}
        (1-14\eps)^2 \cdot u_a(S_\al, \al) \leq  2\beta_q \cdot (1-14\eps)^{i^*} \leq (1-14\eps) \cdot u_a(S_\al, \al).
    \end{equation*}
    In this iteration, by \cref{lemma:InnerSolveLP}, $S$ satisfies $(1-14\eps) \cdot u_a(S_\al, \al) \leq u_a(S, \al)$. It follows that $S_{i^*} \neq \emptyset$. Combining this with the previous inequality, we get: $(1-14\eps)^2 \cdot u_a(S_\al, \al) \leq u_a(S, \al)$.
    
    Finally, we return a set $S_j$ that satisfies \cref{eq:bmicPrincipalBound} and $u_a(S_j, \al) \geq u_a(S_{i^*}, \al)$. Therefore, by choosing an appropriate value of $\eps$, we achieve the desired approximation.

    Now, observe that the following equation holds for all $0<\eps<\frac{1}{2}$:
    \begin{equation}
    \label{eq:epsBound}
        \left \lceil \log_{1-\eps} \left(\frac{\eps}{2}\right) \right \rceil + 1 \leq \log_{1-\eps} \left(\frac{\eps}{2}\right) + 2 \leq 4 \eps^{-2}.
    \end{equation}

    For the running time of \cref{alg:BMICmain}, we note that $|C| = poly(n)$, and \cref{eq:epsBound}, $|\sigma| \leq \frac{16}{\eps^{4}}$. By \cref{lemma:InnerSolveLP} each iteration takes $2^{O \left(\eps^{-1} \cdot \log \left( \frac{1}{\eps} \right) \right)} \cdot poly(n)$ time, and Step \eqref{step:halfApprox} takes $poly(n)$ time, since $\eps = \frac{1}{2}$ is a constant. Therefore, the total running time is bounded by:
    $poly(n) \cdot \eps^{-1} \cdot 2^{O \left(\eps^{-1} \cdot \log \left( \frac{1}{\eps} \right) \right)} = poly(n) \cdot 2^{O \left(\eps^{-1} \cdot \log \left( \frac{1}{\eps} \right) \right)}$.

    To bound the size of the returned representative set, we note that there are at most $16 \eps^{-4} = 4 \psi(\eps)^2$ iterations of Step \eqref{step:sigmaLoop}, and the size of each exchange set in Step \eqref{step:exchangeSet} is at most $\psi(\eps)$ using the results of \cite{huang2023fpt}.
\end{proof}

\subsection{Local Phase: Matching}
\label{subsec:matching_LP}

We now proceed to prove \cref{lemma:InnerSolveLP} for the matching constraint. 

Let $L$ and $D$ be two sets to be defined later. We consider the following linear program.
\begin{align}
\textnormal{LP}(D ,R)~~~~~~~~~~~~~~~~~
\max \quad & \sum_{e \in L \setminus D} q_e \cdot x_e \nonumber \\
\text{subject to\quad} & (1-\al) \sum_{e \in L \setminus D} p_e \cdot x_e \geq (1-6\eps) R - u_p(D, \al), \label{constraint:matching-principal}\\
& \sum_{e \in L \setminus D} w_e \cdot x_e \leq W - w(D), \nonumber \\
& x \in \mathcal{P}_{G_D}. \nonumber
\end{align}
Here, $\mathcal{P}_{G_D}$ denotes the matching polytope of the graph $G_D = (V,\sa_D)$ where $\sa_D$ is obtained by removing from $\sa$ all edges that share an endpoint with any edge in $D$. Let $\beta_p,\beta_q$ be the guesses satisfying
\[
\frac{u_p(S_\al,\al)}{2} \le \beta_p \le u_p(S_\al,\al),
\qquad
\frac{u_a(S_\al,\al)}{2} \le \beta_q \le u_a(S_\al,\al),
\]
We define the set of \emph{low-value edges} as
\[
L
=
\left\{
e \in \sa
\;\middle|\;
p_e' \le 2\eps \beta_p,\;
q_e \le 2\eps \beta_q
\right\}.
\]
Given a representative set $T \subseteq \sa$, we exhaustively enumerate all subsets
$D \subseteq T$ with $|D| \le 2/\eps$.
For each such $D$, we solve $\textnormal{LP}(D,R)$ over the ground set $L\setminus D$.

\begin{claim}
\label{claim:matching-LP}
When $D = S \cap H$ is correctly guessed, the set $S \setminus H$ forms a feasible solution to $\textnormal{LP}(D, R)$.
\end{claim}

Let $x$ denote an optimal (fractional) solution to $\textnormal{LP}(D,R)$. Chekuri et al.~\cite{chekuri2011multi} presented a polynomial-time randomized rounding algorithm which, given a fractional matching $x$, produces an integral matching $\bar{x}$ corresponding to a vertex of the matching polytope $\mathcal{P}_{G_D}$. Let $M = \{ e \in \sa \mid \bar{x}_e = 1 \}$ denote integral matching.

\begin{prop}[\cite{chekuri2011multi}]
\label{claim:rounding_props}
   For any $\gamma \in (0,\frac{1}{2})$ and $\eps \in [0,1]$, it holds that:
    \begin{enumerate}
        \item $\mathbb{E}[\mathbbm{1}_M] = (1-\gamma)x$,
        
        \item $\Pr \left[ z(M) \leq (1-\eps) \cdot \mu_z \right] \leq e^{-\mu_z \gamma \eps^2/20}$, where $\mu_z = \mathbb{E}[z(M)]$ for $z \in \{p', q\}$,
        
        \item $\Pr \left[ w(M) \geq (1+\eps) \cdot \mu_w \right] \leq e^{-\mu_w \gamma \eps^2/20}$, where $\mu_w = \mathbb{E}[w(M)]$.
    \end{enumerate}
\end{prop}

\begin{claim}
\label{claim:utilitiesMatchingProps}
    With high probability, the matching $M$ satisfies:
    \begin{itemize}
        \item $u_a(M, \al) \geq (1 - \eps)^2 \cdot u_a(S \setminus H, \al)$,
        \item $u_p(M, \al) \geq (1 - \eps)^2 \cdot \left((1 - 6\eps) R - u_p(D, \al)\right)$.
    \end{itemize}
\end{claim}

\begin{proof}
    Fix $z \in \{p',q\}$ and let $x$ denote an optimal fractional solution of $\textnormal{LP}(D,R)$.
    Define the fractional $z$-value $Z_x := \sum_{e \in L \setminus D} x_e z_e$, and let $M$ be the integral matching obtained by applying the randomized rounding procedure of \cite{chekuri2011multi} to $x$.
    
    Setting $\gamma=\eps$ in \Cref{claim:rounding_props}, we have for every edge $e \in L\setminus D$ that $\mathbb{E}[\mathbbm{1}_M(e)] = (1-\eps)\,x_e$. By linearity of expectation,
    \begin{align*}
    \mu_z
    \;:=\;
    \mathbb{E}[z(M)]
    =
    \mathbb{E}\!\left[\sum_{e \in L \setminus D} z_e \mathbbm{1}_M(e)\right]
    =
    \sum_{e \in L \setminus D} z_e \mathbb{E}[\mathbbm{1}_M(e)] =
    (1-\eps)\sum_{e \in L \setminus D} x_e z_e
    =
    (1-\eps)\,Z_x .
    \end{align*}
    
    Applying \Cref{claim:rounding_props} with $\gamma=\eps$, we obtain
    \[
    \Pr\!\left[z(M) \le (1-\eps)\mu_z\right]
    \;\le\;
    \exp\!\left(-\frac{\mu_z \eps^3}{20}\right).
    \]
    Substituting $\mu_z=(1-\eps)Z_x$, we get
    \[
    \Pr\!\left[z(M) \le (1-\eps)\cdot(1-\eps)Z_x\right]
    =
    \Pr\!\left[z(M) \le (1-\eps)^2 Z_x\right]
    \;\le\;
    \exp\!\left(-\frac{(1-\eps)Z_x \eps^3}{20}\right).
    \]
    Therefore, with probability at least $1-\exp\!\left(-\frac{(1-\eps)Z_x \eps^3}{20}\right)$, it holds that
    \begin{equation}
    \label{eq:z_rounding_bound}
    z(M) \;\ge\; (1-\eps)^2 Z_x.
    \end{equation}
    
    Take $z=p'$. Since $x$ is feasible for $\textnormal{LP}(D,R)$, it satisfies
    \[
    Z_x = \sum_{e \in L \setminus D} p'_e x_e
    \ge (1-6\eps)R - u_p(D,\al).
    \]
    Combining this with \eqref{eq:z_rounding_bound} yields that, with probability at
    least $1-\exp\!\left(-\frac{(1-\eps)Z_x \eps^3}{20}\right)$,
    \[
    u_p(M,\al)
    =
    \sum_{e \in M} p'_e
    \ge
    (1-\eps)^2 \left((1-6\eps)R - u_p(D,\al)\right).
    \]
    
    Now take $z=q$.
    By \Cref{claim:matching-LP}, the set $S\setminus H$ is feasible for
    $\textnormal{LP}(D,R)$, and thus the optimal solution $x$ satisfies
    \[
    Z_x = \sum_{e \in L \setminus D} x_e q_e \ge q(S\setminus H).
    \]
    Combining this with \eqref{eq:z_rounding_bound} yields that, with probability at
    least $1-\exp\!\left(-\frac{(1-\eps)Z_x \eps^3}{20}\right)$,
    \[
    u_a(M,\al)
    =
    q(M)
    \ge
    (1-\eps)^2 \, q(S\setminus H)
    =
    (1-\eps)^2 \, u_a(S\setminus H,\al).
    \]
    
    Repeat the rounding independently $O(\log n)$ times and take the best feasible
    outcome. A standard Chernoff-style amplification and a union bound over
    $z\in\{p',q\}$ yield that both inequalities hold with high probability.
\end{proof}

We now present \cref{alg:innerLPMatching} and prove \cref{lemma:InnerSolveLP} for the matching constraint.
\algoInnerLPMatching

\begin{proof}[Proof of \cref{lemma:InnerSolveLP}]
     We focus on the iteration in which the algorithm correctly guesses $D = S \cap H$. Let $K = D \cup M$ denote the final solution returned by the algorithm. It is straightforward to verify that $K$ is a feasible matching in $G$, as $M$ is a feasible matching in the subgraph $G_D$ and $D \cap M = \emptyset$ by construction. With high probability, and by applying \cref{claim:utilitiesMatchingProps}, we obtain the following bound on the agent's utility:
    \begin{align*}
        u_a(K, \al) & = u_a(D, \al) + u_a(M, \al) \geq u_a(S \cap H, \al) + (1-\eps)^2 \cdot u_a(S \setminus H, \al) \\
        & \geq (1-\eps)^2 \cdot u_a(S, \al) \geq (1-\eps)^2 \cdot (1-6\eps) \cdot u_a(S_\al, \al) \geq (1-14\eps) \cdot u_a(S_\al, \al).
    \end{align*}
    Now, for the principal’s utility, observe that for \emph{any guess} $D$. it holds that:
    \begin{align*}
        u_p(K, \al) & = u_p(D, \al) + u_p(M, \al) \geq u_p(D, \al) + (1 - \eps)^2 \cdot \left((1 - 6\eps) R - u_p(D, \al)\right) \\
        & = (1-(1-\eps)^2) \cdot u_p(D, \al) + (1-\eps)^2 \cdot (1 - 6\eps) R \\
        & \geq (1-\eps)^2 \cdot (1 - 6\eps) R \geq (1-14\eps) R.
    \end{align*}

    Consider now the budget constraint. Let $W_D := W-w(D)$. Since $x$ is feasible for $\textnormal{LP}(D,R)$, it satisfies
    $\sum_{e\in L\setminus D} x_e w_e \le W_D$.
    Setting $\gamma=\eps$ in \Cref{claim:rounding_props}, we have
    \[
    \mu_w := \mathbb E[w(M)]
    = \sum_{e\in L\setminus D} w_e\,\mathbb E[\mathbbm{1}_M(e)]
    = (1-\eps)\sum_{e\in L\setminus D} x_e w_e
    \le (1-\eps)W_D.
    \]
    By Markov's inequality, $\Pr[w(M) > W_D] \le \frac{\mu_w}{W_D} \le 1-\eps$, therefore $\Pr[w(M)\le W_D]\ge \eps$.
    By standard argument, we can obtain a matching $M$ that respects $W_D$ with probability at least $1- \delta$, for $\delta = \frac{1}{poly(n)}$, by repeating the rounding procedure independently $O(\frac{\ln (1/\delta)}{\eps})$ times.
    In the rare event that the budget is violated, we return the empty solution, which is feasible; since this occurs with probability at most $1/poly(n)$, 
    we have the statement of the lemma.

    It remains to show the running time of \cref{alg:innerLPMatching}. For any fixed guess $D$, the linear program $\textnormal{LP}(D,R)$ can be solved in $poly(n)$ time, and the randomized rounding procedure of \cite{chekuri2011multi} also runs in polynomial time. The number of iterations $\rho$ can be bounded by
    \begin{align*}
        \rho & \leq \left( |T| + 1 \right)^{2 \eps^{-1}} \leq \left( 4 \psi(\eps)^3 + 1 \right)^{2 \eps^{-1}} \leq (33 \cdot \eps^{-6})^{2 \eps^{-1}} = 33^{2 \eps^{-1}} \cdot \eps^{-12 \eps^{-1}} = 2^{O \left(\eps^{-1} \cdot \log \left( \frac{1}{\eps} \right) \right)}.
    \end{align*}
    The third inequality follows from the definition of $\psi(\eps)$ and the fact that $0 < \eps < \frac{1}{2}$. This implies the required running time of the algorithm.    
\end{proof}

\subsection{Local Phase: Matroid}
\label{subsec:matroid_LP}

We now prove \cref{lemma:InnerSolveLP} in the context of budgeted matroid. 

Given a representative set $T \se \sa$, we will exhaustively try each possible option of a subset $D \se T$ such that $|D| \leq \frac{2}{\eps}$. Subsequently, we will solve an LP for the remaining actions with small values of $p_i'$ and $q_i$. Formally, denote by 
\begin{equation}
\label{eq:BMICdefinitionE}
    L = \left \{ i \in \sa \mid p_i' \leq 2\eps R \text{ and } q_i \leq 2\eps \beta_q \right \}
\end{equation}
the set of actions with low values of $p_i'$ and $q_i$.
For each guess of $D$, we solve the following LP:
\begin{align}
\textnormal{LP}(D,L,R)~~~~~~~~~~~~~~~~~
\max \quad & \sum_{i \in L \setminus D} q_i \cdot x_i \nonumber \\
\text{subject to\quad} & (1-\al) \sum_{i \in L \setminus D} p_i \cdot x_i \geq (1-6\eps) \cdot R - u_p(D, \al), \label{constraint:BMIC-principal}\\
& \sum_{i \in L \setminus D} w_i \cdot x_i \leq W - w(D), \nonumber \\
& x \in \mathcal{P}_{M_D}. \nonumber
\end{align}
Here, $M_D = (\sa, \si')$ is the matroid (confirmed by Lemma 2.3 in \cite{doron2023eptas}), with the following independent sets: $\si' = \left \{ S \se L \setminus D \mid S \cup D \in \si \right \}$ and $\mathcal{P}_{M_D}$ is the matroid polytope of $M_D$. Notice that the guess $D = S \cap H$ is valid, since for all $i \in S \cap H$, it holds that $p_i' \geq \eps \cdot u_p(S_\al, \al)$ or $q_i \geq \eps \cdot u_a(S_\al, \al)$, thus, $|S \cap H| \leq \frac{2}{\eps}$. The following claims state some useful properties of $\textnormal{LP}(D,L,R)$:

\begin{claim}[Lemma 3.10 from \cite{doron2023eptas}]
\label{claim:polyLP}
    $\textnormal{LP}(D,L,R)$ can be solved in $poly(n)$ time.
\end{claim}

\begin{claim}[Theorem 3 from \cite{GrZe10}]
\label{claim:BMIC-numOfFractional}
    Let $x^*$ be a basic solution for $\textnormal{LP}(D,L,R)$, then $x^*$ has at most $4$ non-integral entries.
\end{claim}

\begin{claim}
\label{claim:matroid-LP}
When $D = S \cap H$ is correctly guessed, the set $S \setminus H$ forms a feasible solution to $\textnormal{LP}(D, R)$.
\end{claim}

The proof of this claim is given together with the proof of \cref{claim:matching-LP} in Appendix~\ref{appx:missing_proofs}.
Now, we present \cref{alg:innerLPMatroid} and the proof of \cref{lemma:InnerSolveLP} for the matroid constraint.

\algoInnerLP

\begin{proof}[Proof of \cref{lemma:InnerSolveLP}]
    We focus on the iteration of \cref{alg:innerLPMatroid} where $D = S \cap H$. Let $x$ be the basic solution obtained by $\textnormal{LP}(D,L,R)$. Denote by $F$ the set of fractionally allocated actions in $x$, i.e., 
    $F = \left \{ i \in L \setminus D \mid 0 < x_i < 1 \right \}$.

    We aim to analyze the utility of the solution $K$ obtained in Step \eqref{step:BMIC-defS} of \cref{alg:innerLPMatroid}. To generate a feasible solution, we round down $x_i = 0$ for each $i \in F$. The resulting solution satisfies all constraints except possibly the principal's utility constraint \eqref{constraint:BMIC-principal}. For the agent's utility it holds that:
    \begin{align*}
        u_a(K, \al) & = u_a(D, \al) + \sum_{i \in L \setminus D \text{ s.t. } x_i=1} q_i \cdot x_i \geq u_a(D, \al) + \sum_{i \in L \setminus D} q_i \cdot x_i - 4 \cdot 2 \eps \cdot u_a(S_\al, \al) \\
        & \geq u_a(S \cap H, \al) + u_a(S \setminus H, \al) - 8 \eps \cdot u_a(S_\al, \al) = u_a(S, \al) - 8 \eps \cdot u_a(S_\al, \al) \\
        & \geq (1-6\eps) \cdot u_a(S_\al, \al) - 8 \eps \cdot u_a(S_\al, \al) = (1-14\eps) \cdot u_a(S_\al, \al),
    \end{align*}
    where the first inequality is due to \cref{claim:BMIC-numOfFractional} and \cref{eq:BMICdefinitionE}. The second inequality is due to \cref{claim:matroid-LP}, and the third inequality is due to \cref{dfn:representative}.
    Similarly, for \emph{each guess} of $D$, we have the following bound on the principal's utility:
    \begin{align*}
        u_p(K, \al) & = (1-\al) \cdot \left( p(D) + \sum_{i \in L \setminus D \text{ s.t. } x_i=1} p_i \cdot x_i \right) \geq (1-\al) \cdot \left ( p(D) + \sum_{i \in L \setminus D} p_i \cdot x_i - 8 \eps R \right ) \\
        & \geq (1-\al) \cdot p(D) + (1 - 6\eps) \cdot R - (1-\al) \cdot p(D) - (1-\al) \cdot 8 \eps R \\
        & \geq (1 - 6\eps) \cdot R - 8 \eps  R = (1-14\eps) \cdot R,
    \end{align*}
    where the justification to the first inequality is as before. The second inequality is since $x$ is a feasible solution and hence it satisfies the budget constraint \eqref{constraint:BMIC-principal}. The third inequality is due to the fact that $(1-\al) \cdot 8 \eps R \leq 8 \eps R$ since $\al \in (0,1)$. The running time analysis mirrors that of the matching-constrained case.
\end{proof}

\subsection{Global Phase}

\begin{proof}[Proof of \cref{thm:BMICalgo}]
    \cref{lemma:approxBMI} establishes the local approximation guarantees necessary to apply \cref{lemma:phase2}, which in turn provides the conditions required by \cref{lemma:findAlpha}. Together, these results yield the desired global approximation guarantee.

    Regarding the running time, both \cref{lemma:phase2} and \cref{lemma:findAlpha} make a number of calls to \cref{lemma:approxBMI} that is polynomial in $n$ and $1/\eps$. Since \cref{lemma:approxBMI} runs in $2^{O \left(\eps^{-1} \cdot \log \left( \frac{1}{\eps} \right) \right)} \cdot poly(n)$ time, the overall algorithm is an EPTAS.    
\end{proof}

\section{Discussion}
\label{sec:Discussion}

In this paper, we present a method for translating positive algorithmic results for combinatorial problems to contract design subject to combinatorial constraints without demand queries. Our ``local to global'' framework for a single agent yields a multiplicative and additive FPTAS for a budget constraint which is essentially tight. We apply our framework to additional combinatorial constraints, including multi-budget, budgeted matroid, and budgeted matching, achieving approximation guarantees that match the best-known results for the corresponding algorithmic problems. For the multi-agent setting with combinatorial constraint, we obtain pure multiplicative performance guarantees.

Below, we highlight several intriguing directions for future work:
\begin{itemize}
    \item 
    Can we obtain approximately optimal contracts subject to combinatorial constraint using existing algorithms for the combinatorial problems \emph{as black-box}?
        
    \item 
    Can we characterize the class of combinatorial constraints for which the approximation guarantee for the contract problem (in the single- and multi-agent settings) can be matched to the best known approximation for the combinatorial problem? Alternatively, which combinatorial problems become harder when moving from algorithms to contracts?
    
    \item 
    We have shown the hardness of one-sided approximation in the single-agent setting, already for a single budget constraint. For which types of constraints we can obtain purely multiplicative approximation guarantees?
\end{itemize}

\section*{Acknowledgments}

We gratefully acknowledge feedback from anonymous reviewers that led to this improved version of our work. 
This work was supported by the European Research Council (ERC) under the European Union's Horizon 2020 research and innovation program (grant agreement No.~101077862, project ALGOCONTRACT), by the Israel Science Foundation (grant No.~3331/24), by the NSF-BSF (grant No.~2021680), and by a Google Research Scholar Award.

\appendix
\bibliographystyle{plain}
\bibliography{bibliography}

\section*{Appendix Organization} 
Appendix~\ref{appendix:Hardness} contains our hardness results. 
Missing proofs are collected in Appendix~\ref{appx:missing_proofs}.
Appendices~\ref{sec:MASK} and~\ref{sec:MAMK} address multi-agent settings. 
In Appendix~\ref{sec:SASK} we show an FPTAS for budgeted single-agent settings (improving upon the PTAS for multi-budgeted).

\section{Hardness of Contract Design with Combinatorial Constraints}
\label{appendix:Hardness}

In this appendix we investigate the computational complexity of contracts with added combinatorial constraints.
Our main results for the natural budget constraint (in Theorems~\ref{hard:principalEpsNP} and~\ref{hard:agentEpsNP}) indicate the impossibility of obtaining contracts in which only {\em one} of the objective functions (either the IC-constraint or the principal's utility) is {\em approximated}. These results highlight the necessity  of seeking contracts which approximate both the principal's utility and the IC-constraint.

\begin{theorem}
\label{hard:principalEpsNP}
Let $\eps \in (0,1)$ and $\delta > 0$ be fixed constants. There is no polynomial-time algorithm for $\sask$ that yields a solution $(S ,\al)$ satisfying 
    \begin{equation*}
        u_p(S, \al) \geq (1-\eps) \cdot u_p(S_{\al'}, \al') - \delta \quad \textrm{and} \quad u_a(S, \al) \geq u_a(S_\al, \al)
    \end{equation*}
    for every $\al' \in (0,1)$, unless $\pnp$.
\end{theorem}

\begin{theorem}
\label{hard:agentEpsNP}
    Let $\eps \in (0,1)$ be a fixed constant. There is no polynomial-time algorithm for $\sask$ that yields a solution $(S ,\al)$ satisfying 
    \begin{equation*}
        u_a(S, \al) \geq (1-\eps) \cdot u_a(S_{\al}, \al) \quad \textrm{and} \quad u_p(S, \al) \geq u_p(S_{\al'}, \al')
    \end{equation*}
    for every $\al' \in (0,1)$, unless $\pnp$.
\end{theorem}

The next result follows immediately from Theorem \ref{hard:principalEpsNP}.

\begin{corollary}
\label{hard:regularNP}
The budgeted single-agent ($\sask$) problem is NP-hard.
\end{corollary}

We note that the NP-hardness of budgeted multi-agent ($\mask$) follows from the NP-hardness of the multi-agent contract problem (with no constraints), already for additive reward function~\cite{dutting2023multi}.

The next result highlights the computational complexity of the $\samk$, $\mamk$, and $\bmic$ problems, indicating that algorithms which improve the best known approximation guarantees for the corresponding problems without the IC-constraint are unlikely to exist unless $\pnp$.

\begin{theorem*}
    There is no EPTAS for the multi-budgeted single-agent ($\samk$) or the multi-budgeted multi-agent ($\mamk$) problems, and no FPTAS for the budgeted matroid ($\bmic$) problem, unless $\pnp$.
\end{theorem*}

In what follows, Appendices \ref{sub:util_p_approx} and \ref{sub:util_a_approx} demonstrate the impossibility of approximating either the principal's or the agent's utility independently without simultaneously approximating both, thereby establishing Theorems \ref{hard:agentEpsNP} and \ref{hard:principalEpsNP}. In \cref{sub:multi-dim-hardness}, we prove that achieving an EPTAS for the multi-budgeted versions, namely the $\samk$ and $\mamk$ problems, is NP-hard. \cref{sub:BMIC-hardness} further establishes the intractability of obtaining an FPTAS for the $\bmic$ problem.

\subsection{Hardness of One-Sided Approximation: Principal}
\label{sub:util_p_approx}

\begin{proof}[Proof of \cref{hard:principalEpsNP}]
    Assume towards a contradiction that there exists such algorithm $\alg$. Recall that knapsack problem is defined by $n$ items, each with value $v_i$ and weight $w_i$, and a knapsack capacity $W$. The objective is to maximize the total value of the selected items under the knapsack constraint: $\max_{S \se [n]}{\left\{ \sum_{i \in S}{v_i} \mid \sum_{i \in S}{w_i} \leq W \right\}}$.
    
    Let $I$ be an instance of the knapsack problem. Reduce $I$ to an instance of the $\sask$ problem $I'$ as follows. Define $n+1$ actions $\sa = \{a_1, \ldots, a_{n+1}\}$. The first $n$ actions directly correspond to the knapsack items with profits $p_i = v_i$, costs $c_i = 0$, and weights $w_i$ identical to the original items. However, the final action $a_{n+1}$ is defined with profit $p_{n+1} = \frac{2}{1 - \eps} K$, cost $c_{n+1} = \frac{p_{n+1}}{2}$, and weight $w_{n+1} = 0$, where $K = \sum_{i=1}^{n} v_i + \delta$. Since $\epsilon,\delta$ are constant, this reduction can be observed to be polynomial.

    Running $\alg$ on $I'$ produces a solution $(S, \al)$. The structure of $I'$ ensures the following properties:  
    \begin{claim}
    \label{claim:hardAlphaProp12}
        There exists an optimal solution $S^*$ to $I$ for which the following hold:
        \begin{enumerate}
            \item $a_{n+1} \in S$ if and only if $\al \geq \frac{1}{2}$,
            \item For any $\al' \in (0,1)$, if $\al' < \frac{1}{2}$ then $S_{\al'} = S^*$, and if $\al' \geq \frac{1}{2}$ then $S_{\al'} = S^* \cup \{a_{n+1}\}$.
        \end{enumerate}
    \end{claim}
    \begin{proof}
        The first property follows immediately from the agent's utility for $a_{n+1}$. When $\alpha < \frac{1}{2}$ we have $\alpha p_{n+1} - c_{n+1} < \frac{1}{2} p_{n+1} - \frac{1}{2} p_{n+1} = 0$, so a rational agent would never select $a_{n+1}$, implying $a_{n+1} \notin S_\alpha$. Conversely, when $\alpha \geq \frac{1}{2}$, the action $a_{n+1}$ has zero weight and yields nonnegative utility to the agent, so adding it to the solution cannot violate feasibility and only improves the agent’s utility.
    
        For the second property, consider the following: as $w_{n+1}=0$, the inclusion of action $a_{n+1}$ does not impact the budget constraint. Hence, the agent's decision to include $a_{n+1}$ is solely determined by its associated utility, as established before.
        
        Therefore, for the remaining actions, the agent's problem effectively reduces to solving the original knapsack instance. Formally:
        \begin{align*}
            S^* & = \argmax_{S \se [n]}\left\{\sum_{i\in S} v_i \mid \sum_{i\in S} w_i \leq W \right\} = \argmax_{S \se [n]}\left\{\sum_{i\in S} \al' p_i - c_i \mid \sum_{i\in S} w_i \leq W \right\} = S_{\al'},
        \end{align*}
        where the last equality holds since $c_i=0$ and $p_i = v_i$ for all $i \in [n]$. Therefore, when $\alpha' \geq \frac{1}{2}$, the first property guarantees that $a_{n+1} \in S_{\alpha'}$. Combined with the fact that the remaining actions correspond to the optimal solution of the original knapsack instance, we have $S_{\alpha'} = S^{*} \cup \{a_{n+1}\}$. Conversely, when $\alpha' < \frac{1}{2}$, the agent does not select $a_{n+1}$, and the solution reduces to the optimal solution of the original knapsack instance, resulting in $S_{\alpha'} = S^{*}$.
    \end{proof}
    
    \begin{claim}
    \label{claim:alphaHalfMultiplicative}
    The contract $\al$ returned by $\alg$ must satisfy $\al \geq \frac{1}{2}$.
    \end{claim}
    \begin{proof}
        Assume, by contradiction, that $\al < \frac{1}{2}$. The approximation guarantee of $\alg$ implies that $u_p(S, \al) \geq (1-\eps) u_p(S_{\al'}, \al') - \delta$ for all $\al' \in (0,1)$. In particular, it holds also for $\al' = \frac{1}{2}$. To derive a contradiction, it suffices to show that $u_p(S, \al) < (1-\eps) u_p(S_{\al'}, \al') - \delta$. Given that $\al <\frac{1}{2}$ and $\al'=\frac{1}{2}$, \cref{claim:hardAlphaProp12} implies that $a_{n+1} \notin S$ and $S_{\al'} = S^* \cup \{a_{n+1}\}$, respectively. Therefore, we have:
        \begin{align*}
            u_p(S, \al) & = (1-\al) \sum_{i \in S}{p_i} < K - \delta = \frac{1-\eps}{2} \cdot \frac{2}{1-\eps} K - \delta = \frac{1-\eps}{2} \cdot p_{n+1} - \delta \\
            & \leq \frac{1-\eps}{2} \left(\sum_{i \in S^*}{p_i} + p_{n+1} \right) - \delta = (1-\eps)(1-\al') \left(\sum_{i \in S^*}{p_i} + p_{n+1} \right) - \delta \\
            & = (1-\eps) \cdot u_p(S^* \cup \{a_{n+1}\}, \al') - \delta = (1-\eps) \cdot u_p(S_{\al'}, \al') - \delta,
        \end{align*}
        where the first inequality is by the definition of $K$.
    \end{proof}

    Resuming the proof, Claims \ref{claim:hardAlphaProp12} and \ref{claim:alphaHalfMultiplicative} establish that $\al \geq \frac{1}{2}$ and $S_\al = S^* \cup \{a_{n+1}\}$, respectively. Since the agent acts optimally, it follows that $S = S_\al = S^* \cup \{a_{n+1}\}$. Therefore, by excluding $a_{n+1}$ from $S$, we obtain the optimal solution $S^*$ to the original knapsack problem $I$, thus concluding the proof.
\end{proof}

\subsection{Hardness of One-Sided Approximation: Agent}
\label{sub:util_a_approx}

\begin{proof}[Proof of \cref{hard:agentEpsNP}]
    Assume towards a contradiction that there exists such algorithm $\alg$. Let $I$ be an instance of the knapsack problem and reduce it to an instance of the $\sask$ problem $I'$ in the following way. Define $n$ actions $\sa = \{a_1, \ldots, a_n\}$, where each action $a_i \in \sa$ has a profit $p_i = v_i$, a cost $c_i = \frac{p_i}{2}$, and a weight $w_i$. The weight limit $W$ remains the same. This reduction can be observed to be polynomial.
    
    Applying algorithm $\alg$ to the instance $I'$ yields a solution $(S, \al)$. The following claim characterizes the properties of an optimal solution to $I'$.
    \begin{claim}
    \label{claim:propsHalfReduction}
        Let $S^*$ be an optimal solution to the original knapsack problem $I$, and let $(S_{\al'}, \al')$ be an optimal solution to $I'$. Then, $\al'=\frac{1}{2}$ and $S_{\al'} = S^*$.
    \end{claim}
    \begin{proof}
        First of all, in order to ensure the agent's participation, the contract $\al'$ must satisfy $\al' \geq \frac{1}{2}$ to prevent negative utility for the agent. Otherwise, no action would be chosen, resulting in zero utility for the principal.

        Note that for any contract $\tilde{\al} \geq \frac{1}{2}$, the agent's optimal response, $S_{\tilde{\al}}$, is independent of the specific value of the contract and equals $S^*$. This can be shown as follows:
        \begin{align*}
            S_{\tilde{\al}} & = \argmax_{S \in \sa} \left \{ {\sum_{i \in S}{\tilde{\al} p_i - c_i}} \mid \sum_{i \in S}{w_i} \leq W \right \} = \argmax_{S \in \sa} \left \{ {\sum_{i \in S}{\tilde{\al} p_i - \frac{p_i}{2}}} \mid \sum_{i \in S}{w_i} \leq W \right \} \\
            & = \argmax_{S \in \sa} \left \{ {\left( \tilde{\al}-\frac{1}{2} \right) \sum_{i \in S}{p_i}} \mid \sum_{i \in S}{w_i} \leq W \right \} = \argmax_{S \in \sa} \left \{ {\sum_{i \in S}{v_i}} \mid \sum_{i \in S}{w_i} \leq W \right \} = S^*,
        \end{align*}
        where the last equality holds assuming $\tilde{\al} > \frac{1}{2}$ and the tie-breaking rule.

        Assume, by contradiction, that $\al' > \frac{1}{2}$. Since we have established that $S_{\al'} = S^*$, setting $\tilde{\al} = \frac{1}{2}$ yields a higher principal utility:
        \begin{equation*}
            u_p(S_{\tilde{\al}}, \tilde{\al}) = (1-\tilde{\al}) \sum_{i \in S_{\tilde{\al}}}{p_i} = (1-\tilde{\al}) \sum_{i \in S^*}{p_i} > (1-\al') \sum_{i \in S^*}{p_i} = (1-\al') \sum_{i \in S_{\al'}}{p_i} = u_p(S_{\al'}, \al').
        \end{equation*}
        This contradicts the optimality of $(S_{\al'}, \al')$ and completes the proof.
    \end{proof}
    
    Resuming the proof of \cref{hard:agentEpsNP}, \cref{claim:propsHalfReduction} establishes that $\al'=\frac{1}{2}$ and $S_{\al'} = S^*$. Since $(S,\al)$ maximizes the principal's utility, it follows that $\al=\al'$ and $S=S^*$. Consequently, we have identified an optimal solution to the original knapsack instance $I$, contradicting the fact that the problem is NP-hard. This completes the proof.
\end{proof}

\subsection{Multi-Budgeted Variants: Ruling out an EPTAS}
\label{sub:multi-dim-hardness}

We proceed by establishing the hardness of approximation for the multi-budgeted variants ($\samk$ and $\mamk$ problems). This will be achieved by demonstrating a reduction from the multi-dimensional knapsack problem, which is known to have no EPTAS even for the two-dimensional case \cite{kulik2010there}, unless $\pnp$.

\begin{prop}
\label{hard:SAMK-EPTAS}
    There is no EPTAS for the $\samk$ problem, unless $\pnp$.
\end{prop}
\begin{proof}
    Suppose to the contrary that there that there exists an EPTAS $\alg$ for the $\samk$ problem. Let $I$ an instance of the multi-dimensional knapsack problem with a constant weight dimension $d \geq 2$. Reduce it to an instance $I'$ of the $\samk$ problem by adapting the construction from \cref{hard:agentEpsNP} to accommodate the $d$-dimensional weight vector.

    Let $(S, \al)$ be the solution obtained by applying $\alg$ to $I'$, and let $(S^*, \al^*)$ be an optimal solution to the original instance $I$. By \cref{claim:propsHalfReduction}, we have $\al^* = \frac{1}{2}$ and $S^*$ is also an optimal solution to $I$. The following claim establishes the approximation ratio of the solution $(S, \al)$.

    \begin{claim}
    \label{claim:approxReductionEPTAS}
        It holds that: $\sum_{i \in S}{v_i} \geq (1-\eps) \sum_{i \in S^*}{v_i}$.
    \end{claim}
    \begin{proof}
        It holds that:
        \begin{equation*}
        (1-\al) \sum_{i \in S}{v_i} = u_p(S, \al) \geq (1-\eps) u_p(S^*, \al^*) = (1-\eps) (1-\al^*) \sum_{i \in S^*}{v_i}.
        \end{equation*}
        Since $\al^* = \frac{1}{2}$, we obtain:
        \begin{equation*}
            (1-\al) \sum_{i \in S}{v_i} \geq \frac{(1-\eps)}{2} \sum_{i \in S^*}{v_i}.
        \end{equation*}
        As shown in \cref{claim:propsHalfReduction}, $\al \geq \frac{1}{2}$ in order to ensure non-negative utility for the agent. This implies $1 \geq 2 (1-\al)$. Therefore, we have:
        \begin{equation*}
            \sum_{i \in S}{v_i} \geq 2(1-\al) \sum_{i \in S}{v_i} \geq (1-\eps) \sum_{i \in S^*}{v_i},
        \end{equation*}
        as required.
    \end{proof}

    Note that the running time of $\alg$ is $h \left ( \frac{1}{\eps} \right ) \cdot poly(n)$ time for some arbitrary function $h:\dr_+ \rightarrow \dr_+$. Therefore, by \cref{claim:approxReductionEPTAS}, we have developed an EPTAS for the multi-dimensional knapsack problem, which contradicts its NP-hardness.  
\end{proof}

\begin{prop}
\label{hard:MAMK-EPTAS}
    There is no EPTAS for the $\mamk$ problem, unless $\pnp$.
\end{prop}
\begin{proof}
    Suppose to the contrary that there that there exists an EPTAS $\alg$ for the $\mamk$ problem. Let $\eps>0$, and consider an instance $I$ of the multi-dimensional knapsack problem with a constant weight dimension $d \geq 2$.

    We construct a corresponding instance $I'$ of the $\mamk$ problem with $n$ agents, denoted by $\sa = \{a_1, \ldots, a_n\}$ with the following properties. Each agent $a_i$ has a profit $f(\{i\})=p_i=v_i$, and a cost $c_i = \frac{\eps' p_i}{n}$, where $\eps'$ is an error parameter satisfying $\eps' \leq 2\eps - \eps^2$. The weight vector $\Vec{w}_i$ for each agent $a_i$ is identical to the weight vector of item $i$ in $I$. The weight limit $\Vec{W}$ for $I'$ is also identical to the weight limit in $I$.

    Apply $\alg$ on $I'$ using $\eps'$ as the error bound, and denote the resulting solution by $(S, \al)$. By the approximation guarantee of $\alg$, it holds that:
    \begin{equation}
    \label{eq:mmFPTASreduction}
        \left( 1-\sum_{i \in S}{\frac{c_i}{p_i}} \right) \sum_{i \in S}{p_i} \geq  (1-\eps') \left( \left(1-\sum_{i \in S'}{\frac{c_i}{p_i}} \right) \sum_{i \in S'}{p_i} \right) \quad \forall S' \se \sa.
    \end{equation}
    
    We focus on each side of \cref{eq:mmFPTASreduction} separately. For the left-hand side, it holds that:
    \begin{equation}
    \label{eq:mmFPTASleft}
        \sum_{i \in S}{v_i} = \sum_{i \in S}{p_i} \geq \left( 1-\sum_{i \in S}{\frac{c_i}{p_i}} \right) \sum_{i \in S}{p_i}.
    \end{equation}
    For the right-hand side, by definition of $c_i$ and $\eps'$, and since $|S'| \leq n$, we have:
    \begin{align*}
        & (1-\eps') \left( \left(1-\sum_{i \in S'}{\frac{c_i}{p_i}} \right) \sum_{i \in S'}{p_i} \right) = (1-\eps') \left( \left(1-\sum_{i \in S'}{\frac{\eps'}{n}} \right) \sum_{i \in S'}{p_i} \right) \\
        &= (1-\eps') \left( \left(1-\frac{\eps' |S'|}{n} \right) \sum_{i \in S'}{p_i} \right) \geq (1-\eps') \left( \left(1-\frac{\eps' n}{n} \right) \sum_{i \in S'}{p_i} \right) \\
        &= (1-\eps')^2 \sum_{i \in S'}{p_i} \geq (1-\eps) \sum_{i \in S'}{v_i}.
    \end{align*}
    Combining it with Eq. \eqref{eq:mmFPTASreduction} and \eqref{eq:mmFPTASleft}, we get:
    \begin{equation}
    \label{eq:mmFPTASresult}
        \sum_{i \in S}{v_i} \geq (1-\eps) \sum_{i \in S'}{v_i}.
    \end{equation}
    Finally, \cref{eq:mmFPTASresult}  holds for all $S' \se \sa$, it specifically holds for $S' = S^*$. Furthermore, the runtime of $\alg$ is polynomial in $n$ and $\frac{1}{\eps'}$, and consequently in $\frac{1}{\eps}$. This implies that we have developed an EPTAS for the multi-budgeted problem, contradicting the fact that it is NP-hard to do so.
\end{proof}

\subsection{Budgeted Matroid Variant: Ruling out an FPTAS}
\label{sub:BMIC-hardness}

We establish the NP-hardness of achieving an FPTAS for the $\bmic$ problem by a reduction from the budgeted matroid problem, for which it is known that obtaining an FPTAS is NP-hard, as shown in \cite{doron2024lower}.

\begin{prop}
\label{hard:BMIC-FPTAS}
    There is no FPTAS for the $\bmic$ problem, unless $\pnp$.
\end{prop}
\begin{proof}
    Suppose to the contrary that there that there exists an FPTAS $\alg$ for the $\bmic$ problem. Let $I$ an instance of the budgeted matroid problem. Reduce it to an instance $I'$ of the $\bmic$ problem by adapting the construction from \cref{hard:agentEpsNP}.

    Let $(S, \al)$ be the solution obtained by applying $\alg$ to $I'$, and let $(S^*, \al^*)$ be an optimal solution to the original instance $I$. By \cref{claim:propsHalfReduction}, we have $\al^* = \frac{1}{2}$ and $S^*$ is also an optimal solution to $I$.

    Note that \cref{claim:approxReductionEPTAS} holds in this case as well, establishing the approximation ratio of the solution $(S, \al)$. Finally, the running time of $\alg$ is polynomial in $n$ and $\frac{1}{\eps}$. Therefore, we have developed an FPTAS for the budgeted matroid problem, which contradicts its NP-hardness.
\end{proof}

\subsection{Additive Approximation}
\label{appx:add-approx}

We adopt multiplicative approximation guarantees for both the principal and the agent. This choice aligns with the conventions of approximation algorithms and fits naturally with our model: utilities scale linearly with rewards, so multiplicative guarantees capture the relevant structure without depending on arbitrary absolute units.
We note that,
without normalization, additive guarantees are intractable: any additive $\eps$-approximation would imply an additive approximation for a knapsack instance, which is NP-hard~\cite[Theorem 2.5.2]{kellerer2004basic}. 

Moreover, the reductions from additive $\eps$-IC to optimal contracts in \cite{dutting2021complexity,zuo2024new} do not apply here, as they require computing the agent’s demand set under a contract—a subroutine that is NP-hard in our model. More strongly, \cref{hard:principalEpsNP} rules out any workaround: when the agent best-responds optimally, even an asymptotic (i.e., simultaneously multiplicative and additive) approximation to the principal’s utility is NP-hard.

In the regime of bounded rewards that can be normalized to $1$, additive guarantees are straightforward: if the underlying algorithmic problem admits an additive approximation, the same follows for the contract problem.

\begin{prop}
\label{prop:additiveAlgorithm}
    Let $\Pi$ be an optimization problem subject to a constraint~$\mathcal{C}$, 
    and $\alg$ an algorithm that, for any $\eps>0$, computes an $\eps$-additive approximation for $\Pi$ in time polynomial in $\eps$ and in the input size.
    Then, for the corresponding normalized contract-design problem under the same constraint~$\mathcal{C}$, one can compute a solution $(S,\alpha)$ satisfying $u_p(S,\al) \geq u_p(S_{\al^*},\al^*) - \eps$ and $u_a(S,\al) \geq u_a(S_{\al},\al) - 2 \eps$, in $poly(n,\eps)$ time.
\end{prop}

\begin{proof}
    Let $\eps > 0$ and consider the contract $\alpha^* + \eps$.
    Applying $\alg$ with error bound $\eps^2$ to the agent’s problem under this contract yields a set $S$ such that
    \begin{equation}
    \label{eq:lbound_additive_ua}
        u_a(S,\al^*+\eps) \geq u_a(S_{\al^*+\eps},\al^*+\eps) - \eps^2 \geq u_a(S_{\al^*},\al^*+\eps) - \eps^2,
    \end{equation}
    where the second inequality holds since $S_{\al^*+\eps}$ is optimal for $\al^*+\eps$. By expanding the utilities,
    \begin{align*}
        (\alpha^* + \eps)\,f(S) -c(S) \geq (\alpha^* + \eps)\,f(S_{\alpha^*}) - c(S_{\alpha^*}) - \eps^2,
    \end{align*}
    or equivalently,
    \begin{align*}
        \alpha^* f(S) + \eps f(S) -c(S) \geq \alpha^* f(S_{\alpha^*}) + \eps f(S_{\alpha^*}) - c(S_{\alpha^*}) - \eps^2.
    \end{align*}
    Since $S_{\alpha^*}$ is optimal for the contract $\alpha^*$, we have
    \begin{align*}
        \alpha^* f(S) - c(S) \leq \alpha^* f(S_{\alpha^*}) - c(S_{\alpha^*}).
    \end{align*}
    Subtracting the latter inequality from the former and dividing by $\eps > 0$ yields
    \begin{align*}
        f(S) \geq f(S_{\alpha^*}) - \eps.
    \end{align*}
    Multiplying both sides by $(1 - \alpha^*)$ and using $0 < 1 - \alpha^* < 1$, we obtain
    \begin{align*}
        u_p(S, \al^*) = (1-\al^*) f(S) \geq (1-\al^*) f(S_{\al^*}) - (1-\al^*) \eps \geq u_p(S_{\al^*}, \al^*) - \eps.
    \end{align*}
    This establishes the desired guarantee for the principal’s utility. To derive the corresponding bound for the agent, note that
    \begin{equation}
    \label{eq:ubound_ua}
        u_a(S,\al^*+\eps) = \alpha^* f(S) + \eps f(S) - c(S) = u_a(S,\al^*) + \eps f(S) \leq u_a(S,\al^*) + \eps,
    \end{equation}
    where the inequality follows from the normalization $f(S) \le 1$.
    Hence, combining (\ref{eq:ubound_ua}) with (\ref{eq:lbound_additive_ua}), we have that
    \[
    \begin{array}{ll}
    u_a(S, \alpha^*) & \ge u_a(S, \alpha^* + \eps) - \eps \\
    & \geq u_a(S_{\alpha^*}, \alpha^* + \eps) - \eps^2 - \eps \\
    &= \alpha^* f(S_{\alpha^*}) + \eps f(S_{\alpha^*}) -c(S_{\alpha^*}) - \eps^2 - \eps \\
    & \geq u_a(S_{\alpha^*}, \alpha^*) -2 \eps.
    \end{array}
    \]
    
    Finally, while $\alpha^*$ is not directly known, it can be efficiently approximated
    by \cref{claim:alphaStarRange}, one can find a contract $\alpha$ within a $(1 - \eps)$ multiplicative range of $\alpha^*$.
    If rewards are normalized, then any algorithm achieving a multiplicative $(1-\eps)$-approximation to the agent’s utility also achieves an additive $\eps$-approximation. Thus,
    this multiplicative proximity translates into an additive $\eps$ approximation, yielding the desired guarantees for both the principal and the agent.
\end{proof}

\section{Missing Proofs and Algorithms}
\label{appx:missing_proofs}

\subsection{Missing Proofs from \cref{sec:introduction}}
\label{appx:XOS-proof}

Notation: Given a set function $f$ and an item $i$, we use the notation $f(i):=f(\{i\})$. 

\begin{proof}[Proof of \cref{obs:XOS}]
For a base set of $n$ items, a given budget $W > 0$, and additive functions $w\ge 0, p\ge 0$, let $f(S):={\max}_{\{T\subseteq S : w(T)\le W \}}p(T)$. We show that $f$ is monotone XOS, but not submodular.

\paragraph{\textbf{$f$ is monotone}.}

Let $S\subseteq S'$. Let $T^*$ be a maximizer in the definition of $f$, i.e., $f(S)=p(T^*)$, and $T^*$ satisfies the budget constraint. By definition of $f$ as a maximization over such subsets, $f(S')\ge p(T^*)\ge f(S)$, as required.

\paragraph{\textbf{$f$ is XOS}.}
For every $T\subseteq [n]$, define an additive function $a_T$ as follows: If $w(T)\le W$ (i.e., $T$ fits in the knapsack), let $a_T(i)=p(i)$ for every $i\in T$, and let $a_T(i')=0$ for every $i'\notin T$. Otherwise, if $w(T)> W$, let $a_T(i)=0$ for every $i\in [n]$. Define an XOS function $g$ from the collection of additive functions $\mathcal{A}=\{a_T\}_{T\subseteq [n]}$, i.e., $g(S)=\max_{a_T\in\mathcal{A}}\{a_T(S)\}$ for every subset $S$. We now show that $f(S)=g(S)$ for every $S$, implying that $f$ is XOS.

We first observe that for every subset $T$, by definition of $a_T$ we have 
\begin{equation}
    a_T(S)=
    \begin{cases}
        p(S\cap T) & \text{if } w(T)\le W,\\
        0 & \text{otherwise.}
    \end{cases}
    \label{eq:add-func-in-XOS}
\end{equation}
By definition, $f(S)={\max}_{\{T\subseteq S : w(T)\le W \}}p(T)$. 
Notice we can write this also as 
\begin{equation}
    f(S)={\max}_{\{T\subseteq S : w(T)\le W \}}p(S\cap T).\label{eq:f-alternative-def}
\end{equation}
Let $T^*$ be the subset of $S$ achieving the maximum, i.e., 
\begin{equation}
    f(S)=p(S\cap T^*).\label{eq:role-of-T-star} 
\end{equation}
Since $T^*$ fits within the budget, by \eqref{eq:add-func-in-XOS}-\eqref{eq:role-of-T-star} we have $a_{T^*}(S)=p(S\cap T^*)=f(S)$. 

It remains to show that $T^*$ is the subset that achieves the maximum in the definition of $g(S)$, i.e., that $g(S)=a_{T^*}(S)$. Clearly $a_{T^*}(S)$ is not strictly exceeded by $a_T(S)$ for any $T$ that doesn't fit within the budget, since then $a_T(S)=0$. Assume for contradiction that $T$ fits and also $p(S\cap T)>p(S\cap T^*)$, i.e., the maximum in the definition of $g(S)$ is obtained by $T$, not $T^*$. Denote by $T':=S\cap T$ the subset of items in $T$ that are also in $S$; we know that $T'$ fits within the budget just like $T$, and $p(S\cap T')=p(S\cap T)>p(S\cap T^*)$. Thus we get, by the definition of $f$ in \eqref{eq:f-alternative-def} and the fact that $T'$ is a subset of $S$ that fits within the budget, that $f(S)\ge p(S\cap T')>p(S\cap T^*)$, in contradiction to \eqref{eq:role-of-T-star}. This completes the proof that $f$ belongs to XOS.

\paragraph{\textbf{$f$ is not submodular}.} 
We show by example that with $f$, items may exhibit non-diminishing marginal values. Consider a base set with $n=3$ items $\{1,2,3\}$, let $w$ be defined by item weights $(1.5,1,1)$, and let $p$ be defined by item values $(2-\epsilon,1,1)$ for some small $0<\epsilon\ll 1/2$. Let the budget $W$ be $2$. 
\begin{itemize}
    \item Observe that the marginal value of item $3$ given the set $\{1\}$ is $f(3\mid \{1\})=f(\{1,3\})-f(1)=f(1)-f(1)=0$, where the key step uses that $f(\{1,3\})=f(1)$ since the most valuable subset of $\{1,3\}$ that fits within the budget is $\{1\}$. 
    \item By symmetry, it also holds that $f(\{1,2\})=f(1)$.
    \item However, the marginal value of item $3$ given the set $\{1,2\}$ is $f(3\mid \{1,2\})=f(\{1,2,3\})-f(\{1,2\})=f(\{2,3\})-f(1)=p(2)+p(3)-p(1)=2-(2-\epsilon)=\epsilon$, where the key step uses that $f(\{1,2,3\})=f(\{2,3\})$ since the most valuable subset of $\{1,2,3\}$ that fits within the budget is $\{2,3\}$ (notice that $w(2)+w(3)=2$). 
\end{itemize}
\end{proof}

\subsection{Missing Proofs from \cref{sec:global_phase}}

\begin{proof}[Proof of \cref{claim:alphaStarRange}]
    Let $j^* \in \argmax_{j \in S_{\alpha^*}} c_j$ be the action with maximum cost under the optimal contract $\alpha^*$. As we do not know $\alpha^*$ or $S_{\alpha^*}$, we iterate over all $j \in \sa$ and use the following procedure for each candidate $j^*$.

    Let $S_1 = S_{\alpha=1} = \argmax_{S \subseteq \sa} u_a(S,1)$ be the agent's best-response action set when the contract is $\alpha = 1$. If $\alpha^* \in \{0,1\}$, then we can simply test these values directly. If $S_{\alpha^*} = \emptyset$, then any contract is trivially optimal and we are done. Thus, we focus on the non-trivial case where $\alpha^* \in (0,1)$ and $S_{\alpha^*} \neq \emptyset$.

    Since $S_1$ maximizes the agent's utility under the contract $\alpha = 1$, and we are assuming that $\alpha^* < 1$, it follows that
    \begin{equation*}
        u_a(S_1,1) \geq u_a(S_{\al^*}, 1) \geq u_p(S_{\al^*}, \al^*) > 0,
    \end{equation*}
    where the second inequality holds because the agent receives non-negative utility under $\alpha^*$, i.e., $\alpha^* f(S_{\alpha^*}) \geq c(S_{\alpha^*})$. The strict inequality follows from $\alpha^* < 1$ and $S_{\alpha^*} \neq \emptyset$.

    We now make use of the following result from \cite{duetting2025multi}, restated here using our notation:
    \begin{claim}[Lemma 4.1 in \cite{duetting2025multi}]
    \label{claim:alphaHelper}
    Suppose $u_a(S_1,1) > 0$, and let $j^* \in \argmax_{j \in S_{\alpha^*}} c_j$. Then $\al_{\min} \leq \al^* \leq \al_{\max}$, where 
    \begin{equation*}
        \alpha_{\min} = 1 - \frac{u_a(S_1,1)}{c_{j^*} + u_a(S_1,1)} \quad \text{and} \quad \alpha_{\max} = 1 - \frac{u_a(S_1,1)}{n \cdot 2^n (c_{j^*} + u_a(S_1,1))}.
    \end{equation*}
    \end{claim}
    We now construct a candidate set of contract values:
    \begin{equation*}
        C = \{0,1 \} \cup \left\{1-(1-\eps)^{k+1} \cdot \frac{u_a(S_1,1)}{c_{j^*} + u_a(S_1,1)} \mid k \in \left\{0,1,\ldots,\left\lceil\log_{\frac{1}{1-\eps}}n \cdot 2^n\right\rceil \right\} \right\}.
    \end{equation*}
    Since the size of $C$ is $\mathrm{poly}(n, 1/\eps)$, we can exhaustively evaluate each $\alpha \in C$ in polynomial time. We now argue that one of the $\alpha$ values in $C$ satisfies the desired approximation bounds. Specifically, for $k = 0$, we obtain:
    \begin{equation*}
        \frac{1-\al}{1-\eps} = 1-\al_{\min} \geq 1-\al^*,    
    \end{equation*}
    where the inequality follows from \cref{claim:alphaHelper}. Similarly, for the largest value of $k$, we get:
    \begin{equation*}
        1-\al^* \geq 1-\al_{\max} \geq 1-\al,
    \end{equation*}
    where the first inequality again uses \cref{claim:alphaHelper}, and the second follows by the construction of $C$. Thus, there exists some value $\alpha \in C$ that satisfies the desired approximation guarantees.
\end{proof}

\subsection{Missing Proofs from Section~\ref{sec:SAMK}}
\label{appx:missing_proof_from_mask}

\begin{proof}[Proof of \Cref{claim1:SAD}]
    Algorithm \ref{alg:phase1} guarantees a feasible solution as it returns an integral solution by rounding down fractional allocations, thereby preventing violation of the budget constraint. The algorithm's polynomial-time complexity is due to the enumeration of $O(n^{2h}) = O(n^{\frac{2d+2}{\eps}})$ possible combinations of $S_1$ and $S_2$, with each combination requiring the solution of a polynomial-time LP.
\end{proof}

\begin{proof}[Proof of \cref{claim:numOfFractional}]
   Let $n' = n - |S| - |E(S)|$ denote the number of variables in LP \eqref{prob:singleLP} with undetermined values. By renumbering the remaining actions, the LP can be reformulated as follows:
    \begin{align}
    \max_{x \in [0,1]^{n'}} \quad & \sum_{i=1}^{n'} q_i \cdot x_i \nonumber \\
    \textrm{subject to} \quad &  \sum_{i=1}^{n'} w_{i,j} \cdot x_i \leq W_j \quad \text{for all } j \in \{1, \ldots, d\}, \label{constraint:LPKnapsack} \\
    & (1-\alpha) \cdot \sum_{i=1}^{n'} p_i x_i \geq R', \label{constraint:lemmaPrincipalUtility} \\
    & x_i \leq 1 \quad \text{for } i \in \{1,\ldots,n'\}, \label{constraing:upTo1} \\
    & x_i \geq 0 \quad \text{for } i \in \{1,\ldots,n'\}, \nonumber
    \end{align}
    where $R' = R - (1-\al) \cdot \sum_{i \in S}{p_i}$ and $W'_j = W_j - \sum_{i \in S}{w_{i,j}}$ for $j \in \{1, \ldots, d\}$. This LP can be converted to a standard form using $n'+d+1$ slack variables to satisfy the $n'+d+1$ non-trivial constraints.

    Let $y_i \in \{0,1\}$ be a binary indicator denoting whether the slack variable associated with the $i$-th budget constraint \eqref{constraint:LPKnapsack} is strictly positive, where $i \in \{1, \ldots, d\}$. Similarly, let $y \in \{0,1\}$ be a binary indicator for the slack variable associated with the principal utility constraint \eqref{constraint:lemmaPrincipalUtility}. These indicators will be instrumental in bounding the number of fractional variables in the subsequent analysis.

    Recall that the number of strictly positive variables (both original and slack) in any basic feasible solution of the standard form LP is bounded from above by the number of constraints. In our case, with $n'+d+1$ non-trivial constraints, we have:
    \begin{equation*}
    2|F| + (n' - |F|) + \sum_{i=1}^{d} y_i + y \leq n' + d + 1,
    \end{equation*}
    where $2|F|$ represents the number of variables in $F$ and their corresponding slack variables, both of which are strictly positive in constraint \eqref{constraing:upTo1}. The term $n'-|F|$ accounts for the remaining variables not in $F$, where either the variable itself or its corresponding slack variable is strictly positive (but not both) in the same constraint. Simplifying this expression yields:
    \begin{equation*}
        |F| + \sum_{i=1}^{d} y_i + y \leq d + 1, 
    \end{equation*}
    Given that $y, y_i \in \{0,1\}$ for all $i \in \{1, \ldots, d\}$, it follows that $|F| \leq d + 1$, establishing an upper bound of $d+1$ on the number of fractional variables in the optimal basic solution, as required.
\end{proof}

\subsection{Missing Proofs from Section~\ref{sec:BMIC}}

\begin{proof}[Proof of \cref{claim:LinMatroid}]
    Since $|L_p|, |L_q| \leq \eps^{-2}$ (as exceeding this bound would yield a utility greater than that of $S_\alpha$) it follows that:
    \begin{equation*}
        |L| \leq |L_p| + |L_q| \leq 2 \cdot \eps^{-2} = \psi(\eps).
    \end{equation*}
    Moreover, because $L \subseteq S_\alpha$ and $S_\alpha \in \sv$, it follows that $L \in \sv$ as well. This holds either by the hereditary property of matroids or by the fact that any subset of a matching remains a valid matching. Hence, $L \in \sv_{\leq \psi(\eps)}$.
\end{proof}

\begin{proof}[Proof of \cref{claim:findXtoS}]
    We begin the proof by considering each type of constraint separately.

    \paragraph{Matroid constraint.} By the definition of $\Delta_L$ and \cref{dfn:replacement}, it follows that:
    \begin{equation*}
        |\Delta_L| \leq |L \setminus H| + |Z_L| \leq |L \setminus H| + |L \cap H| = |L| \leq |S_\al|.
    \end{equation*}
    By recursively applying the exchange property of the matroid, there exists a set $D \se S_\al \setminus \Delta_L$ such that $\Delta_L \cup D \in \sv$ and $|D| = |S_\al| - |\Delta_L|$.

    Define $X = D \cap Q$. Since $X \cup \Delta_L \se D \cup \Delta_L \in \sv$, the hereditary property implies that $X \cup \Delta_L \in \sv$. Moreover, as $Q \se S_\al$, we obtain:
    \begin{equation*}
        X = D \cap Q \se (S_\al \setminus \Delta_L) \cap Q \se Q \setminus \Delta_L.
    \end{equation*}
    It remains to lower-bound $|X|$. We begin by bounding the size of $D$:
    \begin{equation}
    \label{eq:Dprofit}
        |D| = |S_\al| - |\Delta_L| \geq |S_\al| - |L \setminus H| - |Z_L| = |Q| + |L \cap H| - |Z_L|,
    \end{equation}
    where the last equality uses the decomposition $S_\al = (L \cap H) \cupdot (L \setminus H) \cupdot Q$.

    Moreover, since $D \se S_\al \setminus \Delta_L \se Q \cup (L \cap H)$, every element of $D$ that is not in $Q$ must belong to $L \cap H$. Hence,
    \begin{equation}
    \label{eq:Xprofit}
        |X| = |D \cap Q|
        = |D| - |D \setminus Q| \geq |D| - |L \cap H| \geq |Q| - |Z_L|,
    \end{equation}
    where the last inequality follows from \eqref{eq:Dprofit}.

    Finally, by \cref{dfn:replacement}, $|Z_L| \leq |L \cap H| \leq 2\eps^{-1}$, therefore we have
    \begin{equation*}
        |X| \geq |Q| - |Z_L| \geq |Q| - 2\eps^{-1}.
    \end{equation*}

    \paragraph{Matching constraint.} In this case, observe that $Q, \Delta_L \in \sv$, and that $Q \cup (L \setminus H) \subseteq S_\alpha \in \sv$. However, the union $Q \cup \Delta_L = Q \cup \left( (L \setminus H) \cup Z_L \right)$ is not guaranteed to be a feasible matching, as adding the edges in $Z_L$ may introduce vertex conflicts with edges already in $Q$. To restore feasibility, we construct a set $B \subseteq Q$ of conflicting edges to be removed.

    Observe that each edge $e = (u, v) \in Z_L$ can share endpoints with at most two edges in $Q$—one incident to $u$ and one to $v$. Therefore, to maintain feasibility when adding all edges in $Z_L$, it suffices to remove at most $2|Z_L|$ edges from $Q$. Thus, $|B| \leq 2 |Z_L|$. By \cref{dfn:replacement}, $|Z_L| \leq |L \cap H|$, and since $|H| \leq 2\eps^{-1}$—to avoid surpassing the utility of $S_\alpha$—we have:
    \begin{equation}
    \label{eq:Bsize}
        |B| \leq 2 |Z_L| \leq 2|L \cap H| \leq 4 \eps^{-1}.
    \end{equation}
    
    Finally, we define $X = Q \setminus (B \cup Z_L)$, noting that we remove $Z_L$ as well since some edges may belong to both $Q$ and $Z_L$. Observe that $X \subseteq Q \setminus \Delta_L$, and by construction, $X \cup \Delta_L \in \sv$. It remains to establish a lower bound on the size of $X$. We have
    \begin{align*}        
        |X| & = |Q \setminus (B \cup Z_L)| \geq |Q| - |B| - |Z_L| \geq |Q| - 4 \eps^{-1} - \eps^{-1} = |Q| - 5 \eps^{-1},
    \end{align*}
    where the second inequality follows from \eqref{eq:Bsize}.
\end{proof}

\begin{proof}[Proof of \cref{claim:SisFeasible}]
    By the definition of $X$ it holds that $S = \Delta_L \cup X \in \sv$. Moreover, 	
    \begin{align*}
        w(S) &= w(\Delta_L \cup X) \leq w(L \setminus H) + w(Z_L) + w(X) \\
        & \leq w(L \setminus H) + w(L \cap H) + w(X) \leq w(L) + w(Q) \leq w(S_\al) \leq W.
    \end{align*}
    The second inequality holds since, by \cref{dfn:replacement}, $w(Z_L) \leq w(L \cap H)$. For the third inequality, recall that $X \subseteq Q \setminus \Delta_L$. Lastly, the final inequality is justified by the feasibility of the solution $S_\al$.

    It remains to show that $S \cap H \se T$. Indeed, by the definition of $Q$, we have that $Q \cap H = \emptyset$. Also, 
    \begin{align*}
        S \cap H & = (\Delta_L \cup X) \cap H = \left( (L \setminus H) \cup Z_L \cup X \right) \cap H = Z_L \cap H \se T.
    \end{align*}
    where the third equality is since $X \se Q \setminus \Delta_L$ and $Q \cap H = \emptyset$ and by the fact that $(L \setminus H) \cap H = \emptyset$. The last transition holds because $Z_L \se T$.
\end{proof}

\begin{proof}[Proof of \cref{claim:thereIsSubstitution}]
    Let $G \in \sv_{\leq \psi(\eps)}$, and let $\pi(G)$ denote all substitutions of $G$. Denote by
    \begin{equation*}
        Z_G = \argmax_{Z \in \pi(G)} \left \{ |Z \cap T| \right \}
    \end{equation*}
    the substitution of $G$ such that $|Z \cap T|$ is maximal among all substitutions $Z$ of $G$. Note that $Z_G$ is well defined since $G \cap H$ is a substitution of $G$, and thus $\pi(G) \neq \emptyset$.

    For convenience, given a set $A$ and an element $x$, we define the following notations:
    \begin{equation*}
        A + x = A \cup \{x\}, \quad A - x = A \setminus \{x\}.
    \end{equation*}

    Now, it remains to show that $Z_G \se T$. Assume to the contrary that there is some $a \in Z_G \setminus T$, hence, by \cref{obs:exactlyOne}, there is exactly one pair $(r, t) \in \sigma$ such that $a \in {\sk}_{r, t}(\beta_p, \beta_q)$. Let $\Delta_G = (G \setminus H) \cup Z_G$. By \cref{dfn:sub}, it holds that $\Delta_G \in \sv_{\leq \psi(\eps)}$. Since $T \cap {\sk}_{r, t}(\beta_p, \beta_q)$ is an exchange set and $a \in (\Delta_G \cap {\sk}_{r, t}(\beta_p, \beta_q)) \setminus (T \cap {\sk}_{r, t}(\beta_p, \beta_q))$, then, by definition, there is $b \in ({\sk}_{r, t}(\beta_p, \beta_q) \cap T) \setminus \Delta_G$ such that $w(b) \leq w(a)$ and $\Delta_G - a + b \in \sv_{\leq \psi(\eps)}$. 
    
    We aim to show that $\Delta_G - a + b$ is also a substitution of $G$. If we can establish this, it will lead to a contradiction to the maximality of $Z_G$ since it holds that:
    \begin{equation*}
        |T \cap (Z_G - a + b)| > |T \cap Z_G| = \max_{Z \in \pi(G)} |T \cap Z|,
    \end{equation*}
    where the inequality holds because $a \in Z_G \setminus T$ and $b \in T$.

    To show that $Z_G-a+b$ is indeed a substitution, we need to verify that it satisfies the properties of \cref{dfn:sub}:
    \begin{enumerate}
        \item By the definition of $b$, it holds that $(G \setminus H) \cup (Z_G - a + b) = \Delta_G - a + b \in \sv_{\leq \psi(\eps)}$.

        \item Since $w(b) \leq w(a)$, then $w(Z_G - a + b) \leq w(Z_G) \leq w(G \cap H)$.

        \item $\lvert (G \setminus H) \cap (Z_G - a + b) \rvert \leq \lvert (G \setminus H) \cap Z_G \rvert = 0$, implying that $(G \setminus H) \cap (Z_G - a + b) = \emptyset$. The inequality holds since $b \notin \Delta_G$, and the equality follows from the fact that $Z_G$ is a substitution of $G$.

        \item For all $(r', t') \in \sigma$, it holds that:
        \begin{equation*}
            \lvert \mathcal{K}_{r',t'}(\alpha) \cap (Z_G - a + b) \rvert = \lvert \mathcal{K}_{r',t'}(\alpha) \cap Z_G \rvert = \lvert \mathcal{K}_{r',t'}(\alpha) \cap G \cap H \rvert,
        \end{equation*}
        The first equality holds because, by \cref{obs:exactlyOne}, $a$ and $b$ are exclusively in $\sk_{r ,t}(\beta_p, \beta_q)$. The last equality follows from the fact that $Z_G$ is a substitution.
    \end{enumerate}
\end{proof}

\begin{proof}[Proof of \Cref{claim:subtitueIsRepset}]
    Let $G \in \sv_{\leq \psi(\eps)}$ and let $Z_G$ be a substitution of $G$ such that $Z_G \se T$. It holds that:
    \begin{equation}
    \label{eq:profitR}
    \begin{aligned}
        z(Z_G) \geq & \sum_{(r, t) \in \sigma} z({\sk}_{r,t}(\beta_p, \beta_q) \cap Z_G) \\
        \geq & \sum_{(r, t) \in \sigma \text{ s.t. } {\sk}_{r,t}(\beta_p, \beta_q) \neq \emptyset} |{\sk}_{r,t}(\beta_p, \beta_q) \cap Z_G| \cdot \min_{i \in {\sk}_{r,t}(\beta_p, \beta_q)} z_i \\
        \geq & \sum_{(r, t) \in \sigma \text{ s.t. } {\sk}_{r,t}(\beta_p, \beta_q) \neq \emptyset} |{\sk}_{r,t}(\beta_p, \beta_q) \cap G \cap H| \cdot (1-\eps) \cdot \max_{i \in {\sk}_{r,t}(\beta_p, \beta_q)} z_i \\
        = & (1-\eps) \cdot \sum_{(r, t) \in \sigma \text{ s.t. } {\sk}_{r,t}(\beta_p, \beta_q) \neq \emptyset} |{\sk}_{r,t}(\beta_p, \beta_q) \cap G \cap H| \cdot \max_{i \in {\sk}_{r,t}(\beta_p, \beta_q)} z_i \\
        \geq & (1-\eps) \cdot z(G \cap H).
    \end{aligned}
    \end{equation}
    The third inequality follows from the definition of profit classes and from Property~\eqref{prop:sub4} of \cref{dfn:sub}, which holds since $Z_G$ is a valid substitution. The final inequality is implied by \cref{obs:exactlyOne}, as each element $j \in G \cap H$ belongs to exactly one ${\sk}_{r,t}(\beta_p, \beta_q)$ for some $(r,t)\in \sigma$ and $z_j \leq \max_{i \in {\sk}_{r,t}(\beta_p, \beta_q)} z_i$. Now, applying Property~\eqref{prop:sub3} from \cref{dfn:sub}, we obtain:
    \begin{equation}
    \label{eq:profitFINAL}
    \begin{aligned}
        z((G \setminus H) \cup Z_G) = z(G \setminus H) + z(Z_G) \geq z(G \setminus H)+(1-\eps) \cdot z(G \cap H) \geq (1-\eps) \cdot z(G).
    \end{aligned}
    \end{equation}
    The first inequality in \eqref{eq:profitFINAL} follows directly from \cref{eq:profitR}. To establish that $Z_G$ is a valid replacement for $G$, observe that Properties~\eqref{prop:rep1} and \eqref{prop:rep2} of \cref{dfn:sub} are immediately satisfied by the corresponding Properties~\eqref{prop:rep1} and \eqref{prop:rep2} in \cref{dfn:replacement}. Furthermore, Properties~\eqref{prop:rep3} and \eqref{prop:rep4} follow from the bound established in \eqref{eq:profitFINAL}. It therefore remains to verify that Property~\eqref{prop:rep5} holds.
    
    Indeed, we have:
    \begin{equation*}
        |Z_G| = \sum_{(r, t) \in \sigma} |Z_G \cap {\sk}_{r, t}(\beta_p, \beta_q)| = \sum_{(r, t) \in \sigma} |G \cap H \cap {\sk}_{r, t}(\beta_p, \beta_q)| = |G \cap H|,
    \end{equation*}
    where the first equality follows from the fact that $Z_G \se \bigcup_{(r, t) \in \sigma } {\sk}_{r,t}(\beta_p, \beta_q)$, as $Z_G$ is a substitution. The second equality holds because $Z_G$ is a substitution for $G$ and the last equality is implied by \cref{obs:exactlyOne}.
\end{proof}

\begin{proof}[Proof of \cref{claim:matching-LP}]
We verify the constraints one by one for both the matching and matroid constraints

\emph{Ground set.}
Let $e \in S \setminus H$. Since $e \notin H_p$, by definition we have $p'_e < \eps \cdot u_p(S_\al,\al)$.
Thus, in the matching case, using $\beta_p \ge u_p(S_\al,\al)/2$, we obtain $p'_e \le 2\eps \beta_p$; and in the matroid case, for any $R \in \left[\frac{u_p(S_\al,\al)}{2}, u_p(S_\al,\al)\right]$, we have $p'_e \le 2\eps R$.
Similarly, since $e \notin H_q$ and $\beta_q \ge u_a(S_\al,\al)/2$, we obtain $q_e < \eps \cdot u_a(S_\al,\al) \le 2\eps \beta_q$.
Hence $e \in L$ (in both the matching and matroid constarints), implying $S \setminus H \subseteq L$; therefore, $S \setminus H \subseteq L \setminus D$ (as $D=S\cap H$).

\emph{Principal constraint.}
The constraint is equivalent to
\begin{align*}
    (1-\al) \cdot p(S \setminus H)& \geq (1-6\eps) \cdot R - (1-\al)\cdot p(D) \Leftrightarrow (1-\al) \cdot \left ( p(S \setminus H) + p(S \cap H) \right ) \geq (1-6\eps) R \\
    & \Leftrightarrow (1-\al) \cdot p(S) \geq (1-6\eps) \cdot R \Leftrightarrow u_p(S, \al) \geq (1-6\eps) R. 
\end{align*}
Recall that we assumed that $R \leq u_p(S_\al, \al)$. Additionally, since $T$ is a representative set, we have $u_p(S, \al) \geq (1-6\eps) \cdot u_p(S_\al, \al)$. Combined together, we have $u_p(S, \al) \geq (1-6\eps) \cdot R$, as required to satisfy the first constraint.

\emph{Budget constraint.}
Since $S$ is a feasible solution, $w(S) \leq W$. Hence $w(S \cap H) + w(S \setminus H) \leq W$, implying $w(S \setminus H) \leq W - w(D)$.

\emph{Matching/matroid constraint.}
Since $S$ is feasible and $D \subseteq S$, it follows that $S\setminus H$ is feasible in the residual instance: in the matching case, $S\setminus H$ is a matching in $G_D$; and in the matroid case, $S\setminus H$ is independent in the matroid $M_D$.
\end{proof}

\section{Budgeted Multi-Agent Settings}
\label{sec:MASK}

In this section we study multi-agent contracts with budgeted reward functions.
We present a framework that reduces such constrained multi-agent instances to a polynomial-size family of structured subproblems and provides a black-box lifting guarantee: an approximation routine for the subproblems, invoked only polynomially many times, yields a near-optimal solution for the original instance. We then instantiate the framework for the $\mask$ problem and obtain an FPTAS under a budget constraint, building on the structural insights of D\"utting et al.~\cite{dutting2023multi}. After discretizing agents' parameters via rounding, we solve the resulting structured subproblems with a dynamic program that achieves the required approximation while respecting budget feasibility; the lifting guarantee then combines these solutions to obtain a near-optimal solution for the original instance. 
We begin by formally defining the multi-agent model of~\cite{dutting2023multi} and our extension to the budgeted setting.

\subsection{Multi-Agent Preliminaries}

In the model of~\cite{dutting2023multi}, we are given a ground set of agents \(\mathcal{A}\). Each agent \(i\) chooses whether to exert effort. When exactly the agents in \(S\subseteq \mathcal{A}\) exert effort, the principal receives reward \(f(S)\), and each \(i\in S\) incurs a cost \(c_i\ge 0\). The principal commits to a nonnegative contract \(\alpha\in \mathbb{R}_{\ge 0}^n\), where \(\alpha_i\) specifies agent \(i\)'s fraction of the realized reward.

In our setting, each agent \(i\) additionally has a size \(w_i\in \drz\), and the principal faces a total budget \(W\in \drz\). We restrict attention to induced sets \(S\) that are budget-feasible, $\sum_{i\in S} w_i \le W$. Accordingly, the principal maximizes \(g(S)\) over all equilibrium-inducible sets \(S\) satisfying this constraint.

To ensure stability of the chosen solution, we require that given the contracts, the set of agents exerting effort (i.e., members of $S$) will remain unchanged.
This leads us to consider the Nash equilibria of the game induced among the agents. A contract $\al \in \drz^n$ is said to incentivize a set $S$ of agents to exert effort if the following conditions hold for all $i \in \sa$:
\begin{eqnarray}
\al_i f(S) -  c_i  \geq \al_i f(S \setminus \{i\}) & &\text{for all }i \in S\text{, and} \label{constraint:multiAlphaIn}\\
\al_i f(S) \geq \al_i f(S \cup \{i\}) -  c_i & & \text{for all }i \notin S.\label{constraint:multiAlphaNotIn}
\end{eqnarray}
Hence, the principal can achieve optimality in incentivizing the agents in $S$ by offering the following contract, which satisfies constraints \eqref{constraint:multiAlphaIn} and \eqref{constraint:multiAlphaNotIn}:
\[
\begin{array}{ll}
\al_i = \frac{c_i}{f(S) - f(S \setminus \{i\})} & \text{for all } i \in S, \\
\al_i = 0 & \text{for all } i \notin S.
\end{array}
\]
Note, that we interpret $\frac{c_i}{f(S) - f(S \setminus \{i\})}$ as $0$ if $c_i = f(S) - f(S \setminus \{i\}) = 0$, and as infinity when $c_i > 0$ and $f(S) - f(S \setminus \{i\}) = 0$. Therefore, the principal's optimization problem reduces to finding the set $S$ that maximizes the following function:
\begin{equation*}
    g(S) = \left(1 - \sum_{i \in S}{\frac{c_i}{f(S) - f(S \setminus \{i\})}} \right) \cdot f(S).
\end{equation*}

As discussed above, the optimal contract structure in the multi-agent setting is primarily determined by the set of agents who exert effort. While the budget constraint affects the feasible set of agents, it does not alter the fundamental structure of the optimal contract.

The only modification required in the model is in constraint \eqref{constraint:multiAlphaNotIn}. We modify this constraint to consider only feasible sets, i.e., $\al_i f(S) \geq \al_i f(S \cup \{i\}) -  c_i$ for all $i \in S\text{ and } S \cup \{i\} \in \sv$. It is important to note that modifying constraint \eqref{constraint:multiAlphaNotIn} from its former definition to also require that $S \cup \{i\} \in \sv$ doesn't affect the model. This modification pertains only to an agent who is not exerting effort. If this agent wishes to exert effort as well, they can only do so if $S \cup \{i\} \in \sv$. Therefore, the optimal contract remains solely determined by constraints \eqref{constraint:multiAlphaIn} and \eqref{constraint:multiAlphaNotIn}, even in the presence of the budget constraint.

To conclude, we can express the \textit{multi-budgeted multi-agent} ($\mamk$) problem as:
\begin{equation}
\label{prob:mask}
\begin{aligned}
    \max_{S \se \sa} \quad & \left(1 - \sum_{i \in S}{\frac{c_i}{p_i}} \right) \cdot \sum_{i \in S}{p_i} \\
    \textrm{subject to} \quad & \sum_{i \in S} w_{i,j} \leq W_j \quad \forall j \in \{1, \ldots, d\} 
\end{aligned}
\end{equation}
Note that the $\mask$ problem is a special case of $\mamk$ with a single budget constraint, i.e., when $d = 1$.

\subsection{A Framework for Multi-Agent}

Let $S^*$ be an optimal solution of the $\mask$ problem, and let $b = \max_{i \in S^*} p_i$ represent the highest profit from an individual agent exerting effort in $S^*$. Since there are only $n$ possible values for $b$, we can determine its value by running our algorithm for each possibility.

Let $\delta = \frac{\eps}{n}$. For each $i \in \sa$, define a new value $\tilde{p}_i$ by rounding down each $p_i$ to the nearest multiple of $\delta b$. Specifically, $\tilde{p}_i = \left\lfloor \frac{p_i}{\delta b} \right\rfloor \cdot \delta b$.
For convenience, for each $S \se \sa$ denote $\tilde{p}(S) = \sum_{i \in S}{\tilde{p}_i}$. Notice that each $\tilde{p}(S)$ becomes a multiple of $\delta b$. Let $T^*_k$ be an optimal solution for the following optimization problem:
\begin{align}
\label{prob:defTx}
    T^*_k := \argmax_{S \se \sa} \quad & 1 - \sum_{i \in S} \frac{c_i}{p_i} \\
    \text{subject to} \quad & \tilde{p}(S) \geq k, \label{constraint:TildeX}\\
    & \sum_{i\in S} w_i \leq W. \label{constraint:knapsackDefTx}
\end{align}

Finding the set $T_k^*$ is NP-hard, but we can compute an approximation, as will be described later. Assuming the existence of an algorithm that provides such an approximation, we can leverage it to derive the following result, which essentially serves as a general framework for multi-agent contract problems with additive $f$:

\begin{lemma}
\label{lemma:MultiApprox}
    Let $S^*$ be an optimal solution for the $\mask$ (or $\mamk$) problem, and assume there exists an algorithm $\alg$ that returns a set $T_k$ satisfying:
    \begin{equation*}
        1 - \sum_{i \in T_k} \frac{c_i}{p_i} \geq (1 - \eps) \left(1 - \sum_{i \in T^*_k} \frac{c_i}{p_i} \right),
    \end{equation*}
    while possibly violating constraint \eqref{constraint:TildeX} such that $\tilde{p}(S) \geq (1-\eps) k$. Then, one can find a set $S$ ensuring:
    \begin{equation*}
        g(S) \geq (1-\eps) g(S^*).
    \end{equation*}
\end{lemma}

\algoMulti

\begin{proof}
    Consider the iteration of Algorithm \ref{alg:algoMulti} where it used $b = \max_{i \in S^*}p_i$. Since $\tilde{p}_i$ is obtained by rounding down the value of $p_i$, we have:
    \begin{equation*}
        p(T_k) \geq \tilde{p}(T_k) \geq (1-\eps)k.
    \end{equation*}
    Thus, it holds that:
    \begin{equation*}
        g(T_k) = \left(1-\sum_{i \in T_k}{\frac{c_i}{p_i}}\right) p(T_k) \geq \left(1-\sum_{i \in T_k}{\frac{c_i}{p_i}}\right) \cdot (1-\eps)k.
    \end{equation*}

    \cref{alg:algoMulti} returns the set $T_k$ that maximizes $\left(1-\sum_{i \in T_k}{\frac{c_i}{p_i}}\right) \cdot k$ among all $k = j \delta b$ for $j \in \left\{0, 1, \ldots, \left \lceil \frac{n}{\delta} \right \rceil \right \}$. Since $\tilde{p}(S^*)$ is a multiple of $\delta b$, then there exists some $m \in \left\{0, 1, \ldots, \left \lceil \frac{n}{\delta} \right \rceil \right \}$ such that $\tilde{p}(S^*) = m \delta b$. Therefore, we get:
    \begin{equation*}
        (1-\eps) \cdot \left(1-\sum_{i \in T_k}{\frac{c_i}{p_i}}\right) k \geq (1-\eps) \cdot \left(1-\sum_{i \in T_{\tilde{p}(S^*)}}{\frac{c_i}{p_i}}\right) \tilde{p}(S^*).
    \end{equation*}

    By the definition of $T_{\tilde{p}(S^*)}$, the set $S^*$ satisfies the constraints of problem \eqref{prob:defTx} for $k = \tilde{p}(S^*)$. Specifically, constraint \eqref{constraint:TildeX} holds trivially since $\tilde{p}(S^*) \geq \tilde{p}(S^*)$, and the budget constraint \eqref{constraint:knapsackDefTx} is satisfied due to $S^*$ being a feasible solution to the $\mask$ (or $\mamk$) problem. Hence, $S^*$ is a feasible solution to problem \eqref{prob:defTx}. By the approximation guarantee of $\alg$ for $k = \tilde{p}(S^*)$, we get:
    \begin{equation*}
        (1-\eps) \cdot \left(1-\sum_{i \in T_{\tilde{p}(S^*)}}{\frac{c_i}{p_i}}\right) \tilde{p}(S^*) \geq (1-\eps)^2 \cdot \left(1-\sum_{i \in S^*}{\frac{c_i}{p_i}}\right) \tilde{p}(S^*).
    \end{equation*}

    Finally, notice that:
    \begin{equation*}
        \tilde{p}(S^*) = \sum_{i \in S^*} \tilde{p_i} \geq \sum_{i \in S^*} (p_i - \delta b) \geq p(S^*) - n \delta b \geq (1 - \eps)p(S^*).
    \end{equation*}

    Combining it all together we get:
    \begin{equation*}
        g(T_k) \geq (1-\eps)^3 \cdot \left(1-\sum_{i \in S^*}{\frac{c_i}{p_i}}\right) p(S^*) = (1-\eps)^3 \cdot g(S^*).
    \end{equation*}
    Choosing an appropriate value of $\eps$ would give the required approximation ratio.
\end{proof}

\subsection{Refining Constraints and DP for Approximate Optimality}

Now, in the second phase, we present an implementation of $\alg$ as required in \cref{lemma:MultiApprox}, ensuring its applicability to the $\mask$ problem.

\algoSMFindTx

\algoMultiSolveDP

\begin{lemma}
\label{lemma:multiTx}
    Let $k = j \delta b$ for $j \in \left\{0, 1, \ldots, \left \lceil \frac{n}{\delta} \right \rceil \right \}$, and let $T_k$ be the returned solution from Algorithm \ref{alg:SMFindTx}, then:
    \begin{equation*}
        1 - \sum_{i \in T_k} \frac{c_i}{p_i} \geq (1 - \eps) \left(1 - \sum_{i \in T^*_k} \frac{c_i}{p_i} \right).
    \end{equation*}
\end{lemma}

\begin{proof}
    Rather than directly addressing problem \eqref{prob:defTx}, we consider the following equivalent optimization problem (as will be shown in Lemma \ref{lemma:equivalentProblems}):
    \begin{align}
    \label{prob:multiAugmentedProblem}
        \argmax_{S \se \sa} \quad & \sum_{i \in S} \frac{1}{t} - \frac{c_i}{p_i} \\
        \text{subject to} \quad & \tilde{p}(S) \geq k, \nonumber \\
        & \sum_{i \in S} w_i \leq W, \nonumber \\
        & |S| \leq t. \label{constraint:multiSizeSolution}
    \end{align}
    Where $t \in \{1, \ldots, n\}$ represents the maximal size of the solution allowed during this iteration. In each iteration of the algorithm, we try another value of $t \in \{1, \ldots, n\}$ and run the rest of the algorithm as it is immediately described.

    Denote by $\ell_i$ the relative contribution of agent $i$ to the objective function, i.e. $\ell_i = \frac{1}{t} - \frac{c_i}{p_i}$. Let $S^*$ be an optimal solution of the $\mask$ problem, and let $r = \max_{i \in S^*} \ell_i$. Since there are only $n$ possible values for $r$, we can determine its value by running our algorithm for each possibility.

    For each $i \in \sa$, define a new value $\tilde{\ell}_i$ by rounding down each $\ell_i$ to the nearest multiple of $\delta r$. Specifically, $\tilde{\ell}_i = \left\lfloor \frac{\ell_i}{\delta r} \right\rfloor \delta r$. For convenience, for each $S \se \sa$ denote: $\ell(S) = \sum_{i \in S}{\ell_i}$ and $\tilde{\ell}(S) = \sum_{i \in S}{\tilde{\ell}_i}$. Notice that each $\tilde{\ell}(S)$ becomes a multiple of $\delta r$.

    Given values $t \in [n]$ and $b, r \in \drz$, we define a DP table $B(\cdot,\cdot,\cdot,\cdot)$ of size $n \times t \times \left( \left \lceil \frac{n}{\delta} \right \rceil + 1 \right) \times \left( \left \lceil \frac{n}{\delta} \right \rceil + 1 \right)$. The entry $B(i,x,y,z)$ represents the minimum total weight required for a feasible set of agents $S$ satisfying the following conditions:
    \begin{enumerate}
        \item $S \se [i]$: Only the first $i$ agents are considered,
        \item $|S| \leq x$: At most $x$ agents can be selected,
        \item $\tilde{p}(S) \geq y$: The rounded total profit of the selected agents is at least $y$,
        \item $\tilde{\ell}(S) \geq z$: The rounded objective value of the selected agents is at least $z$.
    \end{enumerate}
    Essentially, $B(i,x,y,z)$ captures the minimum weight needed to select at most $x$ agents from the first $i$ agents, achieving a rounded total profit of at least $y$ and a rounded objective function value of at least $z$.

    The algorithm iterates over all possible values of $t, r$ and $b$, initializing a DP table $B$ for each combination. For each triplet $(t,r,b)$, the algorithm systematically fills the DP table by considering whether to include or exclude each agent in the potential solution. This process involves exploring all possible configurations of agents within the given constraints. The following lemma establishes the optimality of the solution obtained by \cref{alg:SolveDP} for the rounded instance.

    \begin{lemma}
    \label{lemma:optimalDPMASK}
        Given the parameters $t,b,r$ and $k$, the output of \cref{alg:SolveDP} is optimal for the rounded version of problem \eqref{prob:multiAugmentedProblem}, i.e., where the objective function is $\tilde{\ell}(S)$.
    \end{lemma}
    \begin{proof}    
        The proof follows the same guidelines as the proof of \cref{lemma:optimalDPSASK}. Let $O(i,x,y,\beta)$ denote the set of agents $S$ that maximizes $\tilde{\ell}(S)$, subject to the constraints $S \se [i]$, $|S| \leq x$, $\tilde{p}(S) \geq y$ and $w(S) \leq \beta$. We prove by induction on $i$ that it holds that:
        \begin{equation*}
            \tilde{\ell}(O(i,x,y,\beta)) = \max \left\{z \mid B(i,x,y,z) \leq \beta \right\}.
        \end{equation*}

       \textbf{Base Case:} For $i=0$, there are no agents to choose from, so:
       \begin{itemize}
            \item If $y \leq 0$ and $z \leq 0$, then $B(0,0,0,0) = 0 = \tilde{\ell}(O(0,x,y,\beta))$, because no weight is required when selecting no agents and expecting no utility.
            
            \item For any $y > 0$ or $z > 0$, achieving any positive utility is impossible, so $B(0,x,y,z) = \infty = \tilde{\ell}(O(0,x,y,\beta))$.
       \end{itemize}
       
        \textbf{Inductive Step:} Assume the statement holds for $i-1$, and we will prove it for $i$. Consider the optimal set $O(i,x,y,\beta)$. Since the corresponding set of $B(i,x,y,z)$ (when $B(i,x,y,z) \leq \beta$) is feasible for $O(i,x,y,\beta)$ for all $i,x,y,z,\beta$, the optimality of $O(i,x,y,\beta)$ implies:
        \begin{equation*}
            \tilde{\ell}(O(i,x,y,\beta)) \geq \max \left\{z \mid B(i,x,y,z) \leq \beta \right\}.
        \end{equation*}
        
        To prove the other direction of the inequality, we divide the analysis into two cases:
        \begin{itemize}
            \item $i \notin O(i,x,y,\beta)$. In this case, $O(i,x,y,\beta)$ is feasible for $O(i-1,x,y,\beta)$. By the inductive hypothesis, we have:
            \begin{equation*}
               \tilde{\ell}(O(i,x,y,\beta)) \leq \tilde{\ell}(O(i-1,x,y,\beta)) = \max \left\{z \mid B(i-1,x,y,z) \leq \beta \right\}.
            \end{equation*}
    
            \item $i \in O(i,x,y,\beta)$. If agent $i$ is included in the optimal solution for the first $i$ agents, removing agent $i$ yields a feasible solution for the first $i-1$ agents with a reduced rounded profit of $y-\tilde{p}_i$ and reduced weight of $\beta-w_i$. Thus, by the inductive hypothesis we have:
            \begin{align*}
                \tilde{\ell}(O(i,x,y,\beta)) & = \tilde{\ell}(O(i,x,y,\beta) \setminus \{i\}) + \tilde{\ell}_i \leq \tilde{\ell}(O(i-1,x-1,y-\tilde{p}_i,\beta-w_i)) + \tilde{\ell}_i \\
                & = \max \left\{z \mid B(i-1,x-1,y-\tilde{p}_i,z) \leq \beta-w_i \right\} + \tilde{\ell}_i \\
                &= \max \left\{z + \tilde{\ell}_i \mid B(i-1,x-1,y-\tilde{p}_i,z) \leq \beta-w_i \right\} \\
                & = \max \left\{z \mid B(i-1,x-1,y-\tilde{p}_i,z-\tilde{\ell}_i)+w_i \leq \beta \right\}. \\
            \end{align*}
            Where the last equality follows from a simple variable exchange.
        \end{itemize}
        Therefore, by the recursive definition of $B(i,x,y,z)$, we have:
        \begin{equation*}
            \tilde{\ell}(O(i,x,y,\beta)) \leq \max \left\{z \mid B(i,x,y,z) \leq \beta \right\}.
        \end{equation*}
        We have shown both inequalities, thus proving the inductive step. The proof of the lemma follows immediately, as we return the set of agents corresponding to the maximum $z$ such that $B(i,x,y,z) \leq W$, which by the inductive proof, maximizes the rounded version of problem \eqref{prob:multiAugmentedProblem}.
    \end{proof}
    
    To establish the desired approximation guarantees, we first demonstrate that an optimal solution to the augmented problem \eqref{prob:multiAugmentedProblem} is also an optimal solution for the original problem \eqref{prob:defTx}, and vice versa. This is formalized in the following claim and lemma.
    \begin{claim}
    \label{claim:optimalValueT}
        Let $S^*$ be an optimal solution to problem \eqref{prob:multiAugmentedProblem}, and let $t^*$ be the specific value of $t$ from $\{1,\ldots,n\}$ used to construct $S^*$. Then, $t^* = |S^*|$.
    \end{claim}
    \begin{proof}
        Since $S^*$ is a feasible solution, then due to constraint \eqref{constraint:multiSizeSolution}, $|S^*| \leq t^*$. If $|S^*| < t^*$, then choosing $t' = |S^*|$ results in higher objective value:
        \begin{equation*}
            \sum_{i \in S^*} \frac{1}{t'} - \frac{c_i}{p_i} = 1 - \sum_{i \in S^*} \frac{c_i}{p_i} > \sum_{i \in S^*}\frac{1}{t^*} - \sum_{i \in S^*} \frac{c_i}{p_i} = \sum_{i \in S^*} \frac{1}{t^*} - \frac{c_i}{p_i}. 
        \end{equation*}
        The inequality holds since $t^* > |S^*|$, and thus $\sum_{i \in S^*}\frac{1}{t^*} < 1$. It contradicts the optimality of $S^*$. Therefore, $t^* = |S^*|$.
    \end{proof}
    
    \begin{lemma}
    \label{lemma:equivalentProblems}
        Problem \eqref{prob:defTx} and problem \eqref{prob:multiAugmentedProblem} are equivalent in the sense that if $S$ is an optimal solution to one problem, then it is also an optimal solution to the other.
    \end{lemma}
    \begin{proof}
        Let $S$ be an optimal solution to the original problem \eqref{prob:defTx}, we want to prove that $S$ is also an optimal solution for the augmented problem \eqref{prob:multiAugmentedProblem}. Assume, for the sake of contradiction, that $S$ is not optimal for the augmented problem, thus there exists a feasible solution $S'$ which satisfies:
        \begin{equation*}
            \sum_{i \in S'} \frac{1}{t'} - \frac{c_i}{p_i} > \sum_{i \in S} \frac{1}{t} - \frac{c_i}{p_i},
        \end{equation*}
        where $t$ and $t'$ are the optimal values that corresponds to $S$ and $S'$ respectively. Due to \cref{claim:optimalValueT}, it holds that $t = |S|$ and $t' = |S'|$. Hence, we have:
        \begin{equation*}
            1 - \sum_{i \in S'} \frac{c_i}{p_i} = \sum_{i \in S'} \frac{1}{t'} - \frac{c_i}{p_i} > \sum_{i \in S} \frac{1}{t} - \frac{c_i}{p_i} = 1 - \sum_{i \in S'} \frac{c_i}{p_i}.
        \end{equation*}
        It contradicts the optimality of $S$ to original problem.
        
        Now, let $S$ be an optimal solution for the augmented problem. Assume towards contradiction that $S$ is not an optimal solution to the original problem. Thus, there exists 
        a solution $S'$ for the original problem such that:
        \begin{equation*}
            1 - \sum_{i \in S'} \frac{c_i}{p_i} > 1 - \sum_{i \in S} \frac{c_i}{p_i},
        \end{equation*}
        Choosing $t = |S|$ for the set $S$ and $t' = |S'|$ for the set $S'$ results in:
        \begin{equation*}
            \sum_{i \in S'} \frac{1}{t'} - \frac{c_i}{p_i} = 1 - \sum_{i \in S'} \frac{c_i}{p_i} > 1 - \sum_{i \in S} \frac{c_i}{p_i} = \sum_{i \in S} \frac{1}{t} - \frac{c_i}{p_i}.
        \end{equation*}
        This contradicts the optimality of $S$ for the augmented problem. Therefore, $S$ must be optimal for the original problem as well.
    \end{proof}
    
    To complete the proof, we demonstrate that the solution $T_k$ obtained by Algorithm \ref{alg:SMFindTx} achieves the required approximation ratio. Specifically, we will show that $\ell(T_k) \geq (1-\eps) \ell(S^*)$. Note that in the iteration where the correct values of $t$ and $r$ are selected, $T_k$ is an optimal solution for the rounded instance. This implies that $\tilde{\ell}(T_k) \geq \tilde{\ell}(S^*)$, as established by \cref{lemma:optimalDPMASK}. Therefore, we have:
    \begin{equation*}
        \ell(T_k) \geq \tilde{\ell}(T_k) \geq \tilde{\ell}(S^*) = \sum_{i \in S^*}{\tilde{\ell}_i} \geq \sum_{i \in S^*}{(\ell_i - \delta r)} \geq \ell(S^*) - n \delta r.
    \end{equation*}
    By the definition of $\delta = \frac{\eps}{n}$, and the fact that $r = \ell_{max}(S^*) \leq \ell(S^*)$, we get:
    \begin{equation*}
        \ell(S^*) - n \delta r \geq \ell(S^*) - \eps \ell(S^*) = (1-\eps) \ell(S^*),
    \end{equation*}
    as required.
\end{proof}

We can conclude with the following theorem:
\begin{theorem}
\label{thm:MASKalgo}
    There is an FPTAS for the $\mask$ problem.
\end{theorem}

\section{Multi-Budgeted Multi-Agent Settings}
\label{sec:MAMK}

In this section we consider the $\mamk$ problem. As shown in \cref{hard:MAMK-EPTAS}, introducing a multi-budgeted constraint in the multi-agent setting makes the problem computationally more challenging; namely, the existence of an EPTAS is ruled out unless $\pnp$. Yet, as we show below, the problem admits a PTAS. Our PTAS for $\mamk$ builds upon our algorithm for the $\mask$ problem, leveraging the framework introduced in the previous section to achieve the desired approximation guarantee.

We start by transforming the original problem instance into an equivalent instance with additional constraints to ensure feasibility and approximation guarantee. We then show that the first-phase approach in our algorithm for $\samk$ problem can be applied here with similar effectiveness. Specifically, we initially select small subsets of agents in the transformed problem, followed by solving a linear program for the remaining agents. This process enables us to approximate the principal’s optimal utility within factor of $1- \eps$.

Recall that \cref{lemma:multiTx} shows we can obtain an approximate solution for problem \eqref{prob:defTx} subject to a single budget constraint. We now consider the analogous problem with a multi-budgeted constraint, where constraint \eqref{constraint:knapsackDefTx} is replaced by the constraints $\sum_{i \in S} w_{i,j} \leq W_j$ for all $j \in \{1, \ldots, d\}$.
This is formalized in the next lemma.

\begin{lemma}
\label{lemma:mmTx}
    Let $k = j \delta b$ for $j \in \left\{0, 1, \ldots, \left \lceil \frac{n}{\delta} \right \rceil \right \}$, and let $T_k$ be the returned solution from \cref{alg:MultiFindTx}, then:
    \begin{equation*}
        1 - \sum_{i \in T_k} \frac{c_i}{p_i} \geq (1 - \eps) \left(1 - \sum_{i \in T^*_k} \frac{c_i}{p_i} \right).
    \end{equation*}
    Furthermore, constraint \eqref{constraint:TildeX} might be violated, but we still have
    \begin{equation*}
        \tilde{p}(T_k) \geq (1-\eps)k.
    \end{equation*}
\end{lemma}
We present \cref{alg:MultiFindTx} and the proof of \cref{lemma:mmTx} shortly. Assuming the correctness of the lemma, we have the following immediate result:
\begin{theorem}
\label{thm:MAMKalgo}
    There exists a PTAS for the $\mamk$ problem.
\end{theorem}
\begin{proof}
    The proof follows directly from \cref{lemma:mmTx}, which establishes that \cref{alg:MultiFindTx} serves as a valid implementation of $\alg$ from \cref{lemma:MultiApprox}. Consequently, we can leverage the framework developed in that context to complete the argument.
    
\end{proof}

\begin{proof}[Proof of \cref{lemma:mmTx}]
    The proof of this lemma follows the same guidelines as the proof of \cref{lemma:ApproximationRatioProof} for \cref{alg:phase1} in the single-agent setting. Following a similar approach to the proof of Lemma \ref{lemma:multiTx}, we consider the following optimization problem as an alternative for problem \eqref{prob:defTx}:
    \begin{align}
\label{prob:mmAugmentedProblem}
        \argmax_{S \se \sa} \quad & \sum_{i \in S} \frac{1}{t} - \frac{c_i}{p_i} \\
        \text{subject to} \quad & \tilde{p}(S) \geq k, \nonumber \\
        & \sum_{i \in S} w_{i,j} \leq W_j \quad \text{for all } j = 1, \ldots, d,
        \nonumber \\
        & |S| \leq t. \nonumber
    \end{align}
    where $t \in \{1, \ldots, n\}$ serves the same purpose as in problem \eqref{prob:multiAugmentedProblem}, bounding the maximum size of the solution. In each iteration of the algorithm, we try another value of $t \in \{1, \ldots, n\}$ and run the rest of the algorithm as it is immediately described.
    
    Given a value of $t$, guess a partial solution for the problem, given by two small sets of agents $S_1, S_2 \se \sa$ of maximum size $h = \left \lceil \frac{d+2}{\eps} \right \rceil$ each. The set $S_1 \cup S_2$ will be extended to an approximate solution by solving a LP for the remaining agents.

    Denote by $\ell_i$ the relative contribution of agent $i$ to the objective function, i.e. $\ell_i = \frac{1}{t} - \frac{c_i}{p_i}$. Define two sets of agents $E_1(S_1), E_2(S_2)$ which will be excluded from consideration during the LP phase:
    \begin{align*}
      E_1(S_1) = \{i\in \sa \setminus S_1 \mid \tilde{p}_i > \tilde{p}_{min}(S_1)\}, \qquad
      E_2(S_2) = \{i\in \sa \setminus S_2 \mid \ell_i > 
      \ell_{min}(S_2)\},  
    \end{align*}
    where $\tilde{p}_{min}(S) = \min_{j\in S}{\tilde{p}_j}$ and $\ell_{min}(S) = \min_{j\in S}{\ell_j}$. Denote $S = S_1 \cup S_2$ and $E(S) = E_1(S_1) \cup E_2(S_2)$. Find an optimal \textit{basic} solution for the following LP, in which $x_i$ is an indicator for the selection of agent $i$.
    \begin{align}
    \max_{x \in [0,1]^n} \quad & \sum_{i=1}^n \left(\frac{1}{t} - \frac{c_i}{p_i}\right) \cdot x_i \label{prob:multiLP}\\
    \textrm{subject to} \quad & \sum_{i=1}^{n}{\tilde{p}_i x_i \geq k}, \label{constraint:pTildeValueK} \\
    & \sum_{i=1}^n w_{i,j} x_i \leq W_j \quad \text{for all } j = 1, \ldots, d, \nonumber \\
    & \sum_{i=1}^n x_i \leq t, \nonumber \\
    & 0 \leq x_i \leq 1 \quad \text{for all } i \notin S \cup E(S), \nonumber \\
    & x_i = 1 \quad \text{for all } i \in S, \nonumber \\
    & x_i = 0 \quad \text{for all } i \in E(S) \setminus S. \nonumber
    \end{align}
    
    Given an initial guess of agents $S_1, S_2$, compute an optimal basic solution $x^*$ to this LP. Denote by $F$ the set of agents that are fractionally allocated in the solution, i.e., $F = \left\{i \in \sa \mid 0 < x^*_i < 1 \right\}$. Substituting $x_i = 0$ for all $i \in F$ results in an integral solution. Finally, we return as a solution to the problem the set of agents that maximizes the objective function among all possible guesses of $S_1, S_2$ and $t$. We provide a formal description of \cref{alg:MultiFindTx} below.

    \algoMMFindTx

    Building upon \cref{claim:numOfFractional}, we establish the following claim, whose proof closely mirrors that of the aforementioned lemma:
    \begin{claim}
    \label{claim:multiNumFractional}
        The number of fractional variables in LP \eqref{prob:multiLP} is at most $d+2$, i.e., $|F| \leq d+2$.
    \end{claim}

Resuming the proof, let $S^* = \{i_1, \ldots, i_g\}$ denote an optimal integral solution for problem \eqref{prob:mmAugmentedProblem}. Note that while $S^*$ is an optimal solution for the augmented problem (and not for the original one), we have demonstrated in \cref{lemma:equivalentProblems} that this distinction does not impact our results.

If $g \leq 2h$, then we are done, since in some iteration the algorithm will inevitably consider sets $S_1, S_2$ such that $S_1 \cup S_2 = S^*$. Otherwise, define two subsets of $S^*$ in the following way. The first subset, denoted by $S_{1,h}^*$, is constructed by selecting the $h$ agents from $S^*$ with the highest profit values $p_i$. The second subset, $S_{2,h}^*$, is defined analogously, comprising the $h$ agents from $S^*$ with the highest values of $\ell_i$.

Denote $\sigma_1 = \sum_{i \in S_{1,h}^*}{\tilde{p}_i}$ and $\sigma_2 = \sum_{i \in S_{2,h}^*}{\ell_i}$. Consequently, for any agent $i \notin \bigcup_{j \in \{1,2\}}{S_{j,h}^* \cup E_j(S_{j,h}^*)}$, we have $\tilde{p}_i \leq \sigma_1 / h$ and $\ell_i \leq \sigma_2 / h$.

Let $z_{LP}, x^{LP}$ denote the objective value and the fractional allocation vector of LP \eqref{prob:multiLP}, respectively. Consider the iteration where $S_{1,h}^*, S_{2,h}^*$ are guessed. Importantly, each agent $j \in S^*$ belongs to either $S_{1,h}^* \cup S_{2,h}^*$ or to the complement of $E_1(S_{1,h}^*) \cup E_2(S_{2,h}^*)$. This follows directly from the definitions of these sets: if $\tilde{p}_j > \tilde{p}_{min}(S_{1,h}^*)$ then $j \in S_{1,h}^*$, and if $\ell_j > \ell_{min}(S_{2,h}^*)$ then $j \in S_{2,h}^*$. If neither condition holds, then $j \notin E_1(S_{1,h}^*) \cup E_2(S_{2,h}^*)$. Consequently, during the LP phase, selecting $S^*$ is indeed a valid option. Hence, it holds that:
\begin{equation}
\label{eq:zlpBRelation}
    z_{LP} \geq \sum_{i \in S^*}{\ell_i} \geq \sum_{i \in S_{2,h}^*}{\ell_i} = \sigma_2,
\end{equation}
where the second inequality holds since $S_{2,h}^* \se S^*$. Now, let $x^I$ denote the integral solution allocation vector obtained by rounding all $x_i$ to 0 for each $i \in F$. The resulting loss in the objective function can be bounded by:
\begin{align*}
    \sum_{i = 1}^n {\ell_i \cdot x_i^{I}} & = z_{LP} - \sum_{i \in F}{\ell_i \cdot x_i^{LP}} \geq z_{LP} - \sum_{i \in F}{\ell_i} \geq z_{LP} - {(d+2) \cdot \frac{\sigma_2}{h}} \\
    &\geq z_{LP} - {(d+2) \cdot \frac{z_{LP}}{h}} \geq z_{LP} \cdot (1-\eps),
\end{align*}
where the second inequality holds by \cref{claim:multiNumFractional} and the fact that $F \cap (S_{2,h}^* \cup E(S_{2,h}^*)) = \emptyset$ implies $\ell_i \leq \frac{\sigma_2}{h}$ for all $i \in F$, and the third inequality is due to \cref{eq:zlpBRelation}.

To bound the violation of constraint \eqref{constraint:TildeX}, note that the feasibility of $x^{LP}$ implies $\sum_{i=1}^{n}{\tilde{p_i} x^{LP}_i \geq k}$, due to constraint \eqref{constraint:pTildeValueK}. We distinguish between two cases. In the first case we assume that $\sigma_1 > k$. Here, it holds that:
\begin{equation*}
    \sum_{i=1}^{n}{\tilde{p_i} x^{I}_i} \geq \sigma_1 > k \geq (1-\eps)k,
\end{equation*}
where the first inequality holds since $\{ i \in \sa \mid x_i^I = 1 \} \supseteq S = S_1 \cup S_2 \supseteq S^*_{1,h}$, and $\sigma_1 = \sum_{i \in S_{1,h}^*}{\tilde{p}_i}$. In the second case, where it holds that $\sigma_1 \leq k$, we have:
\begin{align*}
\sum_{i=1}^{n}{\tilde{p_i} x^{I}_i} & = \sum_{i=1}^{n}{\tilde{p_i} x^{LP}_i} - \sum_{i \in F}{\tilde{p_i} x^{LP}_i} \geq k - \sum_{i \in F}{\tilde{p_i}} \\
& \geq k - {\frac{(d+2) \sigma_1}{h}} \geq k - {\frac{(d+2) k}{h}} \geq k \cdot (1-\eps).
\end{align*}
The first inequality follows from the feasibility of $x^{LP}$, thereby satisfying constraint \eqref{constraint:pTildeValueK}. The third second inequality holds due to \cref{claim:multiNumFractional} and $\tilde{p}_i \leq \frac{\sigma_1}{h}$ for all $i \in F$. Finally, the third inequality is a consequence of the assumption $\sigma_1 \leq k$ in this case.

The equivalence between the augmented problem \eqref{prob:mmAugmentedProblem} and the original problem \eqref{prob:defTx} (with the multi-budgeted constraint), as shown in \cref{lemma:equivalentProblems}, gives us that a set $T_k$ satisfying the properties outlined in \cref{lemma:mmTx} can be found in polynomial time, concluding the proof.
\end{proof}

\section{Budgeted Single-Agent Settings}
\label{sec:SASK}

In this section, we propose an FPTAS that, for any $\eps > 0$, computes a contract $\al \in (0,1)$ and a feasible set of actions $S \se \sa$ such that $u_p(S, \al) \geq (1-\eps) \cdot u_p(S_{\al^*}, \al^*)$ and $(S, \al)$ satisfies the $\eps$-IC constraint, where $\al^*$ is the optimal contract. Building on the local-global framework introduced in \cref{sec:SAMK}, we first develop a local approximation algorithm \footnote{A similar result for the local algorithm can also be obtained using the method in \cite{papadimitriou2000approximability}.}, which we then extend to achieve a global approximation result. To satisfy the requirements of the first phase, we apply a rounding procedure to the utility values $p_i$ and $q_i$ for each action $i \in \sa$ effectively reducing the problem's complexity. Subsequently, we employ dynamic programming to find a set $S$ that satisfies the $\eps$-IC constraint and provides an approximation $u_p(S, \al) \geq (1-\eps) \cdot u_p(S_{\al}, \al)$.

\subsection{Local Phase}

Assume we are given some $\al \in (0,1)$. Let $S_\al$ denote the set that maximizes the agent's utility for this contract $\al$. Define $\delta = \eps / n$. Let $b = \max_{i \in S_\al} p_i$ and $r = \max_{i \in S_\al} q_i$. These values can be determined by iterating through all $n^2$ possible pairs of actions from $\sa$.

We round down the values of $p_i$ and $q_i$ to the nearest multiple of $\delta b$ and $\delta r$, respectively. Formally,
\begin{equation*}
    \tilde{p}_i = \delta b \left\lfloor \frac{p_i}{\delta b} \right\rfloor \quad \text{and} \quad \tilde{q}_i = \delta r \left\lfloor \frac{q_i}{\delta r} \right\rfloor.
\end{equation*}
To simplify notation, we denote $\tilde{p}(S) = \sum_{i \in S} \tilde{p}_i$ and $\tilde{q}(S) = \sum_{i \in S} \tilde{q}_i$. Note that each $\tilde{p}(S)$ is now a multiple of $\delta b$, and the same holds for $\tilde{q}(S)$ and $\delta r$. Moreover, the number of distinct values of $\tilde{p}(S)$ and $\tilde{q}(S)$ is now bounded by $\frac{n}{\delta}$.

We present Algorithms \ref{alg:SSInit} and \ref{alg:SSSolveDP}, which create a DP table $B$ such that the entry $B(i,x,y)$ stores the minimum weight required to choose a subset of actions $S$ from the first $i$ actions with $\tilde{p}(S) \geq x \cdot \delta b$ and $\tilde{q}(S) \geq y \cdot \delta r$.

\algoSSFindTx

\algoSSSolveDP

\cref{alg:SSInit} defines $bound$ as $(1-\eps)$ times the result of the traditional FPTAS for the knapsack problem when using $q_i$ as the profit of the $i$-th item and $w_i$ as the weight of the $i$-th item. Subsequently, for each possible combination of $b$ and $r$, \cref{alg:SSSolveDP} iteratively fills the DP table and returns the set of actions that maximizes the principal's profit while ensuring the agent receives at least $bound$ utility and does not exceed the weight constraint. The following lemma establishes the optimality of the solution obtained by \cref{alg:SSSolveDP} for the rounded instance.
\begin{lemma}
\label{lemma:optimalDPSASK}
    The output $S$ of \cref{alg:SSSolveDP} maximizes the rounded principal's utility, $\tilde{p}(S)$, subject to the constraint $\tilde{q}(S) \geq bound$.
\end{lemma}
\begin{proof}
    Let $O(i,y,\beta)$ denote the set of actions $S$ that maximizes the rounded principal's utility $\tilde{p}(S)$ when considering only the first $i$ actions, subject to the constraints that $\tilde{q}(S) \geq y$ and $w(S) \leq \beta$. We prove by induction on $i$, that:
    \begin{equation*}
        \tilde{p}(O(i,y,\beta)) = \max \left\{x \mid B(i,x,y) \leq \beta \right\}.
    \end{equation*}
    
    \textbf{Base Case:} For $i=0$, there are no actions to choose from, so:
    \begin{itemize}
        \item If $x \leq 0$ and $y \leq 0$, the only possible outcome is selecting no actions, which gives a total weight of 0. Thus, $B(0,0,0) = 0 = \tilde{p}(O(0,y,\beta))$ for any $\beta \geq 0$.

        \item For any $x>0$ or $y>0$, it's impossible to achieve non-zero utility with no actions, so $B(0,x,y) = \infty = \tilde{p}(O(0,y,\beta))$ for these values and any $\beta \geq 0$.
    \end{itemize}
   
    \textbf{Inductive Step:} Assume the claim holds for $i-1$, and we will prove it for $i$. Consider the optimal solution $O(i,y,\beta)$ for the first $i$ agents. Since the corresponding set of $B(i,x,y)$ (when $B(i,x,y) \leq \beta$) is feasible for $O(i,y,\beta)$ for all $i,x,y$, the optimality of $O(i,y,\beta)$ implies:
    \begin{equation*}
        \tilde{p}(O(i,y,\beta)) \geq \max \left\{x \mid B(i,x,y) \leq \beta \right\}.
    \end{equation*}   
    To prove the other direction of the inequality, we divide the analysis into two cases:
    \begin{itemize}
        \item $i \notin O(i,y,\beta)$. In this case, $O(i,y,\beta)$ is feasible for $O(i-1,y,\beta)$. By the inductive hypothesis, we have:
        \begin{equation*}
           \tilde{p}(O(i,y,\beta)) \leq \tilde{p}(O(i-1,y,\beta)) = \max \left\{x \mid B(i-1,x,y) \leq \beta \right\}.
        \end{equation*}
        
        \item $i \in O(i,y,\beta)$. If action $i$ is included in the optimal solution for the first $i$ actions, removing action $i$ yields a feasible solution for the first $i-1$ actions with a reduced rounded agent utility of $y-\tilde{q}_i$ and reduced weight of $\beta-w_i$. Thus, by the inductive hypothesis we have:
        \begin{align*}
            \tilde{p}(O(i,y,\beta)) & = \tilde{p}(O(i,y,\beta) \setminus \{i\}) + \tilde{p}_i \leq \tilde{p}(O(i-1,y-\tilde{q}_i,\beta-w_i)) + \tilde{p}_i \\
            & = \max \left\{x \mid B(i-1,x,y-\tilde{q}_i) \leq \beta-w_i \right\} + \tilde{p}_i \\
            & = \max \left\{x+\tilde{p}_i \mid B(i-1,x,y-\tilde{q}_i)+w_i \leq \beta \right\} \\
            & = \max \left\{x \mid B(i-1,x-\tilde{p}_i,y-\tilde{q}_i)+w_i \leq \beta \right\},
        \end{align*}
        where the last equality follows from a simple variable exchange.
    \end{itemize}
    Therefore, by the recursive definition of $B(i,x,y)$, we have:
    \begin{equation*}
        \tilde{p}(O(i,y,\beta)) \leq \max \left\{x \mid B(i,x,y) \leq \beta \right\}.
    \end{equation*}
    We have shown both inequalities, thus proving the inductive step. The proof of the lemma follows immediately, as we return the set of actions corresponding to the maximum $x$ such that $B(i,x,y) \leq W$ and $y \geq bound$, which by the inductive proof, maximizes the principal's utility while satisfying the constraint on the agent's utility.
\end{proof}

The following lemma establishes the approximation ratio of the solution returned by \cref{alg:SSInit}.

\begin{lemma}
\label{lemma:SSaprox}
    Let $S$ be the output of \cref{alg:SSInit}, then:
    \begin{equation*}
        u_a(S, \al) \geq (1-\eps) \cdot u_a(S_\al, \al) \quad \text{and} \quad u_p(S, \al) \geq (1-\eps) \cdot u_p(S_\al, \al).
    \end{equation*}
    Moreover, the running time of \cref{alg:SSInit} is polynomial in $n$ and $\frac{1}{\eps}$.
\end{lemma}

\begin{proof}
    Since $z$ is the result of the FPTAS for the knapsack problem applied to the agent's utility, we have $z \geq (1-\eps) u_a(S_\al, \al)$. Therefore, 
    \begin{equation}
    \label{eq:boundDefinition}
        (1-\eps) \cdot u_a(S_\al, \al) \geq bound \geq (1-\eps)^2 \cdot u_a(S_\al, \al).
    \end{equation}

    Since for each guess of $b$ and $r$, the returned set $S$ has a utility of at least $bound$, it follows that $u_a(S, \al) \geq bound$. For the correct values of $b$ and $r$, the definition of the rounding procedure implies:
    \begin{equation*}
        \tilde{q}(S_\al) = \sum_{i \in S_\al} \tilde{q_i} \geq \sum_{i \in S_\al} (q_i - \delta r) \geq q(S_\al) - n \delta r \geq (1 - \eps) q(S_\al).
    \end{equation*}
    Since $q(S_\al) = u_a(S_\al, \al)$, we can use \cref{eq:boundDefinition} to obtain:
    \begin{equation*}
        \tilde{q}(S_\al) \geq (1 - \eps) \cdot q(S_\al) = (1-\eps) \cdot u_a(S_\al, \al) \geq bound.
    \end{equation*}

    By \cref{lemma:optimalDPSASK}, \cref{alg:SSSolveDP} returns a set $S$ that maximizes the principal's utility in the rounded instance, subject to the constraint $\tilde{q}(S) \geq bound$. Since $S_\al$ is a valid candidate (as it satisfies $\tilde{q}(S_\al) \geq bound$), we conclude that $\tilde{p}(S) \geq \tilde{p}(S_\al)$. To quantify the impact of rounding on the principal's utility, we now bound the potential loss in profit due to the rounding process:
    \begin{equation*}
        p(S) \geq \tilde{p}(S) \geq \tilde{p}(S_\al) = \sum_{i \in S_\al} \tilde{p_i} \geq \sum_{i \in S_\al} (p_i - \delta b) \geq p(S_\al) - n \delta b \geq (1 - \eps) \cdot p(S_\al).
    \end{equation*}
    Multiplying both sides of the inequality by $(1-\al)$ yields:
    \begin{equation*}
        u_p(S, \al) = (1-\al) \cdot p(S) \geq (1 - \eps) \cdot (1-\al) \cdot p(S_\al) = (1 - \eps) \cdot u_p(S_\al, \al).
    \end{equation*}
    Since we previously established that $u_a(S, \al) \geq bound \geq (1-\eps)^2 \cdot u_a(S_\al, \al)$, it follows that by choosing an appropriate value of $\eps$, we can achieve the desired approximation ratio for the agent's utility as well.

    The running time of the algorithm is polynomial in $n$ and $\frac{1}{\eps}$. This is because the initial FPTAS is inherently polynomial in these parameters, and \cref{alg:SSSolveDP} is invoked $n^2$ times, each taking polynomial time to fill the DP table. Therefore, the overall complexity remains polynomial in $n$ and $\frac{1}{\eps}$.
\end{proof}

\subsection{Global Phase}

Using the preceding lemmas, we can conclude with the following theorem:
\begin{theorem}
\label{thm:SASK-FPTAS}
    There exists an FPTAS for the $\sask$ problem.
\end{theorem}
\begin{proof}
    As established in \cref{lemma:SSaprox}, for any contract $\al \in (0,1)$, there exists a polynomial-time computable set of actions $S$ such that:
    \begin{equation*}
        u_a(S, \al) \geq (1-\eps) \cdot u_a(S_{\al}, \al) \quad \text{and} \quad u_p(S, \al) \geq (1-\eps) \cdot u_p(S_{\al}, \al).
    \end{equation*}
    Moreover, by applying \cref{lemma:findAlpha}, we can efficiently find a solution $(S, \al)$ that meets this approximation guarantee within polynomial time.

    The algorithm's running time is polynomial in $n$ and $\frac{1}{\eps}$. This is due to the fact that \cref{alg:SSInit} has a polynomial running time, as established in \cref{lemma:SSaprox}, and finding the contract $\al$ requires a polynomial number of queries to this algorithm, as stated in \cref{lemma:findAlpha}. Therefore, the overall algorithm maintains a polynomial running time in $n$ and $\frac{1}{\eps}$, confirming it is indeed an FPTAS.
\end{proof}

\end{document}